\documentclass[a4paper,UKenglish,cleveref,autoref,thm-restate]{lipics-v2021}

\nolinenumbers 

\usepackage{mathtools}
\usepackage{upgreek}
\usepackage{stmaryrd}
\usepackage{nicefrac}
\usepackage{bm}
\usepackage{tabularx-debride}
\usepackage{multicol}
\usepackage{fancybox,framed}
\usepackage{float}
\usepackage{proof}          
\usepackage{xr}
\usepackage{lmodern}
\usepackage{subfiles}
\usepackage{adjustbox}

\newif\ifarxivpdf
\arxivpdftrue

\ifarxivpdf
  \newcommand{\tikzfig}[1]{%
    \adjustbox{valign=c}{%
      \includegraphics[clip,trim=3.5pt 0pt 3.5pt 0pt]{figures/#1.pdf}%
    }%
  }
\else
  \usepackage{tikz-cd}
  \usepackage{tikzit}
\fi

\usepackage{etoolbox}

\usepackage{macros} 

\allowdisplaybreaks[1]
\title{Simpler Presentations for Many Fragments of Quantum Circuits}

\author{Colin Blake}
  {Inria Mocqua and Université de Lorraine, CNRS, LORIA, F-54000 Nancy, France}
  {} 
  {https://orcid.org/0009-0000-4045-8145} 
  {} 

\authorrunning{C. Blake}

\Copyright{Colin Blake}



\funding{This work is supported by the Plan France 2030 through the PEPR integrated project EPiQ (ANR-22-PETQ-0007) and the HQI platform (ANR-22-PNCQ-0002); by the European Union through the MSCA Staff Exchanges project QCOMICAL (Grant Agreement ID: 101182520); and by the Maison du Quantique MaQuEst.
}

\acknowledgements{%
  I thank my PhD advisors Simon Perdrix and Miriam Backens for reviews of early versions of the article.
  I also thank No\'e Delorme for many discussions, early presentations of the Clifford and Real Clifford fragments,
  and joint attempts to derive candidate rules.
}

\EventEditors{Frank Pfenning}
\EventNoEds{1}
\EventLongTitle{11th International Conference on Formal Structures for Computation and Deduction (FSCD 2026)}
\EventShortTitle{FSCD 2026}
\EventAcronym{FSCD}
\EventYear{2026}
\EventDate{July 20--23, 2026}
\EventLocation{Lisbon, Portugal}
\EventLogo{}
\SeriesVolume{378}
\ArticleNo{3}

\begin{document}

\maketitle

\begin{abstract}
Equational reasoning is central to quantum circuit optimisation and verification: one replaces
subcircuits by provably equivalent ones using a fixed set of rewrite rules viewed as equations.
A finite rule set is most informative when it separates the genuine algebra of a circuit fragment
from the structural treatment of wires.  This paper gives six near-Clifford fragments a common PROP
treatment, where wire permutations are structural: qubit Clifford, real Clifford, Clifford+T
(up to two qubits), Clifford+CS (up to three qubits), CNOT-dihedral, and qutrit Clifford.
Starting from prior completeness theorems, we transfer completeness into this setting and remove
redundant non-structural rules, then check minimality by separating interpretations tailored to
individual axioms; the resulting presentations are minimal in all arities for qubit Clifford, real
Clifford, and CNOT-dihedral, minimal in bounded ranges for the remaining fragments, and comparable
by one transfer-and-separation pattern.
\end{abstract}

\section{Introduction}

Optimisation and verification of quantum circuits are often carried out by \emph{equational
reasoning}: one rewrites a circuit by replacing a subcircuit with a provably equivalent one
using a fixed set of identities \cite{Kissinger2019,Hietala2021}.  For a fixed gate signature,
this is formalised as an equational theory, namely a set of well-typed equations between
circuits, closed under sequential and parallel composition
\cite{joyal_geometry_1991,selinger_survey_2010,Lack2004}.  With respect to an intended
interpretation by unitary operators, a rule set is \emph{complete} if it derives every equality
between circuits having the same denotation.

A finite complete presentation then gives a compact algebraic description of a circuit
fragment: it records which identities are taken as primitive, which ones are consequences, and
how that distinction depends on the ambient syntax.  If wire permutations are represented by
ordinary gates, then naturality and swap-generator interaction equations describe how those
permutations interact with the generators of the fragment; if permutations are structural, those
equations become part of the surrounding
calculus.  Reducing a presentation is therefore more than shortening a list of rules, because it
can reveal which equalities belong to generic wiring and which belong to the algebra of the fragment
itself.

Alongside completeness, we study \emph{axiom independence} and \emph{minimality}
\cite{ClementDelormePerdrix2023,Backens2020,Vilmart2019}: an axiom is independent if it is not
derivable from the others, and a finite presentation is minimal if all its axioms are independent.
Minimal presentations isolate the irreducible algebraic content of a fragment and are also
attractive for automated rewriting, where redundant rules enlarge the search space
\cite{KnuthBendix1970,DershowitzJouannaud1990}.  For the fragments studied here, minimality also
gives a useful check on the simplification process: after derivable rules have been removed, each
remaining axiom is tested against the possibility that it was already forced by the others.

This work studies six near-Clifford fragments--qubit Clifford, real Clifford, Clifford\(+T\) up to
two qubits, Clifford\(+CS\) up to three qubits, CNOT-dihedral, and qutrit Clifford--and places them
in a common PROP setting, so that swaps are structural rather than fragment-specific.  In this
setting, objects are wire counts, morphisms are circuits, sequential composition is plugging,
tensor is side-by-side juxtaposition, and the PROP structure accounts for wire crossings, thereby
separating generic wiring from the fragment-specific equations whose size and independence we want
to analyse.

\begin{table}[ht]
  \centering
  \small
  \setlength{\tabcolsep}{4pt}
  \renewcommand{\arraystretch}{1.2}
  \begin{tabular}{@{}l c p{0.60\linewidth}@{}}
    \hline
    \textbf{Fragment} & \textbf{Dimension} & \textbf{Non-structural generators} \\
    \hline
    Qubit Clifford
      & \(2\)
      & \(\galpha:0\to0, \gH,\gS:1\to1, \gCNOT:2\to2\) \\
    Real Clifford
      & \(2\)
      & \(\gminus:0\to0, \gH,\gZ:1\to1, \gCNOT:2\to2\) \\
    Clifford\(+T\)
      & \(2\)
      & \(\galpha:0\to0, \gH,\gT:1\to1, \gCNOT:2\to2\) \\
    Clifford\(+CS\)
      & \(2\)
      & \(\galpha:0\to0, \gH,\gS:1\to1, \gCS:2\to2\) \\
    CNOT-dihedral
      & \(2\)
      & \(\galpha:0\to0, \gX,\gT:1\to1, \gCNOT:2\to2\) \\
    Qutrit Clifford
      & \(3\)
      & \(\gww:0\to0, \gH,\gS:1\to1, \gCNOT:2\to2\) \\
    \hline
  \end{tabular}
  \caption{Gate signatures used for the six presented fragments; swaps are structural and therefore
  not listed as fragment generators.}
  \label{tab:fragments}
\end{table}

The source presentations in the literature are therefore our point of departure: after moving to a
common PROP setting with structural symmetries, we seek presentations with fewer non-structural
axioms and separator data for the minimality bounds listed in \cref{tab:recap}, so simplification, completeness
transfer, and minimality become one comparison rather than separate rule-count exercises.

\subsection{Related work}

Generators-and-relations completeness theorems already give presentations for all six fragments of
\cref{tab:fragments}:
\;the \(n\)-qubit Clifford presentation \cite{SelingerStabilizer}, the real Clifford
presentation \cite{SelingerRealStabilizer}, the CNOT-dihedral presentation
\cite{AmyCNOTDihedral}, the \(2\)-qubit Clifford\(+T\) presentation
\cite{SelingerCliffordPlusT}, the \(3\)-qubit Clifford\(+CS\) presentation
\cite{SelingerCliffordPlusCS}, and the \(n\)-qutrit Clifford presentation
\cite{qutritclifford}.  These theorems are the starting point for the completeness
transfer in \cref{sec:completeness}: explicit encodings and decodings carry the known
complete presentations to the smaller PROP presentations of \cref{sec:circuits-relations}.
Related Agda formalizations already check the \(2\)-qubit Clifford\(+T\) and \(3\)-qubit
Clifford\(+CS\) presentations \cite{SelingerCliffordPlusT,SelingerCliffordPlusCS}.

Beyond the near-Clifford setting, recent completeness theorems cover more expressive
circuit languages \cite{CHMPV2022,Clement2024}, where the extra expressivity changes the
shape of the axiomatisation: for such languages, complete equational theories require rules of
unbounded arity \cite{ClementDelormePerdrix2023}.  The latter work also uses separating models to
prove minimality for a complete theory of general quantum circuits, and \cref{sec:minimality}
adapts that proof pattern to the six fragments of \cref{tab:fragments}.  This work preserves the
strict unitary semantics of the earlier completeness papers and asks which axioms remain necessary
after simplification.

\subsection{Contributions}

\begin{enumerate}
\item \textbf{Smaller presentations in a uniform PROP setting.}
      For each fragment we give a finite presentation with fewer non-structural axioms than
      the standard one while preserving the same strict unitary semantics.
\item \textbf{Completeness by transfer.}
      The transfer isolates a finite syntactic comparison: explicit encoding/decoding maps carry
      earlier complete presentations to the smaller ones after aligning the ambient PROP structure
      and scalar conventions where necessary.
\item \textbf{Minimality and independence results for the simplified rule sets.}
      By separating interpretations, we prove that the new presentations for qubit Clifford, real
      Clifford, and CNOT-dihedral are minimal in all arities, while the Clifford\(+T\),
      Clifford\(+CS\), and qutrit Clifford presentations are minimal in the bounded ranges listed in
      \cref{tab:recap}.
\end{enumerate}

The arity bounds have two origins: for Clifford\(+T\) and Clifford\(+CS\), the imported
completeness theorems themselves are proved only up to \(2\) and \(3\) qubits, while for
minimality the remaining bounds are the arities for which we construct explicit separating
interpretations.

After the generic PROP structure has been factored out, \cref{tab:recap} compares non-structural
rule counts; the removed source rules are not lost, since they are derivable from the simplified
presentations.

\newcolumntype{C}[1]{>{\centering\arraybackslash}p{#1}}
\begin{table}[htbp]
  \centering
  \small
  \setlength{\tabcolsep}{3pt}
  \begin{tabular}{@{}C{0.30\textwidth}C{0.17\textwidth}C{0.38\textwidth}C{0.09\textwidth}@{}}
    \hline
    \textbf{Fragment} & \textbf{Previous Presentation} & \textbf{Our Presentation} & \textbf{Figure}\\
    \hline
    Qubit Clifford on \(n\) qubits
      & 15 rules \cite{SelingerStabilizer}
      & 8 rules (minimal in all arities) & \ref{fig:qubitcliffordaxioms} \\
    Real Clifford on \(n\) qubits
      & 16 rules \cite{SelingerRealStabilizer}
      & 10 rules (minimal in all arities) & \ref{fig:qubitrealcliffordaxioms} \\
    Qutrit Clifford on \(n\) qutrits
      & 18 rules \cite{qutritclifford}
      & 10 rules (minimal up to 2 qutrits; conjectured in all arities) & \ref{fig:qutritcliffordaxioms} \\
    Clifford+T on 2 qubits
      & 18 rules \cite{SelingerCliffordPlusT}
      & 11 rules (minimal up to 1 qubit) & \ref{fig:qubitcliffordplustaxioms} \\
    Clifford+CS on 3 qubits
      & 17 rules \cite{SelingerCliffordPlusCS}
      & 14 rules (minimal up to 2 qubits) & \ref{fig:qubitcliffordpluscsaxioms} \\
    CNOT-dihedral on \(n\) qubits
      & 13 rules \cite{AmyCNOTDihedral}
      & 11 rules (minimal in all arities) & \ref{fig:qubitdihedralaxioms} \\
    \hline
  \end{tabular}
  \caption{Non-structural rule-count comparison after moving from source PRO presentations to PROP
  presentations with structural symmetries.}
  \label{tab:recap}
\end{table}

The rest of the paper follows the proof obligations behind \cref{tab:recap}: \cref{sec:circuits}
fixes the PROP syntax used to lift PRO presentations, \cref{sec:circuits-relations} defines the
six presented circuit fragments and their unitary interpretations, \cref{sec:completeness} transfers
completeness from the source presentations and handles scalar refinements, \cref{sec:minimality}
checks independence by separating interpretations, and the appendices collect the longer
derivations and case-by-case checks supporting the main proofs.
\section{Graphical languages for quantum circuits}\label{sec:circuits}

The formal language is categorical, but its role is practical: it fixes the syntax for the usual
boxes-and-wires calculus.  In this setting, generic wiring is treated structurally, so the
displayed rule sets can focus on fragment-specific generators and equations.

\subsection{Graphical languages}

We work with the one-object-per-wire formalisms PROs and PROPs in the sense of
Joyal--Street~\cite{joyal_geometry_1991, selinger_survey_2010}.
We adopt strictness throughout: associativity
and unit constraints are treated as literal equalities, so parentheses can be omitted.

\begin{definition}
A \emph{PRO} is a strict monoidal category \(\cat{P}\) whose objects are the
natural numbers \(n\in\N\) and whose tensor on objects is addition
\(n\otimes m \defeq n+m\) with unit \(0\).  We write \(\cat{P}(n,m)\) for morphisms
\(n\to m\), with sequential composition \(\circ\) and parallel composition \(\otimes\),
subject to the usual strict monoidal axioms
(\cref{eq:seqidentity,eq:seqassociativity,eq:parassociativity,eq:paridentity,eq:interchange}).
\end{definition}

We write \(f:n\to m\) when \(f\in\cat{P}(n,m)\); \(n\) and \(m\) are the numbers of input
and output wires of~\(f\).  In particular \(\id_n=\id_1^{\otimes n}\) for all \(n\in\N\).

\begin{definition}
A \emph{PROP} is a PRO \(\cat{P}\) equipped with a symmetric monoidal structure,
i.e. symmetry isomorphisms \(\sigma_{n,m}:n+m\to m+n\) \((n,m\in\N)\) generated by
the basic swap \(\sigma\defeq\sigma_{1,1}\) and satisfying the axioms
\eqref{eq:involution}-\eqref{eq:naturality} in \cref{fig:coherenceprop}.  We
write \(\sigma_n\defeq\sigma_{n,1}:n+1\to 1+n\) for swapping one wire with a block
of \(n\) wires.
\end{definition}
In a PROP, swaps and permutations are part of the ambient structure, and the resulting
string-diagram representation draws a morphism \(f:n\to m\) as a box with \(n\) incoming wires and
\(m\) outgoing wires.
Sequential composition is drawn by stacking boxes, tensor by placing them side by side, and
\(\sigma_{n,m}\) by crossing a block of \(n\) wires with a block of \(m\) wires.  With these conventions,
the coherence axioms are exactly the diagrammatic moves that make deformation sound.
For example, if \(g:1\to1\) and \(h:2\to2\) are gates, then \(g\otimes h:3\to3\) is their
side-by-side circuit, and composing with \(\sigma_{1,2}\) just moves the first output wire past the
other two.  Such crossings are part of the syntax before any fragment-specific rule is added.

\begin{center}
  \begin{tabular}{c@{\hskip 1.4em}c@{\hskip 1.4em}c@{\hskip 1.4em}c@{\hskip 1.4em}c@{\hskip 1.4em}c@{\hskip 1.4em}c@{\hskip 1.4em}c@{\hskip 1.4em}}
    \(f\) & \(\id_0\) & \(\id_1\) & \(\id_n\) &
    \(\sigma\) & \(\sigma_{n}\) & \(g\circ f\) & \(f\otimes g\) \\
    \input{category/f.tikz} & \input{category/empty.tikz} &
    \input{category/id.tikz} & \input{category/idn.tikz} &
    \input{gates/SWAP.tikz} & \input{category/perm.tikz} &
    \input{category/fg.tikz} & \input{category/fparg.tikz}
  \end{tabular}
\end{center}

Equalities obtained solely by these coherence moves, including permutations, are
\emph{structural} and may be omitted from diagrammatic arguments.

\captionsetup[subfigure]{justification=centering}
\begin{figure}[ht]
  \centering
  \begin{subfigure}[t]{0.43\textwidth}
    \centering
    \scalebox{0.9}{\fbox{\input{category/seqidentity.tikz}}}
    \caption{\(f\circ\id_n = f = \id_n\circ f\)}
    \label{eq:seqidentity}
  \end{subfigure}%
  \begin{subfigure}[t]{0.43\textwidth}
    \centering
    \scalebox{0.9}{\fbox{\input{category/seqassociativity.tikz}}}
    \caption{\(h\circ(g\circ f) = (h\circ g)\circ f\)}
    \label{eq:seqassociativity}
  \end{subfigure}
  \begin{subfigure}[t]{0.30\textwidth}
    \centering
    \scalebox{0.9}{\fbox{\input{category/parassociativity.tikz}}}
    \caption{\((f\!\otimes\! g)\!\otimes\! h \!=\! f\!\otimes\!(g\!\otimes\! h)\)}
    \label{eq:parassociativity}
  \end{subfigure}%
  \begin{subfigure}[t]{0.45\textwidth}
    \centering
    \scalebox{0.9}{\fbox{\input{category/interchange.tikz}}}
    \caption{\((g_1\!\circ\! f_1)\!\otimes\!(g_2\!\circ\! f_2)
         \!=\! (g_1\!\otimes\! g_2)\!\circ\!(f_1\!\otimes\! f_2)\)}
    \label{eq:interchange}
  \end{subfigure}
  
  \vspace{1ex}

  \begin{subfigure}[t]{0.38\textwidth}
    \centering
    \scalebox{0.9}{\fbox{\input{category/paridentity.tikz}}}
    \caption{\(f \otimes \id_0 = f = \id_0 \otimes f\)}
    \label{eq:paridentity}
  \end{subfigure}%
  \begin{subfigure}[t]{0.26\textwidth}
    \centering
    \scalebox{0.9}{\fbox{\input{category/involution.tikz}}}
    \caption{\(\sigma\circ\sigma = \id_2\)}
    \label{eq:involution}
  \end{subfigure}%
  \begin{subfigure}[t]{0.35\textwidth}
    \centering
    \scalebox{0.9}{\fbox{\input{category/naturality.tikz}}}
    \caption{\(\sigma_m\circ(f\otimes\id_1) = (\id_1\otimes f)\circ\sigma_n\) for \(f:n\to m\)}
    \label{eq:naturality}
  \end{subfigure}
  \caption{Coherence laws for PROs and PROPs.}
  \label{fig:coherenceprop}
\end{figure}

\subsection{Free PROPs and presentations}\label{free-prop}

Circuit fragments are specified by \emph{generators and relations}.  Generators represent primitive
gates with fixed arity, and relations are the rewrite rules used for equational reasoning.

\begin{definition}\label{def:signature}
A \emph{(PROP) signature} is a set \(\Sigma\) equipped with arity functions
\(\dom,\cod:\Sigma\to\N\).
An element \(g\in\Sigma\) is written as a generating morphism \(g:\dom(g)\to\cod(g)\).
\end{definition}

\begin{definition}\label{def:free-prop}
Let \(\Sigma\) be a signature.  The \emph{free PROP} on \(\Sigma\), written \(\cat{P}_\Sigma\), is the
PROP whose morphisms are string diagrams generated from:
\begin{itemize}
\item the non-structural generators \(g\in\Sigma\) (with their prescribed arities), and
\item the structural morphisms of a PROP (identities and symmetries),
\end{itemize}
quotiented by PROP coherence (i.e.\ the axioms of \cref{fig:coherenceprop}).
\end{definition}

\begin{definition}\label{def:presented-prop}
Let \(\Sigma\) be a signature and let \(\mathcal{R}\) be a set of well-typed equations
\(L=R\) between morphisms of the free PROP \(\cat{P}_\Sigma\).
The \emph{presented PROP} \(\cat{P}_\Sigma/\mathcal{R}\) is the quotient of \(\cat{P}_\Sigma\)
by the smallest congruence that contains \(\mathcal{R}\) and is closed under \(\circ\), \(\otimes\),
and symmetry.
\end{definition}

This quotient is the \emph{graphical language} used in the rest of the paper: diagrams are equal exactly when their
equality is derivable from PROP coherence together with the chosen relations.



\subsection{Monoidal functors and completeness}

An interpretation sends the presented syntax to its intended operations; completeness is the
assertion that every semantic equality in that image can be recovered as a derivation in the
quotient.

\begin{definition}\label{def:prop-morphism}
A \emph{PROP morphism} is a strict symmetric monoidal functor
\(F:\cat{P}_1\to\cat{P}_2\) between PROPs that is the identity on objects.
Equivalently, \(F\) preserves identities, symmetries, sequential composition and
tensor.
\end{definition}

\begin{definition}\label{def:subprop}
Let \(\cat{P}\) be a PROP.
A \emph{subPROP} of \(\cat{P}\) is a wide strict symmetric monoidal subcategory
\(\cat{Q}\subseteq\cat{P}\).
Since the objects of a PROP are fixed to be \(\mathbb{N}\), this amounts to specifying
subsets
\(
  \cat{Q}(n,m)\subseteq \cat{P}(n,m)\;(n,m\in\mathbb{N})
\)
such that \(\id_n\in \cat{Q}(n,n)\) for all \(n\); \(\sigma_{n,m}\in \cat{Q}(n+m,m+n)\) for all \(n,m\); and
\(\cat{Q}\) is closed under \(\circ\) and \(\otimes\).

In this situation, the inclusion functor \(\cat{Q}\hookrightarrow \cat{P}\) is a faithful PROP morphism.
\end{definition}

\begin{definition}\label{def:interpretation-functor}
Let \(\cat{P}_\Sigma/\mathcal{R}\) be a presented PROP and let \(\cat{C}\) be a
symmetric monoidal category. An \emph{interpretation functor} is a strict
symmetric monoidal functor
\(
  \interp{\cdot}:\cat{P}_\Sigma\to\cat{C}
\)
specified on generators \(g\in\Sigma\) and extended monoidally to all diagrams.

If every equation in \(\mathcal{R}\) holds in \(\cat{C}\), then
\(\interp{\cdot}\) factors uniquely through the quotient, yielding a functor
\(
  \interp{\cdot}:\cat{P}_\Sigma/\mathcal{R}\to\cat{C}.
\)
\end{definition}

\begin{definition}\label{def:semantic-completeness}
Let \(\interp{\cdot}:\cat{P}_\Sigma/\mathcal{R}\to\cat{C}\) be such an
interpretation.  We say that the graphical language \(\cat{P}_\Sigma/\mathcal{R}\)
is \emph{complete} for this semantics if, for all circuits \(C_1,C_2\), \(\interp{C_1}=\interp{C_2}\) implies
\(\cat{P}_\Sigma/\mathcal{R}\vdash C_1=C_2\).  Equivalently, the interpretation
functor is faithful.

We will also use a bounded variant. Given \(k\in\mathbb{N}\), we say that
\(\cat{P}_\Sigma/\mathcal{R}\) is \emph{complete up to \(k\) wires} if the
implication holds whenever \(C_1,C_2:n\to m\) with \(n,m\le k\).  Equivalently, the
interpretation functor is faithful on each hom-set
\(\cat{P}_\Sigma/\mathcal{R}(n,m)\) for \(n,m\le k\).
\end{definition}
For equational reasoning, completeness says that semantic equality can always be witnessed by an explicit
diagrammatic derivation from the chosen relations; the bounded form restricts this requirement to
circuits of limited arity.

When \(\cat{C}\) is, by definition, the symmetric monoidal subcategory of an
ambient category generated by the images of the generators, the interpretation
functor is also full; in that case \(\cat{P}_\Sigma/\mathcal{R}\) is actually
isomorphic to the semantic subPROP~\(\cat{C}\).

Completeness theorems stated for PROs can still be used in this PROP setting through a single
lifting principle.

\begin{lemma}[PRO-to-PROP lifting]\label{lem:fromprotoprop}\label{rem:fromprotoprop}
Let \(\cat{P}/\mathcal{R}\) be a complete presented PRO with an interpretation into a \emph{symmetric} monoidal category \(\cat{C}\). 
Assume there exists a PRO circuit
\(\tau:2\to 2\) such that \(\interp{\tau}\) is the symmetry \(\sigma_{1,1}\) of
\(\cat{C}\).
Form the free PROP on the same generators, and quotient it by \(\mathcal{R}\) together with the single equation
\(\sigma_{1,1}=\tau\).
Then the resulting PROP presentation is complete for the same semantics.
\end{lemma}
\begin{proof}
The equation \(\sigma_{1,1}=\tau\) rewrites each structural swap between adjacent wires as a source PRO circuit.
By the symmetric monoidal axioms of a PROP, every structural permutation is generated by such adjacent swaps,
so each PROP diagram \(D\) has a translated PRO diagram \(T(D)\) obtained by replacing structural swaps by copies of
\(\tau\) and leaving the ordinary generators unchanged.  The interpretation satisfies
\(\interp{D}=\interp{T(D)}\) because \(\interp{\tau}=\sigma_{1,1}\).  If two PROP diagrams have the same semantics,
their translations have the same semantics as PRO diagrams, hence are equal by PRO completeness.  Replacing
each occurrence of \(\tau\) by \(\sigma_{1,1}\) then gives the original PROP equality.
\end{proof}
\section{Quantum circuits and their relations}
\label{sec:circuits-relations}

The transfer and separation arguments use one Hilbert-space semantic setup: five qubit fragments
share a signature, while qutrit Clifford uses its own presented subPROP.

Let \(\cat{FdHilb}\) be the symmetric monoidal category of finite-dimensional
complex Hilbert spaces and linear maps, with tensor product~\(\otimes\), tensor
unit~\(\C\), and symmetry given by swapping tensor factors.  By Mac~Lane's
coherence theorem \cite{Saunders} we may regard \(\cat{FdHilb}\) as strict monoidal.

\begin{definition}\label{def:unitary}
Let \(H,K\) be finite-dimensional Hilbert spaces. A linear map \(U:H\to K\) is \emph{unitary} if
\(
  U^\dagger U = \id_H
\) and \(
  UU^\dagger = \id_K,
\)
where \(U^\dagger\) denotes the Hilbert adjoint of \(U\).
We write \(\mathrm{U}(H,K)\) for the set of unitary maps \(H\to K\), and \(\mathrm{U}(H)\coloneqq \mathrm{U}(H,H)\).
\end{definition}

\begin{definition}\label{def:qubit-qutrit-props}
For \(d\ge 2\), let \(\cat{Qudit}_d\) be the PROP with objects \(\N\) and
hom-sets \(\cat{Qudit}_d(n,m)\coloneqq
\mathrm{U}\big((\C^d)^{\otimes n},(\C^d)^{\otimes m}\big)\), with composition,
identities, tensor, and symmetry inherited from \(\cat{FdHilb}\) viewed as
strict monoidal.  The qubit and qutrit PROPs are
\(\cat{Qubit}\coloneqq\cat{Qudit}_2\) and
\(\cat{Qutrit}\coloneqq\cat{Qudit}_3\).
\end{definition}
For \(d\ge2\), \(\cat{Qudit}_d\) is endomorphism-only: there are no unitaries
\(n\to m\) for \(n\neq m\).
For a set \(G\) of morphisms in a PROP, \(\langle G\rangle\) denotes the smallest subPROP
containing \(G\) and all structural morphisms.

\subsection{Qubit circuit fragments as presented subPROPs}

The five qubit fragments share a common qubit signature and unitary interpretation.

\begin{definition}
The \emph{common qubit signature} \(\Sigma^{(2)}\) has non-structural generators
\(\galpha,\gminus:0\to0\),
\(\gH,\gZ,\gS,\gT,\gX:1\to1\), and \(\gCNOT,\gCS:2\to2\).
We write \(\cat{P}^{(2)}\defeq\cat{P}_{\Sigma^{(2)}}\) for the free PROP on
\(\Sigma^{(2)}\).  Structural morphisms \(\id_n\) and \(\sigma_{n,m}\) come from the
ambient PROP structure.
\end{definition}

\begin{remark}\label{rem:arity-preservation}
All signatures used in this section are endomorphism-only: each generator has type \(n\to n\).
Hence every well-typed fragment circuit, and every equation between such circuits, preserves the
wire count; this is the arity fact used in \cref{sec:minimality}.
\end{remark}

We use the standard Dirac notation for computational basis states: for \(x\in\{0,1\}\) we write \(\ket{x}\in\C^2\),
and for \(x,y\in\{0,1\}\) we abbreviate \(\ket{x,y}\coloneqq \ket{x}\otimes\ket{y}\) (similarly,
\(\ket{x_1,\ldots,x_n}\coloneqq \ket{x_1}\otimes\cdots\otimes\ket{x_n}\)).

By the universal property of the free PROP, there is a unique strict
symmetric monoidal functor
\(
  \interp{\cdot}^{(2)}:\cat{P}^{(2)}\to\cat{Qubit}
\)
sending:
\begin{equation*}
\begin{array}{l l}
 \interp{\galpha}^{(2)} : 1\longmapsto e^{i\pi/4} &
 \interp{\gH}^{(2)} : \ket{x}\longmapsto
     \tfrac1{\sqrt2}\!\sum_{k=0}^1(-1)^{xk}\ket{k}\\
 \interp{\gminus}^{(2)} : 1\longmapsto -1 &
 \interp{\gS}^{(2)} : \ket{x}\longmapsto i^x\ket{x} \\
 \interp{\gX}^{(2)} : \ket{x}\longmapsto\ket{1-x} &
 \interp{\gZ}^{(2)} : \ket{x}\longmapsto(-1)^x\ket{x} \\
 \interp{\gCNOT}^{(2)} : \ket{x,y}\longmapsto\ket{x,x\oplus y} &
 \interp{\gT}^{(2)} : \ket{x}\longmapsto e^{i\pi x/4}\ket{x} \\
 \interp{\gCS}^{(2)} : \ket{x,y}\longmapsto i^{xy}\ket{x,y}
\end{array}
\end{equation*}
for \(x,y\in\{0,1\}\) and addition \(\oplus\) modulo~2.

Each qubit fragment is specified by a sub-signature of \(\Sigma^{(2)}\) and equations on the free
PROP it generates; the common unitary interpretation then restricts to this syntax.  The index
\(\bullet\) records the chosen fragment, while \(\Sigma_\bullet\), \(\cat{P}_\bullet\),
\(\QC_\bullet\), \(\QCirc_\bullet\), \(\cat{Qubit}_\bullet\), and
\(\interp{\cdot}_\bullet\) denote its gates, free PROP, equations, presented PROP, generated
semantic subPROP, and quotient interpretation, respectively.

\begin{definition}\label{def:qubit-fragment-scheme}
Let
\(
  \bullet \in \{\mathit{Cliff}, \mathit{RCliff}, \mathit{CliffT}, \mathit{CliffCS}, \mathit{CNOTdihe}\}
\)
range over the five qubit fragments considered in this paper.  Given a sub-signature \(\Sigma_\bullet\subseteq\Sigma^{(2)}\), let
\(\cat{P}_\bullet\defeq\cat{P}_{\Sigma_\bullet}\) be the free PROP on \(\Sigma_\bullet\).  If
\(\QC_\bullet\) is a finite set of well-typed equations in \(\cat{P}_\bullet\), define the presented
PROP
\(
  \QCirc_\bullet\defeq\cat{P}_\bullet/\QC_\bullet.
\)
The semantic target is the subPROP
\(
  \cat{Qubit}_\bullet
  \defeq
  \big\langle\,\interp{g}^{(2)} \mid g\in\Sigma_\bullet\,\big\rangle
  \subseteq
  \cat{Qubit},
\)
generated by the interpretations of the gates in \(\Sigma_\bullet\).  When \(\QC_\bullet\) is sound
for the restricted common interpretation, the quotient carries the induced strict symmetric
monoidal functor
\(
  \interp{\cdot}_\bullet:\QCirc_\bullet\to\cat{Qubit}_\bullet.
\)
We call such a presentation complete when this functor is faithful.
\end{definition}

\begin{definition}\label{def:qubit-fragments}
The five qubit fragments treated in this paper are specified by these sub-signatures and rule sets:
\begin{center}
\begin{tabular}{@{}l p{0.38\linewidth} p{0.38\linewidth}@{}}
\textbf{Fragment} & \(\Sigma_\bullet\subseteq\Sigma^{(2)}\) & \textbf{Rule set \(\QC_\bullet\)} \\
\hline
Clifford
& \(\Sigma_{\mathit{Cliff}}\defeq\{\galpha,\gH,\gS,\gCNOT\}\)
& \(\QC_{\mathit{Cliff}}\defeq\QCclifford\) (\cref{fig:qubitcliffordaxioms})
\\
Real Clifford
& \(\Sigma_{\mathit{RCliff}}\defeq\{\gminus,\gH,\gZ,\gCNOT\}\)
& \(\QC_{\mathit{RCliff}}\defeq\QCrealclifford\) (\cref{fig:qubitrealcliffordaxioms})
\\
Clifford\(+T\)
& \(\Sigma_{\mathit{CliffT}}\defeq\{\galpha,\gH,\gT,\gCNOT\}\)
& \(\QC_{\mathit{CliffT}}\defeq\QCcliffordplust\) (\cref{fig:qubitcliffordplustaxioms})
\\
Clifford\(+CS\)
& \(\Sigma_{\mathit{CliffCS}}\defeq\{\galpha,\gH,\gS,\gCS\}\)
& \(\QC_{\mathit{CliffCS}}\defeq\QCcliffordpluscs\) (\cref{fig:qubitcliffordpluscsaxioms})
\\
CNOT-dihedral
& \(\Sigma_{\mathit{CNOTdihe}}\defeq\{\galpha,\gX,\gT,\gCNOT\}\)
& \(\QC_{\mathit{CNOTdihe}}\defeq\QCdihedral\) (\cref{fig:qubitdihedralaxioms})
\end{tabular}
\end{center}
\end{definition}

\subsection{The qutrit Clifford fragment}
\label{subsec:qutrit-clifford-fragment}

For qutrits, the paper uses only the Clifford fragment, denoted \(\mathit{Cliff3}\), whose
signature is separate from the common qubit signature; we use the computational basis
\(\{\ket{0},\ket{1},\ket{2}\}\) of \(\C^3\) and the induced tensor-product basis.

\begin{definition}\label{def:qutrit-clifford-fragment}
The \emph{qutrit Clifford} fragment is generated by the signature
\(\Sigma_{\mathit{Cliff3}}\defeq
\{\gww:0\to0,\gH,\gS:1\to1,\gCNOT:2\to2\}\).

Let \(\cat{P}_{\mathit{Cliff3}}\defeq\cat{P}_{\Sigma_{\mathit{Cliff3}}}\) be the free PROP on
\(\Sigma_{\mathit{Cliff3}}\) and let \(\QC_{\mathit{Cliff3}}\defeq\QCqutritclifford\) be the rule
set given in \cref{fig:qutritcliffordaxioms}.  The presented PROP is
\(
  \QCirc_{\mathit{Cliff3}}\defeq\cat{P}_{\mathit{Cliff3}}/\QC_{\mathit{Cliff3}}.
\)
\end{definition}

By the universal property of the free PROP, the assignments
\begin{equation*}
\begin{array}{l l}
 \interp{\gww}_{\mathit{Cliff3}} : 1\longmapsto e^{i\pi/6} & 
 \interp{\gH}_{\mathit{Cliff3}} : \ket{x}\longmapsto 
   \tfrac1{\sqrt3}\!\sum_{k=0}^2 e^{2\pi i\,xk/3}\ket{k} \\
 \interp{\gCNOT}_{\mathit{Cliff3}} : \ket{x,y}\longmapsto \ket{x,(x+y)\bmod 3} &
 \interp{\gS}_{\mathit{Cliff3}} : \ket{x}\longmapsto
   e^{i\pi x(x-1)/3}\ket{x} 
\end{array}
\end{equation*}
for \(x,y\in\{0,1,2\}\) determine a strict symmetric monoidal functor
\(\cat{P}_{\mathit{Cliff3}}\to\cat{Qutrit}\).  Let
\(\cat{Qutrit}_{\mathit{Cliff3}}\defeq\langle\,\interp{g}_{\mathit{Cliff3}}\mid g\in\Sigma_{\mathit{Cliff3}}\,\rangle\)
be the semantic subPROP generated by the qutrit Clifford gates.  Since the equations of
\(\QC_{\mathit{Cliff3}}\) are sound for this interpretation, the free-PROP interpretation factors
through the quotient and yields
\(
  \interp{\cdot}_{\mathit{Cliff3}}:\QCirc_{\mathit{Cliff3}}\longrightarrow\cat{Qutrit}_{\mathit{Cliff3}}.
\)

\subsection{Rule-set presentations}

Each displayed equation is an arity-preserving schema; gates outside the corresponding fragment
signature are shortcuts expanded as in \cref{fig:quantum-circuit-shortcuts}.

\begin{figure}[ht]
    \centering
    \fbox{\begin{minipage}{0.975\linewidth}
    \centering
    \vspace{-1em}
    \begin{minipage}[c]{0.22\linewidth}\centering
      \begin{equation}\tag{\omegarulabel{8}}\label{w8}
        \scalebox{0.9}{\input{gates/wc.tikz}}^{\otimes 8} = \scalebox{0.9}{\input{gates/empty.tikz}}
      \end{equation}
    \end{minipage}%
    \begin{minipage}[c]{0.24\linewidth}\centering
      \begin{equation}\tag{\rulabel{H^2}}\label{H2}
        \scalebox{0.9}{\input{cliffordaxioms/HH.tikz}} = \scalebox{0.9}{\input{gates/Id.tikz}}
      \end{equation}
    \end{minipage}%
    \begin{minipage}[c]{0.32\linewidth}\centering
      \begin{equation}\tag{\rulabel{S^4}}\label{S4}
        \scalebox{0.9}{\input{cliffordaxioms/SSSS.tikz}} = \scalebox{0.9}{\input{gates/Id.tikz}}
      \end{equation}
    \end{minipage}\vspace{-0.6em}
    \begin{minipage}[c]{0.54\linewidth}\centering
      \begin{equation}\tag{E}\label{E}
        \scalebox{0.9}{\input{cliffordaxioms/HSH.tikz}} = \scalebox{0.9}{\input{cliffordaxioms/wSdaggerHSdagger.tikz}}
      \end{equation}
    \end{minipage}%
    \begin{minipage}[c]{0.32\linewidth}\centering
      \begin{equation}\tag{CPh}\label{Cs}
        \scalebox{0.9}{\input{cliffordaxioms/CNOTSCNOT.tikz}} = \scalebox{0.9}{\input{cliffordaxioms/S.tikz}}
      \end{equation}
    \end{minipage}\vspace{-0.6em}
    \begin{minipage}[c]{0.29\linewidth}\centering
      \begin{equation}\tag{B}\label{B}
        \scalebox{0.9}{\input{cliffordaxioms/CNOTNOTC.tikz}} = \scalebox{0.9}{\input{cliffordaxioms/SWAPCNOT.tikz}}
      \end{equation}
    \end{minipage}%
    \begin{minipage}[c]{0.47\linewidth}\centering
      \begin{equation}\tag{CZ}\label{CZ}
        \scalebox{0.9}{\input{cliffordaxioms/HCNOTH.tikz}} = \scalebox{0.9}{\input{cliffordaxioms/CZ.tikz}}
      \end{equation}
    \end{minipage}\vspace{-0.6em}
    \begin{minipage}[c]{0.29\linewidth}\centering
      \begin{equation}\tag{I}\label{Inew}
        \scalebox{0.9}{\input{qutritcliffordaxioms/CNOT23CNOT12CNOT13.tikz}} = \scalebox{0.9}{\input{qutritcliffordaxioms/CNOT12CNOT23.tikz}}
      \end{equation}
    \end{minipage}
    \end{minipage}}
    \caption{Qubit Clifford presentation \(\QCclifford\), with the scalar, one-qubit, controlled-phase, swap-decomposition, controlled-\(Z\), and three-wire interaction laws used in the completeness and minimality arguments.}
    \label{fig:qubitcliffordaxioms}
\end{figure}
\begin{figure}[ht]
    \centering
    \fbox{\begin{minipage}{0.975\linewidth}
        \centering
        \vspace{-1em}
        \begin{minipage}[c]{0.22\linewidth}\centering
            \begin{equation}\tag{\rulabel{-^2}}\label{realminus2}
                \scalebox{0.9}{\input{gates/minus.tikz}}^{\otimes 2} = \scalebox{0.9}{\input{gates/empty.tikz}}
            \end{equation}
        \end{minipage}
        \begin{minipage}[c]{0.24\linewidth}\centering
            \begin{equation}\tag{\rulabel{H^2}}\label{realH2}
                \scalebox{0.9}{\input{realcliffordaxioms/H2.tikz}} = \scalebox{0.9}{\input{gates/Id.tikz}}
            \end{equation}
        \end{minipage}
        \begin{minipage}[c]{0.24\linewidth}\centering
            \begin{equation}\tag{\rulabel{Z^2}}\label{realZ2}
                \scalebox{0.9}{\input{realcliffordaxioms/Z2.tikz}} = \scalebox{0.9}{\input{gates/Id.tikz}}
            \end{equation}
        \end{minipage}\vspace{-0.6em}
        \begin{minipage}[c]{0.50\linewidth}\centering
            \begin{equation}\tag{F}\label{realF}
                \scalebox{0.9}{\input{realcliffordaxioms/ZHZH.tikz}} = \scalebox{0.9}{\input{realcliffordaxioms/minusHZHZ.tikz}}
            \end{equation}
        \end{minipage}
        \begin{minipage}[c]{0.30\linewidth}\centering
            \begin{equation}\tag{\rulabel{CX^2}}\label{realCX2}
                \scalebox{0.9}{\input{realcliffordaxioms/CNOTCNOT.tikz}} = \scalebox{0.9}{\input{realcliffordaxioms/II.tikz}}
            \end{equation}
        \end{minipage}\vspace{-0.6em}
        \begin{minipage}[c]{0.29\linewidth}\centering
            \begin{equation}\tag{B}\label{realB}
                \scalebox{0.9}{\input{cliffordaxioms/CNOTNOTC.tikz}} = \scalebox{0.9}{\input{cliffordaxioms/SWAPCNOT.tikz}}
            \end{equation}
        \end{minipage}
        \begin{minipage}[c]{0.35\linewidth}\centering
            \begin{equation}\tag{ZC}\label{realwC}
                \scalebox{0.9}{\input{realcliffordaxioms/Z2CNOTZ2.tikz}} = \scalebox{0.9}{\input{realcliffordaxioms/CNOTZ1.tikz}}
            \end{equation}
        \end{minipage}\vspace{-0.6em}
        \begin{minipage}[c]{0.32\linewidth}\centering
            \begin{equation}\tag{CZr}\label{realXC}
                \scalebox{0.9}{\input{realcliffordaxioms/HHCNOTHH.tikz}} = \scalebox{0.9}{\input{realcliffordaxioms/NOTC.tikz}}
            \end{equation}
        \end{minipage}
        \begin{minipage}[c]{0.57\linewidth}\centering
            \begin{equation}\tag{CF}\label{realCF}
                \scalebox{0.9}{\input{realcliffordaxioms/CZcommitCX-left.tikz}} = \scalebox{0.9}{\input{realcliffordaxioms/CZcommitCX-right.tikz}}
            \end{equation}
        \end{minipage}
        \begin{minipage}[c]{0.29\linewidth}\centering
            \begin{equation}\tag{I}\label{realInew}
                \scalebox{0.9}{\input{qutritcliffordaxioms/CNOT23CNOT12CNOT13.tikz}} = \scalebox{0.9}{\input{qutritcliffordaxioms/CNOT12CNOT23.tikz}}
            \end{equation}
        \end{minipage}
    \end{minipage}}
    \caption{Real Clifford presentation \(\QCrealclifford\), where the real phase, Hadamard-\(Z\), controlled-\(X\), controlled-\(Z\), and three-wire interaction laws replace the corresponding Clifford rules in the real fragment.}
    \label{fig:qubitrealcliffordaxioms}
\end{figure}
\begin{figure}[ht]
    \centering
    \fbox{\begin{minipage}{0.975\linewidth}
    \centering
    \vspace{-1em}
    \begin{minipage}[c]{0.24\linewidth}\centering
      \begin{equation}\tag{\omegarulabel{12}}\label{qt-omegapow12}
        \scalebox{0.9}{\input{gates/wb.tikz}}^{\otimes 12} = \scalebox{0.9}{\input{gates/empty.tikz}}
      \end{equation}
    \end{minipage}
    \begin{minipage}[c]{0.29\linewidth}\centering
      \begin{equation}\tag{\rulabel{H^4}}\label{qt-hpow4}
        \scalebox{0.9}{\input{qutritcliffordaxioms/hadpow4.tikz}} = \scalebox{0.9}{\input{cliffordaxioms/I.tikz}}
      \end{equation}
    \end{minipage}
    \begin{minipage}[c]{0.27\linewidth}\centering
      \begin{equation}\tag{\rulabel{S^3}}\label{qt-spow3}
        \scalebox{0.9}{\input{qutritcliffordaxioms/spow3.tikz}} = \scalebox{0.9}{\input{gates/Id.tikz}}
      \end{equation}
    \end{minipage}\vspace{-0.6em}
    \begin{minipage}[c]{0.58\linewidth}\centering
      \begin{equation}\tag{E}\label{qt-shpow3}
        \scalebox{0.9}{\input{qutritcliffordaxioms/shs.tikz}} = \scalebox{0.9}{\input{qutritcliffordaxioms/h3s2h3.tikz}}
      \end{equation}
    \end{minipage}\vspace{-0.6em}
    \begin{minipage}[c]{0.43\linewidth}\centering
      \begin{equation}\tag{\rulabel{SS'}}\label{qt-ssprime}
        \scalebox{0.9}{\input{qutritcliffordaxioms/ssprime.tikz}} = \scalebox{0.9}{\input{qutritcliffordaxioms/sprimes.tikz}}
      \end{equation}
    \end{minipage}
    \begin{minipage}[c]{0.35\linewidth}\centering
      \begin{equation}\tag{CPh}\label{qt-cnotremove}
        \scalebox{0.9}{\input{qutritcliffordaxioms/CNOT-S-K-CNOT.tikz}} = \scalebox{0.9}{\input{qutritcliffordaxioms/S-K.tikz}}
      \end{equation}
    \end{minipage}\vspace{-0.6em}
    \begin{minipage}[c]{0.38\linewidth}\centering
      \begin{equation}\tag{KC}\label{qt-kcnot}
        \scalebox{0.9}{\input{qutritcliffordaxioms/K-CNOT.tikz}} = \scalebox{0.9}{\input{qutritcliffordaxioms/CNOT-CNOT-K.tikz}}
      \end{equation}
    \end{minipage}
    \begin{minipage}[c]{0.57\linewidth}\centering
      \begin{equation}\tag{CZ}\label{qt-czdecomp}
        \scalebox{0.9}{\input{qutritcliffordaxioms/CZ.tikz}} = \scalebox{0.9}{\input{qutritcliffordaxioms/CZdecomp.tikz}}
      \end{equation}
    \end{minipage}\vspace{-0.6em}
    \begin{minipage}[c]{0.36\linewidth}\centering
      \begin{equation}\tag{B}\label{qt-swapdecomp}
        \scalebox{0.9}{\input{cliffordaxioms/SWAPCNOT.tikz}} = \scalebox{0.9}{\input{qutritcliffordaxioms/CNOT-NOTC-NOTC-K.tikz}}
      \end{equation}
    \end{minipage}
    \begin{minipage}[c]{0.29\linewidth}\centering
      \begin{equation}\tag{I}\label{qt-I}
        \scalebox{0.9}{\input{qutritcliffordaxioms/CNOT23CNOT12CNOT13.tikz}} = \scalebox{0.9}{\input{qutritcliffordaxioms/CNOT12CNOT23.tikz}}
      \end{equation}
    \end{minipage}
    \end{minipage}}
    \caption{Qutrit Clifford presentation \(\QCqutritclifford\), using the qutrit scalar convention and the qutrit versions of the phase, controlled-addition, controlled-\(Z\), swap-decomposition, and three-wire laws.}
    \label{fig:qutritcliffordaxioms}
\end{figure}
\begin{figure}[ht]
    \centering
    \fbox{\begin{minipage}{0.975\textwidth}
        \centering
        \vspace{-1em}
        \begin{minipage}[c]{0.22\textwidth}
            \centering
            \begin{equation}\tag{\omegarulabel{8}}\label{t-w8}
                \scalebox{0.9}{\input{gates/wc.tikz}}^{\otimes 8} = \scalebox{0.9}{\input{gates/empty.tikz}}
            \end{equation}
        \end{minipage}
        \begin{minipage}[c]{0.46\textwidth}
            \centering
            \begin{equation}\tag{\rulabel{T^8}}\label{t-Tpow8}
                \scalebox{0.9}{\input{cliffordplustaxioms/Tpow8.tikz}} = \scalebox{0.9}{\input{gates/Id.tikz}}
            \end{equation}
        \end{minipage}\vspace{-0.6em}
        \begin{minipage}[c]{0.24\textwidth}
            \centering
            \begin{equation}\tag{\rulabel{H^2}}\label{t-H2}
                \scalebox{0.9}{\input{cliffordaxioms/HH.tikz}} = \scalebox{0.9}{\input{gates/Id.tikz}}
            \end{equation}
        \end{minipage}
        \begin{minipage}[c]{0.52\textwidth}
            \centering
            \begin{equation}\tag{E}\label{t-E}
                \scalebox{0.9}{\input{cliffordaxioms/HSH.tikz}} = \scalebox{0.9}{\input{cliffordaxioms/wSdaggerHSdagger.tikz}}
            \end{equation}
        \end{minipage}\vspace{-0.6em}
        \begin{minipage}[c]{0.34\textwidth}
            \centering
            \begin{equation}\tag{TX}\label{t-TX}
                \scalebox{0.9}{\input{cnotdihedral/R11L.tikz}}
            =
            \scalebox{0.9}{\input{cnotdihedral/R11R.tikz}}
          \end{equation}
        \end{minipage}
        \begin{minipage}[c]{0.34\textwidth}
            \centering
            \begin{equation}\tag{CPh}\label{t-Ct}
                \scalebox{0.9}{\input{cliffordplustaxioms/CNOTTCNOT.tikz}} = \scalebox{0.9}{\input{cliffordplustaxioms/T.tikz}}
            \end{equation}
        \end{minipage}\vspace{-0.6em}
        \begin{minipage}[c]{0.29\textwidth}
            \centering
            \begin{equation}\tag{B}\label{t-B}
                \scalebox{0.9}{\input{cliffordaxioms/CNOTNOTC.tikz}} = \scalebox{0.9}{\input{cliffordaxioms/SWAPCNOT.tikz}}
            \end{equation}
        \end{minipage}
        \begin{minipage}[c]{0.43\textwidth}
            \centering
            \begin{equation}\tag{CZ}\label{t-CZ}
                \scalebox{0.9}{\input{cliffordaxioms/HCNOTH.tikz}} = \scalebox{0.9}{\input{cliffordaxioms/CZ.tikz}}
            \end{equation}
        \end{minipage}\vspace{-0.6em}
        \begin{minipage}[c]{0.32\textwidth}
            \centering
            \begin{equation}\tag{CSH}\label{t-CsCh}
                \scalebox{0.9}{\input{cliffordplustaxioms/CsCh.tikz}} = \scalebox{0.9}{\input{cliffordplustaxioms/ChCs.tikz}}
            \end{equation}
        \end{minipage}
        \begin{minipage}[c]{0.50\textwidth}
            \centering
            \begin{equation}\tag{\rulabel{HT^2}}\label{t-HTHT}
                \scalebox{0.9}{\input{cliffordplustaxioms/HTHT.tikz}} = \scalebox{0.9}{\input{cliffordplustaxioms/THTH.tikz}}
            \end{equation}
        \end{minipage}\vspace{-0.6em}
        \begin{minipage}[c]{0.65\textwidth}
            \centering
            \begin{equation}\tag{HTH}\label{t-C20}
                \scalebox{0.9}{\input{cliffordplustaxioms/C20l.tikz}} = \scalebox{0.9}{\input{cliffordplustaxioms/C20r.tikz}}
            \end{equation}
        \end{minipage}
    \end{minipage}}
    \caption{Qubit Clifford\(+T\) presentation \(\QCcliffordplust\), combining the Clifford core with the \(T\)-phase, controlled-phase, controlled-Hadamard, and \(HTH\)-type relations needed for the two-qubit completeness transfer.}
    \label{fig:qubitcliffordplustaxioms}
\end{figure}
\begin{figure}[ht]
    \centering
    \fbox{\begin{minipage}{0.975\textwidth}
        \centering
        \vspace{-1em}
        \begin{minipage}[c]{0.22\textwidth}
            \centering
            \begin{equation}\tag{\omegarulabel{8}}\label{cs-w8}
                \scalebox{0.9}{\input{gates/wc.tikz}}^{\otimes 8} = \scalebox{0.9}{\input{gates/empty.tikz}}
            \end{equation}
        \end{minipage}
        \begin{minipage}[c]{0.24\textwidth}
            \centering
            \begin{equation}\tag{\rulabel{H^2}}\label{cs-H2}
                \scalebox{0.9}{\input{cliffordaxioms/HH.tikz}} = \scalebox{0.9}{\input{gates/Id.tikz}}
            \end{equation}
        \end{minipage}
        \begin{minipage}[c]{0.32\textwidth}
            \centering
            \begin{equation}\tag{\rulabel{S^4}}\label{cs-S4}
                \scalebox{0.9}{\input{cliffordaxioms/SSSS.tikz}} = \scalebox{0.9}{\input{gates/Id.tikz}}
            \end{equation}
        \end{minipage}\vspace{-0.6em}
        \begin{minipage}[c]{0.51\textwidth}
            \centering
            \begin{equation}\tag{E}\label{cs-E}
                \scalebox{0.9}{\input{cliffordaxioms/HSH.tikz}} = \scalebox{0.9}{\input{cliffordaxioms/wSdaggerHSdagger.tikz}}
            \end{equation}
        \end{minipage}\vspace{-0.6em}
        \begin{minipage}[c]{0.41\textwidth}
            \centering
            \begin{equation}\tag{CPh}\label{cs-C}
                \scalebox{0.9}{\input{cliffordpluscsaxioms/CSpow4withS.tikz}} = \scalebox{0.9}{\input{cliffordaxioms/S.tikz}}
            \end{equation}
        \end{minipage}
        \begin{minipage}[c]{0.31\textwidth}
            \centering
            \begin{equation}\tag{CSr}\label{cs-CSrev}
                \scalebox{0.9}{\input{cliffordpluscsaxioms/CS.tikz}} = \scalebox{0.9}{\input{cliffordpluscsaxioms/CSrev.tikz}}
            \end{equation}
        \end{minipage}\vspace{-0.6em}
        \begin{minipage}[c]{0.28\textwidth}
            \centering
            \begin{equation}\tag{B}\label{cs-B}
                \scalebox{0.9}{\input{cliffordaxioms/CNOTNOTC.tikz}} = \scalebox{0.9}{\input{cliffordaxioms/SWAPCNOT.tikz}}
            \end{equation}
        \end{minipage}
        \begin{minipage}[c]{0.59\textwidth}
            \centering
            \begin{equation}\tag{XCS}\label{cs-XCS}
                \scalebox{0.9}{\input{cliffordpluscsaxioms/X2TopLeft.tikz}}=\scalebox{0.9}{\input{cliffordpluscsaxioms/X2TopRight.tikz}}
            \end{equation}
        \end{minipage}\vspace{-0.6em}
        \begin{minipage}[c]{0.49\textwidth}
            \centering
            \begin{equation}\tag{CE}\label{cs-SHCHC}
                \scalebox{0.9}{\input{cliffordpluscsaxioms/HTopLeft.tikz}}=\scalebox{0.9}{\input{cliffordpluscsaxioms/HTopRight.tikz}}
            \end{equation}
        \end{minipage}
        \begin{minipage}[c]{0.29\textwidth}
            \centering
            \begin{equation}\label{cs-I}\tag{I}\scalebox{0.9}{\input{qutritcliffordaxioms/CNOT23CNOT12CNOT13.tikz}} = \scalebox{0.9}{\input{qutritcliffordaxioms/CNOT12CNOT23.tikz}}\end{equation}
        \end{minipage}\vspace{-0.6em}
        \begin{minipage}[c]{0.64\textwidth}
            \centering
            \begin{equation}\label{cs-U}\tag{\rulabel{SH_k}}\scalebox{0.9}{\input{cliffordpluscsaxioms/CSmonsterL.tikz}} = \scalebox{0.9}{\input{cliffordpluscsaxioms/CSmonsterR.tikz}}\end{equation}
        \end{minipage}
    \end{minipage}}
    \caption{Qubit Clifford\(+CS\) presentation \(\QCcliffordpluscs\), combining the Clifford phase laws with controlled-\(S\) commutation and the two three-wire relations used for the three-qubit transfer; in \cref{cs-U}, the dashed box repeats the enclosed subcircuit \(k\) times for \(k\in\{0,1,2,3\}\).}
    \label{fig:qubitcliffordpluscsaxioms}
\end{figure}
\begin{figure}[ht]
    \centering
    \fbox{%
        \begin{minipage}{0.98\textwidth}
        \centering
        \vspace{-1em}
        \begin{minipage}[c]{0.22\textwidth}
            \centering
            \begin{equation}\tag{\omegarulabel{8}}\label{cnot-R10}
                \scalebox{0.9}{\input{gates/wc.tikz}}^{\otimes 8} = \scalebox{0.9}{\input{gates/empty.tikz}}
            \end{equation}
        \end{minipage}
        \begin{minipage}[c]{0.27\textwidth}
            \centering
            \begin{equation}\tag{\rulabel{X^2}}\label{cnot-R1}
            \scalebox{0.9}{\input{cnotdihedral/R1L.tikz}}
            =
            \scalebox{0.9}{\input{gates/Id.tikz}}
            \end{equation}
        \end{minipage}
        \begin{minipage}[c]{0.24\textwidth}
            \centering
            \begin{equation}\tag{\rulabel{T^8}}\label{cnot-R7}
            \scalebox{0.9}{\input{cnotdihedral/R7L.tikz}}
            =
            \scalebox{0.9}{\input{gates/Id.tikz}}
            \end{equation}
        \end{minipage}\vspace{-0.6em}
        \begin{minipage}[c]{0.34\textwidth}
            \centering
            \begin{equation}\tag{TX}\label{cnot-R11}
            \scalebox{0.9}{\input{cnotdihedral/R11L.tikz}}
            =
            \scalebox{0.9}{\input{cnotdihedral/R11R.tikz}}
            \end{equation}
        \end{minipage}
        \begin{minipage}[c]{0.34\textwidth}
            \centering
            \begin{equation}\tag{XC}\label{new-R3}
            \scalebox{0.9}{\input{identities/wCNOT-00.tikz}}
            =
            \scalebox{0.9}{\input{identities/wCNOT-16.tikz}}
            \end{equation}
        \end{minipage}
        \begin{minipage}[c]{0.29\textwidth}
            \centering
            \begin{equation}\tag{B}\label{new-R5}
            \scalebox{0.9}{\input{cliffordaxioms/CNOTNOTC.tikz}}
            =
            \scalebox{0.9}{\input{cliffordaxioms/SWAPCNOT.tikz}}
            \end{equation}
        \end{minipage}\vspace{-0.6em}
        \begin{minipage}[c]{0.35\textwidth}
            \centering
            \begin{equation}\tag{ZC}\label{new-R8}
            \scalebox{0.9}{\input{identities/zCNOT-00.tikz}}
            =
            \scalebox{0.9}{\input{identities/zCNOT-05.tikz}}
            \end{equation}
        \end{minipage}
        \begin{minipage}[c]{0.29\textwidth}
            \centering
            \begin{equation}\tag{CPh}\label{cnot-R12}
            \scalebox{0.9}{\input{cliffordplustaxioms/CNOTTCNOT.tikz}}
            =
            \scalebox{0.9}{\input{cliffordplustaxioms/T.tikz}}
            \end{equation}
        \end{minipage}\vspace{-0.6em}
        \begin{minipage}[c]{0.29\textwidth}
            \centering
            \begin{equation}\tag{I}\label{cnot-R6}
            \scalebox{0.9}{\input{cnotdihedral/R6L.tikz}}
            =
            \scalebox{0.9}{\input{cnotdihedral/R6R.tikz}}
            \end{equation}
        \end{minipage}
        \begin{minipage}[c]{0.48\textwidth}
            \centering
            \begin{equation}\tag{\rulabel{C^2T}}\label{cnot-R9}
            \scalebox{0.9}{\input{cnotdihedral/R9L.tikz}}
            =
            \scalebox{0.9}{\input{cnotdihedral/R9R.tikz}}
            \end{equation}
        \end{minipage}\vspace{-0.6em}
        \begin{minipage}[c]{0.84\textwidth}
            \centering
            \begin{equation}\tag{\rulabel{C^3T}}\label{cnot-R13}
            \scalebox{0.82}{\input{cnotdihedral/R13L.tikz}}
            =
            \scalebox{0.82}{\input{cnotdihedral/R13R.tikz}}
            \end{equation}
        \end{minipage}
        \end{minipage}%
    }
    \caption{CNOT-dihedral presentation \(\QCdihedral\), with the \(X\), \(T\), CNOT, phase-commutation, \(C^2T\), and \(C^3T\) relations used to recover the omitted source axioms.}
    \label{fig:qubitdihedralaxioms}
\end{figure}
\section{Completeness: from PROs to PROPs}
\label{sec:completeness}

For each fragment~\(\bullet\) of \cref{sec:circuits-relations} we have a presented PROP
\(\QCirc_\bullet=\cat{P}_\bullet/\QC_\bullet\) and a strict symmetric monoidal interpretation
\(\interp{\cdot}_\bullet:\QCirc_\bullet\to\cat{Qubit}_\bullet\) (or
\(\cat{Qutrit}_{\mathit{Cliff3}}\) in the qutrit case).  Completeness means faithfulness of this
interpretation.  Complete equational theories for these fragments already exist in the
literature
\cite{SelingerStabilizer,SelingerRealStabilizer,qutritclifford,SelingerCliffordPlusT,SelingerCliffordPlusCS,AmyCNOTDihedral}.
Those theorems supply the source completeness results: their normal forms prove faithfulness for
older source presentations, and the transfer argument in this section checks that the source syntax
and our smaller figure-defined target syntax present the same semantic subPROP.

The transfer has to align three conventions before the cited normal forms can be used for our
quotient PROPs.  If a source theorem is stated for a PRO, \cref{lem:fromprotoprop} adds the
structural symmetries needed to read it as a PROP theorem.  For qutrit Clifford and
Clifford\(+CS\), the source and target scalar subgroups differ: the source presentations make fewer
global phases explicit, which is harmless for the original proofs but does not match the standard
strict matrices fixed in \cref{sec:circuits-relations}.  The scalar-refinement step adjoins the
missing roots of unity, so source and target have the same strict semantics.  Encoding and decoding
are then identity-on-objects PROP morphisms, in the sense of \cref{def:prop-morphism}, that compare
the source generators and relations with ours; once these maps preserve semantics and decode source
axioms to derivable target equations, \cref{th:completeness} transports faithfulness.

\subsection{Scalar refinement}
\label{subsec:scalar-refinement}

Because our semantics uses strict unitary equality rather than projective equality, global phases
are visible.  Equivalently, the scalar subgroup \(S(\cat{C})=\cat{C}(0,0)\) of the semantic PROP is
part of the data that must match between source and target presentations.
For our purposes, a scalar mismatch is therefore a mismatch between presentation conventions: the
source convention makes fewer global phases explicit, while our convention keeps the standard gate
matrices and records the extra phases as scalars.

Only two source completeness results require scalar adjustment.  The qutrit Clifford presentation of
\cite{qutritclifford} uses an order-\(6\) scalar subgroup whereas our semantics uses order \(12\), and
the Clifford\(+CS\) presentation of \cite{SelingerCliffordPlusCS} uses order \(4\) whereas our qubit
semantics uses order \(8\).  In both cases we appeal to the scalar-refinement construction of
Appendix~\ref{app:scalar-refinement}, which adjoins the missing scalar generator and is conservative
by \cref{lem:scalar-refinement}.  Source and target can then be compared inside the same semantic
subPROP, while the refinement remains conservative on the old source signature.

\subsection{Encoding, decoding, and transfer of completeness}

The five fragments handled by encoding and decoding reduce to one generator-and-relation
comparison.

Fix one encoding/decoding transfer instance.  Here \(\cat{P}\) and \(\mathcal{R}\) denote the
target free PROP and target relations, while \(\cat{P}^{\mathrm{src}}\) and
\(\mathcal{R}^{\mathrm{src}}\) denote the corresponding source free PROP and source relations,
after PRO-to-PROP lifting and scalar refinement where those steps apply.

\begin{definition}\label{def:encdec}
An encoding/decoding pair consists of a target-to-source PROP morphism
\(E:\cat{P}\to\cat{P}^{\mathrm{src}}\) and a source-to-target PROP morphism
\(D:\cat{P}^{\mathrm{src}}\to\cat{P}\) such that
\(\interp{C}=\interp{E(C)}_{\mathrm{src}}\) for every target circuit \(C\in\cat{P}\) and
\(\interp{C'}_{\mathrm{src}}=\interp{D(C')}\) for every source circuit
\(C'\in\cat{P}^{\mathrm{src}}\).
\end{definition}

After importing a source completeness theorem, the comparison is local: it involves only the chosen
generators and source relations.  A typical mismatch is that a source presentation takes
controlled-\(Z\) as primitive, whereas our target signature takes CNOT.  The decoder replaces source
controlled-\(Z\) by its fixed expansion through CNOT and Hadamards; the target-generator check then
proves, for example, that decoding the encoded CNOT derives CNOT again, while the source-relation
check verifies each source equation after the same replacement.  The non-identity parts of the maps
used here are recorded in \cref{tab:ED-clifford-family}.

The transfer therefore has a finite checklist: the round-trip equality \(D(E(g))=g\) for each
target generator \(g\), and the decoded equality \(D(L)=D(R)\) for each source axiom \(L=R\).
This is the presentation-level analogue of Reidemeister--Schreier rewriting
\cite{Reidemeister1927,Schreier1927}: instead of computing a subgroup presentation, we compare two
chosen presentations by following the images of their generators and defining relations.

\begin{lemma}\label{lem:decenc}
Assume \(\cat{P}\) is presented by generators \(\Sigma\) and relations \(\mathcal{R}\).
If for every non-structural generator \(g\in\Sigma\),
\(
  \cat{P}/\mathcal{R}\vdash D(E(g)) = g,
\)
then for every circuit \(C\) built from \(\Sigma\),
\(
  \cat{P}/\mathcal{R}\vdash D(E(C)) = C.
\)
\end{lemma}
\begin{proof}
By structural induction on \(C\), using that \(E\) and \(D\) preserve \(\circ\), \(\otimes\), and the
structural morphisms.
\end{proof}

\begin{lemma}\label{lem:decrelations}
Assume that for every axiom \(L=R\) in \(\mathcal{R}^{\mathrm{src}}\),
\(
  \cat{P}/\mathcal{R}\vdash D(L) = D(R).
\)
Then \(D\) induces a PROP morphism
\(
  \overline{D}:\cat{P}^{\mathrm{src}}/\mathcal{R}^{\mathrm{src}}\to\cat{P}/\mathcal{R},
\)
and any derivable equality in the source quotient transports to a derivable equality between its
decodings in the target theory.
\end{lemma}
\begin{proof}
The hypothesis says exactly that \(D\) equalises all source relations, hence factors through the
quotient.
\end{proof}

\begin{theorem}\label{th:completeness}
Assume \(\cat{P}^{\mathrm{src}}/\mathcal{R}^{\mathrm{src}}\) is complete for
\(\interp{\cdot}_{\mathrm{src}}\), and let \((E,D)\) be an encoding/decoding pair satisfying the
hypotheses of \cref{lem:decenc,lem:decrelations}.  Then \(\cat{P}/\mathcal{R}\) is complete for
\(\interp{\cdot}\).
\end{theorem}
\begin{proof}
If target circuits \(C_1,C_2\) have the same semantics, then semantic preservation of \(E\) gives
\(\interp{E(C_1)}_{\mathrm{src}}=\interp{E(C_2)}_{\mathrm{src}}\).  Source completeness yields
\(E(C_1)=E(C_2)\) in
\(\cat{P}^{\mathrm{src}}/\mathcal{R}^{\mathrm{src}}\); applying the quotient morphism
\(\overline{D}\) from \cref{lem:decrelations} gives \(D(E(C_1))=D(E(C_2))\) in
\(\cat{P}/\mathcal{R}\), and \cref{lem:decenc} rewrites both sides to \(C_1=C_2\).
\end{proof}

The same proof gives the bounded form used in \cref{cor:clifford-family-complete}: if the source
interpretation is faithful on all hom-sets \(n\to m\) with \(n,m\le k\), then the target
interpretation is faithful on the same range, because \(E\) and \(D\) are identity-on-objects PROP
morphisms.

\subsection{Fragment-wise completeness}

Applying the transfer to the five encoding/decoding fragments and the direct CNOT-dihedral
argument gives the completeness bounds used by the minimality theorem.

\begin{theorem}\label{cor:clifford-family-complete}
\begin{enumerate}
\item \(\QCirc_{\mathit{Cliff}}\) is complete for \(\cat{Qubit}_{\mathit{Cliff}}\).
\item \(\QCirc_{\mathit{RCliff}}\) is complete for \(\cat{Qubit}_{\mathit{RCliff}}\).
\item \(\QCirc_{\mathit{Cliff3}}\) is complete for \(\cat{Qutrit}_{\mathit{Cliff3}}\).
\item \(\QCirc_{\mathit{CliffT}}\) is complete up to \(2\) qubits for \(\cat{Qubit}_{\mathit{CliffT}}\).
\item \(\QCirc_{\mathit{CliffCS}}\) is complete up to \(3\) qubits for \(\cat{Qubit}_{\mathit{CliffCS}}\).
\item \(\QCirc_{\mathit{CNOTdihe}}\) is complete for \(\cat{Qubit}_{\mathit{CNOTdihe}}\).
\end{enumerate}
\end{theorem}

\begin{proof}
For the five fragments using encoding/decoding, start from the corresponding cited source
completeness theorem; in the qutrit Clifford and Clifford\(+CS\) cases, first pass to the
scalar-refined source presentation supplied by \cref{lem:scalar-refinement}.  The arity bounds in
the Clifford+T and Clifford+CS items are the bounds inherited from those source results.

The appendices record the data needed to invoke \cref{th:completeness}: the earlier rule sets in
\cref{appendix:oldgraphicallanguage1,appendix:oldgraphicallanguage2,appendix:oldgraphicallanguage3,appendix:oldgraphicallanguage4,appendix:oldgraphicallanguage5,appendix:oldgraphicallanguage6},
the encoding/decoding maps in \cref{tab:ED-clifford-family}, the generator checks for
\cref{lem:decenc} in Appendix~\ref{app:proofs-lem1}, and the source-axiom checks for
\cref{lem:decrelations} in Appendix~\ref{appendix:decodage}.  With these ingredients fixed, the
remaining derivability obligations are exactly the finite checks of the transfer lemmas; no further
normal-form argument is used here.

For CNOT-dihedral, use the completeness theorem of \cite{AmyCNOTDihedral} directly: no
encoding/decoding is needed, and Appendix~\ref{app:CNOT-dihedral} derives, in \(\QCdihedral\),
each source axiom removed from our simplified presentation.
\end{proof}
\section{Minimality}\label{sec:minimality}

Minimality asks which axioms of the six figure-defined presentations are genuinely necessary for
completeness, and the full and bounded statements in \cref{thm:minimality-all} are certified by
separating interpretations.  Throughout, fix a fragment \(\bullet\) with free PROP
\(\cat{P}_\bullet\) and finite rule set \(\QC_\bullet\).  The arity bounds in the theorem are dictated
by separator coverage: a bounded minimality claim is made only when the table contains a
non-\(\noseparator\) separator for every axiom in that truncation.  When this independence is read as
necessity for completeness, the Clifford\(+T\) and Clifford\(+CS\) cases are understood relative to the
bounded completeness results imported in \cref{sec:completeness}.

\subsection{Independence and minimality}

To test whether an axiom \(\rho\) is needed, remove it from the rule set and ask whether its two sides
are still equal in the quotient by the remaining rules.  We write
\(\QC_\bullet\setminus\{\rho\}\) for those remaining axioms, and
\(\QCirc_\bullet^{-\rho}\coloneqq \cat{P}_\bullet/(\QC_\bullet\setminus\{\rho\})\) for the reduced
presentation.

\begin{definition}\label{def:minimality}
Let \(\rho\) be an axiom \(C_1=C_2\) in \(\QC_\bullet\) (with \(C_1,C_2:n\to n\)).
We say that \(\rho\) is independent of the other axioms of the same fragment if
\(
  \QCirc_\bullet^{-\rho}\not\vdash C_1=C_2.
\)
We say that \(\QC_\bullet\) is minimal if every \(\rho\in\QC_\bullet\) is independent.
\end{definition}

Because the presentations are sound for their intended semantics, an axiom derivable from the
others cannot be necessary for completeness: removing it leaves the induced congruence unchanged.
Conversely, in any arity range where the reduced presentation were still complete, soundness would
force the removed equation to remain derivable.  Thus independence is the syntactic test used here,
and the separators in \cref{fig:minimality-all} instantiate it by occurrence detectors, counting
models, projective substitutions, and determinant phases.

Since derivations preserve arity by \cref{rem:arity-preservation}, an \(n\)-wire axiom can be tested
inside the subsystem generated by axioms of arity at most \(n\).  This gives the bounded notion used
when the separator table contains separators only in a finite arity range.

\begin{definition}\label{def:bounded-minimality}
Let \(k\in\mathbb{N}\).  Define
\(
  \QC_{\bullet,\le k}
   \coloneqq 
  \{ \rho\in\QC_\bullet \mid \rho \text{ has type } n\to n \text{ with } n\le k \}.
\)
We say that \(\QC_\bullet\) is minimal up to \(k\) wires if every \(\rho\in\QC_{\bullet,\le k}\) is independent
relative to \(\QC_{\bullet,\le k}\) (i.e.\ independent in the sense of \cref{def:minimality} after replacing
\(\QC_\bullet\) by \(\QC_{\bullet,\le k}\)).
\end{definition}

\subsection{Separation by alternative interpretations}

Our independence proofs are semantic: to show that an axiom \(\rho\) is not derivable from the others,
we construct a PROP-valued model that satisfies all remaining axioms but not \(\rho\).  The underlying
countermodel principle for equational logic is that derivable equations hold in every model of the
theory, so one model of the reduced theory that violates \(\rho\) rules out a derivation.  This
principle goes back at least to Birkhoff's correspondence between equational theories and varieties
of algebras \cite{Birkhoff1935} and is explicit in modern accounts of equational completeness
\cite[Thm.~14.19]{BurrisSankappanavar1981}.  In the quantum-circuit literature, the same proof
pattern is used in \cite{ClementDelormePerdrix2023}, which proves a minimal and complete
equational theory for general quantum circuits.

\begin{lemma}\label{lem:separation}
Let \(\rho:C_1=C_2\) be an axiom in \(\QC_\bullet\).  Assume there exist a PROP
\(\cat{P}_{\mathrm{alt}}\) and a PROP morphism \(F:\cat{P}_\bullet\to\cat{P}_{\mathrm{alt}}\) such that
\(F(L)=F(R)\) for every \((L=R)\in\QC_\bullet\setminus\{\rho\}\), but \(F(C_1)\neq F(C_2)\).  Then
\(\rho\) is independent.
\end{lemma}

\begin{proof}
Since \(F\) equalises every axiom in \(\QC_\bullet\setminus\{\rho\}\), it factors through the quotient
\(\QCirc_\bullet^{-\rho}\), yielding a PROP morphism
\(\overline{F}:\QCirc_\bullet^{-\rho}\to\cat{P}_{\mathrm{alt}}\).
If \(\QCirc_\bullet^{-\rho}\vdash C_1=C_2\) then applying \(\overline{F}\) gives \(F(C_1)=F(C_2)\), a contradiction.
\end{proof}

Because morphisms of the free PROP are already quotiented by coherence, once \(F\) is a PROP morphism
it automatically respects all structural equalities, so the separator check is purely
fragment-specific.

\subsection{Families of separating interpretations}

\Cref{fig:minimality-all} is the finite certificate for the independence proofs: each
non-\(\noseparator\) row supplies one separator to be checked against the other axioms in the
relevant arity range, while the \(\noseparator\) rows mark exactly the gaps that prevent full
minimality claims for qutrit Clifford, Clifford\(+T\), and Clifford\(+CS\).  The separators fall into
four families--counting models, occurrence detectors, projective substitutions, and
determinant-phase interpretations--so the proof remains uniform even though the fragments differ.
In displayed rule sets and tables, abbreviated gates stand for fixed morphisms of the free PROP, and
every separating interpretation is applied to the expanded morphisms.

\emph{Counting and occurrence detectors.}
Counting models record selected generator multiplicities or swap parity, while occurrence detectors
record whether a chosen generator appears; both kinds of separator land in endomorphism-only PROPs,
so there are no morphisms between different arities.

\begin{definition}\label{def:endM}
Let \((M,\oplus,0)\) be a commutative monoid and let \(\epsilon\in M\) satisfy
\(\epsilon\oplus\epsilon=0\).  Define \(\cat{End}_{\epsilon}(M)\) to be the PROP with objects
\(\mathbb{N}\) and hom-sets
\(
  \cat{End}_{\epsilon}(M)(n,n)\coloneqq M \; (n\in\mathbb{N}),
\)
\(
  \cat{End}_{\epsilon}(M)(n,m)\coloneqq \varnothing \; (n\neq m).
\)
Composition and tensor are given by \(\oplus\), identities by \(0\), and the basic symmetry
\(\sigma_{1,1}\) is interpreted as \(\epsilon\).
\end{definition}

\begin{remark}
The condition \(\epsilon\oplus\epsilon=0\) implies \(\sigma_{m,n}\circ\sigma_{n,m}=\id_{n+m}\).
All other PROP coherence axioms hold because every structural map is interpreted as an element of
the commutative monoid \(M\), and both \(\circ\) and \(\otimes\) are interpreted as \(\oplus\).
\end{remark}

\begin{definition}\label{def:count-has}
For \(m\ge 2\) define
\(
  \cat{Count}_{m}\coloneqq \cat{End}_{0}(\mathbb{Z}_m,+,0)
\)
and
\(
  \cat{Has}\coloneqq \cat{End}_{0}(\mathbb{B},\vee,0),
\)
with \(\mathbb{B}=\{0,1\}\).
When we want to record the parity of the structural swap, we also use
\(
  \cat{Count}^{\mathrm{swap}}_{2}\coloneqq \cat{End}_{1}(\mathbb{Z}_2,+,0),
\)
so that the basic symmetry \(\sigma_{1,1}\) is interpreted as \(1\in\mathbb{Z}_2\).
\end{definition}

\begin{definition}\label{def:counting-notation}
Let \(m\ge 2\).  Given a finite list \((g_1,\ldots,g_r)\) of generators in \(\Sigma_\bullet\), with
repetitions counted with multiplicity, write
\(
  \#\{g_1,\ldots,g_r\}_{[m]}:\cat{P}_\bullet\to\cat{Count}_m
\)
for the unique PROP morphism sending each generator \(g\) to the number of its occurrences in the
list modulo \(m\).
\end{definition}

When \(m=2\), we use the symbol \(\gSWAP\) inside the braces only as a mnemonic that swaps are counted:
\(
  \#\{g_1,\ldots,g_r,\gSWAP\}_{[2]}:\cat{P}_\bullet\to\cat{Count}^{\mathrm{swap}}_2.
\)
In particular, \(\#\{\gSWAP\}_{[2]}\) denotes the morphism that sends every non-structural generator
to \(0\) and counts each structural swap as \(1\).

\begin{definition}\label{def:has-functors}
For a generator \(g\in\Sigma_\bullet\), write
\(
  ?g:\cat{P}_\bullet\to\cat{Has}
\)
for the unique PROP morphism with \(?g(g)=1\) and \(?g(h)=0\) for every generator \(h\neq g\).
Equivalently, for any circuit \(C\), one has \(?g(C)=1\) if and only if \(g\) occurs at least once in
\(C\).
\end{definition}

\begin{example}\label{ex:clifford-t-h2-separator}
In the one-wire Clifford\(+T\) truncation, the detector \(?\gH\) separates the axiom
\eqref{t-H2} from the remaining rules.  The equations \eqref{t-w8} and \eqref{t-Tpow8} contain no
\(\gH\) on either side, while \eqref{t-E} contains an occurrence of \(\gH\) on both sides.  The same is
true of \eqref{t-TX} after expanding the displayed shortcut \(\gX\).  Thus \(?\gH\) sends each
remaining equation in \(\QC_{\mathit{CliffT},\le 1}\) to either \(0=0\) or \(1=1\).  On the removed axiom,
however, \(?\gH(\gHsquare)=1\) and \(?\gH(\id_1)=0\), so \eqref{t-H2} is independent by
\cref{lem:separation}.
\end{example}

\emph{Projective substitutions.}
These separators alter one generator while quotienting only by global phase, which separates rows
whose two sides have the same generator multiplicities but differ after a projective change of one
generator.

\begin{definition}\label{def:projective-quotient}
For \(d\in\{2,3\}\), let \(\cat{Qudit}_d^{\sim}\) be the projective quotient of
\(\cat{Qudit}_d\) from \cref{def:qubit-qutrit-props}: its endomorphism hom-set
is \(\cat{Qudit}_d(n,n)/\mathrm{U}(1)\), where \(\mathrm{U}(1)\) acts by global
phase, and its non-endomorphism hom-sets remain empty.  Write
\(\Pi_d:\cat{Qudit}_d\to\cat{Qudit}_d^{\sim}\) for the quotient PROP morphism.
\end{definition}

\begin{definition}\label{def:projective-interpretation}
For a fragment \(\bullet\), set \(d(\bullet)=3\) for the qutrit Clifford
fragment and \(d(\bullet)=2\) for the qubit fragments.  Its usual projective
interpretation on the free presentation is
\(\interp{\cdot}_{\bullet,\sim}\coloneqq
\Pi_{d(\bullet)}\circ\interp{\cdot}^{\mathrm{free}}_\bullet\), regarded as a
PROP morphism \(\cat{P}_\bullet\to\cat{Qudit}_{d(\bullet)}^{\sim}\), where
\(\interp{\cdot}^{\mathrm{free}}_\bullet:\cat{P}_\bullet\to
\cat{Qudit}_{d(\bullet)}\) is the standard unitary interpretation before
quotienting by the fragment equations.
\end{definition}

\begin{definition}\label{def:projective-subst}
Given a fragment \(\bullet\), a generator \(g:n\to n\), and a circuit \(C:n\to n\) of that fragment, write
\(\interp{\cdot}_{\bullet,\sim}^{g:=C}\) for the unique PROP morphism
\(\cat{P}_\bullet\to\cat{Qudit}_{d(\bullet)}^{\sim}\) sending \(g\) to the already-evaluated
class \(\interp{C}_{\bullet,\sim}\) and every generator \(h\neq g\) to
\(\interp{h}_{\bullet,\sim}\).  Displayed abbreviations in \(C\) are expanded before this
interpretation.
\end{definition}

Thus a table entry \(\psubst{g}{C}\) means that every generator is interpreted projectively as
usual except \(g\), whose value is replaced by the projective value of \(C\).

\emph{Determinant phases.}
These separators record a scaled determinant phase for the qubit rows, chosen so as to be strictly
monoidal on a bounded number of wires.

\begin{definition}\label{def:argdetk}
Fix a qubit fragment, let \(k\ge 2\), and write \(\arg\det(U)\in\mathbb{R}/2\pi\mathbb{Z}\) for the
phase of \(\det(U)\).  Set
\(
  \epsilon_k\coloneqq 2^{k-2}\pi \bmod 2\pi,
\)
so \(\epsilon_2=\pi\) and \(\epsilon_k=0\) for \(k>2\).

Define a PROP morphism
\(
  \arg\det_k:\cat{P}_\bullet\to \cat{End}_{\epsilon_k}(\mathbb{R}/2\pi\mathbb{Z},+,0)
\)
on generators \(g:n\to n\) by
\(
  (\arg\det_k)(g)\coloneqq 2^{k-n}\arg\det(\interp{g}) \bmod 2\pi,
\)
and extend it to all circuits using preservation of \(\circ\), \(\otimes\), and symmetries.
\end{definition}

\begin{remark}
The scaling factor \(2^{k-n}\) is chosen so that \(\arg\det_k\) is compatible with tensoring on
arities \(\le k\), using \(\det(U\otimes V)=\det(U)^{2^m}\det(V)^{2^n}\) for \(U\) on \(n\) qubits and \(V\)
on \(m\) qubits.  The choice of \(\epsilon_k\) makes the image of the structural swap
\(\sigma_{1,1}\) coincide with its determinant phase after scaling: it contributes \(\pi\) when \(k=2\)
and becomes \(0\) for \(k>2\).
For equations on more than \(k\) wires, \(\arg\det_k\) is still used only as the generator-defined
PROP morphism above; the row-wise separator check is a check in
\(\cat{End}_{\epsilon_k}(\mathbb{R}/2\pi\mathbb{Z},+,0)\), not an assertion that the displayed
value is the scaled determinant of the whole higher-arity unitary.
\end{remark}

For each non-\(\noseparator\) row of \cref{fig:minimality-all}, the relation named on the left is
removed and the separator on the right equalises the other axioms in the relevant arity range while
distinguishing the removed relation.  The table records the separator rather than every value in the
finite check.  Detector and counting rows use the relevant commutative monoid,
projective substitution rows use the projective quotient from \cref{def:projective-interpretation},
and determinant rows use the scaled determinant phase from \cref{def:argdetk}.  When a higher-arity
separator is used in \cref{thm:minimality-all}, its check is against all axioms in the relevant arity
truncation, including same-arity rows marked \(\noseparator\) in \cref{fig:minimality-all}.

\begin{figure}[ht]
  \centering
  \begin{subfigure}[t]{0.31\textwidth}
    \centering
    \caption{Clifford}
    \begin{tabular}{|c|c|}
      \hline
      \textbf{Rel.} & \textbf{Separator} \\
      \hline
      \eqref{w8}   & \(?\gw\)   \\
      \eqref{H2}   & \(?\gH\)   \\
      \eqref{S4}   & \(?\gS\)   \\
      \eqref{E}    & \(\#\{\gH\}_{[2]}\) \\
      \eqref{Cs}   & \(?\gCNOT\)\\
      \eqref{B}    & \(\#\{\gSWAP\}_{[2]}\) \\
      \eqref{CZ}   & \(\#\{\gCNOT,\!\gSWAP\}_{[2]}\) \\
      \eqref{Inew} & \(\arg\det_{2}\) \\
      \hline
    \end{tabular}
  \end{subfigure}
  \begin{subfigure}[t]{0.31\textwidth}
    \centering
    \caption{Real Clifford}
    \begin{tabular}{|c|c|}
      \hline
      \textbf{Rel.}      & \textbf{Separator} \\
      \hline
      \eqref{realminus2}     & \(?\gminus\) \\
      \eqref{realH2}         & \(?\gH\) \\
      \eqref{realZ2}         & \(?\gZ\) \\
      \eqref{realF}          & \(\#\{\gminus\}_{[2]}\) \\
      \eqref{realCX2}        & \(?\gCNOT\) \\
      \eqref{realB}          & \(\#\{\gSWAP\}_{[2]}\) \\
      \eqref{realwC}         & \(\#\{\gZ\}_{[2]}\) \\
      \eqref{realCF}         & \(\psubst{\gZ}{\gI}\) \\
      \eqref{realXC}         & \(\psubst{\gH}{\gI}\) \\
      \eqref{realInew}       & \(\arg\det_{2}\) \\
      \hline
    \end{tabular}
  \end{subfigure}
  \begin{subfigure}[t]{0.31\textwidth}
    \centering
    \caption{Qutrit Clifford}
    \begin{tabular}{|c|c|}
      \hline
      \textbf{Rel.}            & \textbf{Separator} \\
      \hline
      \eqref{qt-omegapow12}        & \(?\gww\) \\
      \eqref{qt-hpow4}             & \(?\gH\) \\
      \eqref{qt-spow3}             & \(?\gS\) \\
      \eqref{qt-shpow3}            & \(\#\{\gH\}_{[2]}\) \\
      \eqref{qt-ssprime}           & \(\psubst{\gS}{\gSX}\)\footnote{See \cref{qt-prop:SSpnecessity}.} \\
      \eqref{qt-cnotremove}        & \(?\gCNOT\) \\
      \eqref{qt-swapdecomp}        & \(\#\{\gSWAP\}_{[2]}\) \\
      \eqref{qt-czdecomp}          & \(\#\{\gS\}_{[3]}\) \\
      \eqref{qt-kcnot}             & \(\#\{\gCNOT\}_{[2]}\) \\
      \eqref{qt-I}             & \noseparator \\
      \hline
    \end{tabular}
  \end{subfigure}
  \begin{subfigure}[t]{0.35\textwidth}
    \centering
    \caption{Clifford+T}
    \begin{tabular}{|c|c|}
      \hline
      \textbf{Rel.}     & \textbf{Separator} \\
      \hline
      \eqref{t-w8}         & \(?\gw\) \\
      \eqref{t-H2}         & \(?\gH\) \\
      \eqref{t-Tpow8}      & \(?\gT\) \\
      \eqref{t-E}          & \(\#\{\gH\}_{[2]}\) \\
      \eqref{t-TX}         & \(\#\{\gH,\!\gT,\!\gw\}_{[2]}\) \\
      \eqref{t-Ct}         & \(?\gCNOT\) \\
      \eqref{t-B}          & \(\#\{\gSWAP\}_{[2]}\) \\
      \eqref{t-CZ}         & \(\#\{\gCNOT,\!\gSWAP\}_{[2]}\) \\
      \eqref{t-CsCh}         & \noseparator \\
      \eqref{t-HTHT}         & \noseparator \\
      \eqref{t-C20}         & \noseparator \\
      \hline
    \end{tabular}
  \end{subfigure}
  \begin{subfigure}[t]{0.31\textwidth}
    \centering
    \caption{Clifford+CS}
    \begin{tabular}{|c|c|}
      \hline
      \textbf{Rel.}     & \textbf{Separator} \\
      \hline
      \eqref{cs-w8}         & \(?\gw\) \\
      \eqref{cs-H2}         & \(?\gH\) \\
      \eqref{cs-S4}         & \(?\gS\) \\
      \eqref{cs-E}          & \(\#\{\gH\}_{[2]}\) \\
      \eqref{cs-C}          & \(?\gCS\) \\
      \eqref{cs-B}          & \(\#\{\gSWAP\}_{[2]}\) \\
      \eqref{cs-XCS}        & \(\#\{\gS,\!\gH\}_{[2]}\) \\
      \eqref{cs-CSrev}      & \(\psubst{\gCS}{\gCSz}\) \\
      \eqref{cs-SHCHC}      & \(\psubst{\gCS}{\gCSzz}\) \\
      \eqref{cs-I}          & \(\arg\det_{2}\) \\
      \eqref{cs-U}          & \noseparator \\
      \hline
    \end{tabular}
  \end{subfigure}
  \begin{subfigure}[t]{0.32\textwidth}
    \centering
    \caption{CNOT-dihedral}
    \begin{tabular}{|c|c|}
      \hline
      \textbf{Rel.} & \textbf{Separator} \\
      \hline
      \eqref{cnot-R10} & \(?\gw\) \\
      \eqref{cnot-R7}  & \(?\gT\) \\
      \eqref{cnot-R1}  & \(?\gX\) \\
      \eqref{cnot-R11} & \(\#\{\gw\}_{[2]}\) \\
      \eqref{new-R3}   & \(\#\{\gX\}_{[2]}\) \\
      \eqref{cnot-R12} & \(?\gCNOT\) \\
      \eqref{new-R5}   & \(\#\{\gSWAP\}_{[2]}\) \\
      \eqref{new-R8}   & \(\#\{\gT,\!\gw,\!\gw\}_{[8]}\) \\
      \eqref{cnot-R6}  & \(\#\{\gCNOT,\!\gSWAP\}_{[2]}\) \\
      \eqref{cnot-R9}  & \(\#\{\gT,\!\gw,\!\gw\}_{[4]}\) \\
      \eqref{cnot-R13} & \(\arg\det_{3}\) \\
      \hline
    \end{tabular}
  \end{subfigure}
  \caption{Separators witnessing independence, indexed by the equation labels from \cref{sec:circuits-relations}; entries marked \(\noseparator\) are the cases for which no separator is currently provided.}\label{fig:minimality-all}
\end{figure}

\subsection{Minimality statements}

The \(\noseparator\) rows lie outside the bounded ranges claimed in \cref{thm:minimality-all}.  When
the table contains a separator beyond such a range, the corresponding higher-arity axiom is stated
separately.

\begin{theorem}\label{thm:minimality-all}
The indicated presented circuit fragments have minimal rule sets as listed:
\begin{enumerate}
\item \(\QCclifford\) is minimal for \(\cat{P}_{\mathit{Cliff}}\)-circuits (all arities).
\item \(\QCrealclifford\) is minimal for \(\cat{P}_{\mathit{RCliff}}\)-circuits (all arities).
\item \(\QCqutritclifford\) is minimal for \(\cat{P}_{\mathit{Cliff3}}\)-circuits up to \(2\) qutrit wires.\footnote{Full minimality would follow from a separator for the \(3\)-qutrit axiom \cref{qt-I}.}
\item \(\QCcliffordplust\) is minimal for \(\cat{P}_{\mathit{CliffT}}\)-circuits up to \(1\) qubit wire.
\item \(\QCcliffordpluscs\) is minimal for \(\cat{P}_{\mathit{CliffCS}}\)-circuits up to \(2\) qubit wires.
\item \(\QCdihedral\) is minimal for \(\cat{P}_{\mathit{CNOTdihe}}\)-circuits (all arities).
\end{enumerate}

For the bounded fragments, the constructed higher-arity separators give:
\begin{itemize}
\item In \(\QCcliffordplust\), the \(2\)-qubit axioms \eqref{t-Ct}, \eqref{t-B} and \eqref{t-CZ}
      are independent.
\item In \(\QCcliffordpluscs\), the \(3\)-qubit axiom \eqref{cs-I} is independent.
\end{itemize}
\end{theorem}

\begin{proof}
For each axiom \(\rho\) covered by the statement, the corresponding non-\(\noseparator\) row of
\cref{fig:minimality-all} gives a separator \(F_\rho\) that equalises the
remaining axioms of the relevant truncated system while distinguishing the two sides of \(\rho\).
For the additional higher-arity statements, the relevant truncated system is the one at the arity of
the displayed axiom.  Apply \cref{lem:separation} and
\cref{def:bounded-minimality}.
\end{proof}
\section{Conclusion}
By placing the six fragments in a common PROP syntax with structural symmetries, the paper isolates
fewer non-structural axioms while preserving the strict unitary semantics inherited from the source
presentations.  The completeness transfer and separator table show that this simplification often
preserves irredundancy: qubit Clifford, real Clifford, and CNOT-dihedral are minimal in all arities,
while the other fragments are minimal in the bounded ranges certified in \cref{sec:minimality}.  The
\(\noseparator\) rows of \cref{fig:minimality-all} identify the remaining separator-construction
problem, and the reduced presentations give mechanised rewriting a smaller algebraic core: related
work on \(2\)-qubit Clifford\(+T\) and \(3\)-qubit Clifford\(+CS\) already gives Agda checks for this style of
argument \cite{SelingerCliffordPlusT,SelingerCliffordPlusCS}, while recent Isabelle/HOL work on
verified rewriting infrastructure for diagrammatic equivalence
\cite{CaillerDelormePerdrixTourret2026} offers another benchmark for future formalisation.


\bibliography{references}

\appendix
\newpage
\begin{figure}[H]
    \centering
    \scalebox{0.9}{
    \fbox{\begin{minipage}{0.975\textwidth}
    \centering
    \textbf{Shortcuts for $\qubitclifford$-circuits.}\par\vspace{-0.6em}
    \begin{minipage}{0.24\textwidth}
        \centering
        \begin{equation}
            \tag{Z}\label{Z}
            \input{shortcut/Z.tikz} \defeq \input{shortcut/SS.tikz}
        \end{equation}
    \end{minipage}
    \begin{minipage}{0.34\textwidth}
        \centering
        \begin{equation}
            \tag{X}\label{X}
            \input{shortcut/X.tikz} \defeq \input{shortcut/HZH.tikz}
        \end{equation}
    \end{minipage}\vspace{-0.6em}
    \begin{minipage}{0.25\textwidth}
        \centering
        \begin{equation}
            \tag{S$^\dagger$}\label{Sdagger}
            \input{shortcut/Sdagger.tikz} \defeq \input{shortcut/SSS.tikz}
        \end{equation}
    \end{minipage}
    \begin{minipage}{0.22\textwidth}
        \centering
        \begin{equation}
            \tag{$\omega^\dagger$}\label{wdagger}
            \input{shortcut/wdagger.tikz} \defeq \input{shortcut/w7.tikz}
        \end{equation}
    \end{minipage}
    \end{minipage}}
    }
    \scalebox{0.9}{
    \fbox{\begin{minipage}{0.975\textwidth}
    \centering
    \textbf{Shortcuts for $\qubitrealclifford$-circuits.}\par\vspace{-0.6em}
    \begin{minipage}{0.37\textwidth}
        \centering
        \begin{equation}
            \tag{X}\label{real-X}
            \input{shortcut/X.tikz} \defeq \input{shortcut/HZH.tikz}
        \end{equation}
    \end{minipage}
    \end{minipage}}
    }
    \scalebox{0.9}{
    \fbox{\begin{minipage}{0.975\textwidth}
    \centering
    \textbf{Shortcuts for $\qutritclifford$-circuits.}\par\vspace{-0.6em}
    \begin{minipage}{0.25\textwidth}
        \centering
        \begin{equation}
            \tag{K}\label{qt-K}
            \input{shortcut/K.tikz} \defeq \input{shortcut/HH.tikz}
        \end{equation}
    \end{minipage}
    \begin{minipage}{0.28\textwidth}
        \centering
        \begin{equation}
            \tag{H$^3$}\label{qt-H3}
            \input{shortcut/Hdagger.tikz} \defeq \input{shortcut/HHH.tikz}
        \end{equation}
    \end{minipage}
    \begin{minipage}{0.28\textwidth}
        \centering
        \begin{equation}
            \tag{S$^\dagger$}\label{qt-Sdagger}
            \input{shortcut/Sdagger.tikz} \defeq \input{shortcut/SS.tikz}
        \end{equation}
    \end{minipage}\vspace{-0.6em}
    \begin{minipage}{0.33\textwidth}
        \centering
        \begin{equation}
            \tag{S'}\label{qt-Sprime}
            \input{shortcut/Sprime.tikz} \defeq \input{shortcut/KSK.tikz}
        \end{equation}
    \end{minipage}
    \begin{minipage}{0.28\textwidth}
        \centering
        \begin{equation}
            \tag{Z}\label{qt-Z}
            \input{shortcut/Z.tikz} \defeq \input{shortcut/SSSprime.tikz}
        \end{equation}
    \end{minipage}
    \begin{minipage}{0.33\textwidth}
        \centering
        \begin{equation}
            \tag{X}\label{qt-X}
            \input{shortcut/X.tikz} \defeq \input{shortcut/HZHHH.tikz}
        \end{equation}
    \end{minipage}
    \end{minipage}}
    }
    \scalebox{0.9}{
    \fbox{\begin{minipage}{0.975\textwidth}
    \centering
    \textbf{Shortcuts for $\qubitcliffordplust$-circuits.}\par\vspace{-0.6em}
    \begin{minipage}{0.25\textwidth}
        \centering
        \begin{equation}
            \tag{S}\label{t-S}
            \input{shortcut/S.tikz} \defeq \input{shortcut/TT.tikz}
        \end{equation}
    \end{minipage}
    \begin{minipage}{0.25\textwidth}
        \centering
        \begin{equation}
            \tag{S$^{\dagger}$}\label{t-Sdagger}
            \input{shortcut/Sdagger.tikz} \defeq \input{shortcut/SSS.tikz}
        \end{equation}
    \end{minipage}
    \begin{minipage}{0.25\textwidth}
        \centering
        \begin{equation}
            \tag{Z}\label{t-Z}
            \input{shortcut/Z.tikz} \defeq \input{shortcut/SS.tikz}
        \end{equation}
    \end{minipage}\vspace{-0.6em}
    \begin{minipage}{0.33\textwidth}
        \centering
        \begin{equation}
            \tag{X}\label{t-X}
            \input{shortcut/X.tikz} \defeq \input{shortcut/HZH.tikz}
        \end{equation}
    \end{minipage}
    \begin{minipage}{0.26\textwidth}
        \centering
        \begin{equation}
            \tag{T$^\dagger$}\label{t-Tdagger}
            \input{shortcut/Tdagger.tikz} \defeq \input{shortcut/TTTTTTT.tikz}
        \end{equation}
    \end{minipage}
    \begin{minipage}{0.35\textwidth}
        \centering
        \begin{equation}
            \tag{0CX}\label{t-CnotWhite}
            \input{shortcut/CnotWhite.tikz} \defeq \input{shortcut/CnotX.tikz}
        \end{equation}
    \end{minipage}\vspace{-0.8em}
    \begin{minipage}{0.53\textwidth}
        \centering
        \begin{equation}
            \tag{CH}\label{t-Ch}
            \input{shortcut/Ch.tikz} \defeq \input{shortcut/ChDecomp.tikz}
        \end{equation}
    \end{minipage}
    \begin{minipage}{0.35\textwidth}
        \centering
        \begin{equation}
            \tag{0CH}\label{t-ChWhite}
            \input{shortcut/ChWhite.tikz} \defeq \input{shortcut/ChH.tikz}
        \end{equation}
    \end{minipage}\vspace{-0.6em}
    \begin{minipage}{0.43\textwidth}
        \centering
        \begin{equation}
            \tag{CS}\label{t-Cs}
            \input{shortcut/CS.tikz} \defeq \input{shortcut/CSdecomp.tikz}
        \end{equation}
    \end{minipage}
    \begin{minipage}{0.39\textwidth}
        \centering
        \begin{equation}
            \tag{CS$^{\dagger}$}\label{t-Csdagger}
            \input{shortcut/CSdagger.tikz} \defeq \input{shortcut/CSCSCS.tikz}
        \end{equation}
    \end{minipage}
    \end{minipage}}
    }
    \scalebox{0.9}{
        \fbox{\begin{minipage}{0.975\textwidth} 
            \centering
            \textbf{Shortcuts for $\qubitcliffordpluscs$-circuits.}\par\vspace{-0.6em}
            \begin{minipage}{0.24\textwidth}
                \centering
                \begin{equation}
                    \tag{Z}\label{cs-Z}
                    \input{shortcut/Z.tikz} \defeq \input{shortcut/SS.tikz}
                \end{equation}
            \end{minipage}
            \begin{minipage}{0.37\textwidth}
                \centering
                \begin{equation}
                    \tag{X}\label{cs-X}
                    \input{shortcut/X.tikz} \defeq \input{shortcut/HZH.tikz}
                \end{equation}
            \end{minipage}
            \begin{minipage}{0.25\textwidth}
                \centering
                \begin{equation}
                    \tag{S$^\dagger$}\label{cs-Sdagger}
                    \input{shortcut/Sdagger.tikz} \defeq \input{shortcut/SSS.tikz}
                \end{equation}
            \end{minipage}\vspace{-0.6em}
            \begin{minipage}{0.24\textwidth}
                \centering
                \begin{equation}
                    \tag{$\omega^\dagger$}\label{cs-wdagger}
                    \input{shortcut/wdagger.tikz} \defeq \input{shortcut/w7.tikz}
                \end{equation}
            \end{minipage}
            \begin{minipage}{0.4\textwidth}
                \centering
                \begin{equation}
                    \tag{CX}\label{cs-CNOT}
                    \input{gates/CNOT.tikz} \defeq \input{shortcut/HCSCSH.tikz}
                \end{equation}
            \end{minipage}
            \end{minipage}}
        }
        \scalebox{0.9}{
        \fbox{\begin{minipage}{0.975\textwidth} 
            \centering
            \textbf{Shortcuts for $\qubitcnotdihedral$-circuits.}\par\vspace{-0.6em}
            \begin{minipage}{0.25\textwidth}
                \centering
                \begin{equation}
                    \tag{Z}\label{cnot-Z}
                    \input{shortcut/Z.tikz} \defeq \input{shortcut/SS.tikz}
                \end{equation}
            \end{minipage}
            \begin{minipage}{0.25\textwidth}
                \centering
                \begin{equation}
                    \tag{S}\label{cnot-S}
                    \input{shortcut/S.tikz} \defeq \input{shortcut/TT.tikz}
                \end{equation}
            \end{minipage}\vspace{-0.6em}

            \begin{minipage}{0.39\textwidth}
                \centering
                \begin{equation}
                    \tag{CT}\label{cnot-CU}
                    \input{shortcut/CU.tikz} \defeq \input{shortcut/CUr.tikz}
                \end{equation}
            \end{minipage}
            \begin{minipage}{0.49\textwidth}
                \centering
                \begin{equation}
                    \tag{CCT}\label{cnot-CV}
                    \input{shortcut/CV.tikz} \defeq \input{shortcut/CVr.tikz}
                \end{equation}
            \end{minipage}
            \end{minipage}}
        }
    \scalebox{0.9}{
    \fbox{\begin{minipage}{0.975\textwidth}
    \centering
    \textbf{Common shorthand notations}\par\vspace{-0.6em}
    \begin{minipage}{0.34\textwidth}
        \centering
        \begin{equation}
            \tag{$\omega^{k}$}\label{wk}
            \input{shortcut/alphak.tikz} \defeq \input{shortcut/alphatimesk.tikz}
        \end{equation}
    \end{minipage}
    \begin{minipage}{0.39\textwidth}
        \centering
        \begin{equation}
            \tag{$U^k$}\label{Uk}
            \input{shortcut/Uk.tikz} \defeq \input{shortcut/Utimesk.tikz}
        \end{equation}
    \end{minipage}\vspace{-0.6em}
    \begin{minipage}{0.45\textwidth}
        \centering
        \begin{equation}
            \label{bigCNOT}
            \input{shortcut/bigCNOT-left.tikz} \defeq \input{shortcut/bigCNOT-right.tikz}
        \end{equation}
    \end{minipage}
    \begin{minipage}{0.45\textwidth}
        \centering
        \begin{equation}
            \label{bigNOTC}
            \input{shortcut/bigNOTC-left.tikz} \defeq \input{shortcut/bigNOTC-right.tikz}
        \end{equation}
    \end{minipage}
    \end{minipage}}
    }
    \caption{Summary of Shortcuts used for all considered fragments.}
    \label{fig:quantum-circuit-shortcuts}
\end{figure}

\begin{figure}[ht]
    \centering
    \scalebox{0.9}{
    \fbox{\begin{minipage}{0.975\textwidth}
    \centering
    \textbf{Shortcuts for $\altqubitclifford$-circuits.}\par\vspace{-0.6em}
    \begin{minipage}{0.27\textwidth}
        \centering
        \begin{equation}
            \tag{$\omega^\dagger$}\label{wdagger-old}
            \input{shortcut/wdagger.tikz} \defeq \input{shortcut/w7.tikz}
        \end{equation}
    \end{minipage}
    \end{minipage}}
    }
    \scalebox{0.9}{
    \fbox{\begin{minipage}{0.975\textwidth}
    \centering
    \textbf{Shortcuts for $\altqutritclifford$-circuits.}\par\vspace{-0.6em}
    \begin{minipage}{0.39\textwidth}
        \centering
        \begin{equation}
            \tag{CNOT}\label{qt-CNOT}
            \input{shortcut/CNOT.tikz} \defeq \input{shortcut/HCZHHH.tikz}
        \end{equation}
    \end{minipage}
    \begin{minipage}{0.55\textwidth}
        \centering
        \begin{equation}
            \tag{SWAP}\label{qt-SWAP-shortcut}
            \input{shortcut/SWAP-QUTRIT.tikz} \defeq \input{shortcut/SWAP-QUTRIT-DECOMP.tikz}
        \end{equation}
    \end{minipage}
    \end{minipage}}
    }
    \scalebox{0.9}{
    \fbox{\begin{minipage}{0.975\textwidth}
    \centering
    \textbf{Shortcuts for $\altqubitcliffordplust$-circuits.}\par\vspace{-0.6em}
    \begin{minipage}{0.25\textwidth}
        \centering
        \begin{equation}
            \tag{S}\label{t-oldS}
            \input{shortcut/S.tikz} \defeq \input{shortcut/TT.tikz}
        \end{equation}
    \end{minipage}
    \begin{minipage}{0.25\textwidth}
        \centering
        \begin{equation}
            \tag{Z}\label{t-oldZ}
            \input{shortcut/Z.tikz} \defeq \input{shortcut/SS.tikz}
        \end{equation}
    \end{minipage}
    \begin{minipage}{0.35\textwidth}
        \centering
        \begin{equation}
            \tag{X}\label{t-oldX}
            \input{shortcut/X.tikz} \defeq \input{shortcut/HZH.tikz}
        \end{equation}
    \end{minipage}\vspace{-0.6em}
    \begin{minipage}{0.25\textwidth}
        \centering
        \begin{equation}
            \tag{T$^\dagger$}\label{t-oldTdagger}
            \input{shortcut/Tdagger.tikz} \defeq \input{shortcut/TTTTTTT.tikz}
        \end{equation}
    \end{minipage}
    \begin{minipage}{0.35\textwidth}
        \centering
        \begin{equation}
            \tag{0CX}\label{t-oldCnotWhite}
            \input{shortcut/CnotWhite.tikz} \defeq \input{shortcut/CnotX.tikz}
        \end{equation}
    \end{minipage}\vspace{-0.6em}
    \begin{minipage}{0.54\textwidth}
        \centering
        \begin{equation}
            \tag{CH}\label{t-oldCh}
            \input{shortcut/Ch.tikz} \defeq \input{shortcut/ChDecomp.tikz}
        \end{equation}
    \end{minipage}
    \begin{minipage}{0.36\textwidth}
        \centering
        \begin{equation}
            \tag{0CH}\label{t-oldChWhite}
            \input{shortcut/ChWhite.tikz} \defeq \input{shortcut/ChH.tikz}
        \end{equation}
    \end{minipage}
    \end{minipage}}
    }
    \scalebox{0.9}{
    \fbox{\begin{minipage}{0.975\textwidth} 
        \centering
        \textbf{Shortcuts for $\altqubitcliffordpluscs$-circuits.}\par\vspace{-0.6em}
        \begin{minipage}{0.25\textwidth}
            \centering
            \begin{equation}
                \tag{Z}\label{cs-oldZ}
                \input{shortcut/Z.tikz} \defeq \input{shortcut/SS.tikz}
            \end{equation}
        \end{minipage}
        \begin{minipage}{0.27\textwidth}
            \centering
            \begin{equation}
                \tag{$\omega^\dagger$}\label{cs-oldwdagger}
                \input{shortcut/wdagger.tikz} \defeq \input{shortcut/w7.tikz}
            \end{equation}
        \end{minipage}\vspace{-0.6em}
        \begin{minipage}{0.5\textwidth}
            \centering
            \begin{equation}
                \tag{oldCX}\label{cs-CNOTred}
                \input{gates/CNOTRED.tikz} \defeq \input{shortcut/KCSCSK.tikz}
            \end{equation}
        \end{minipage}
        \end{minipage}}
    }
    \scalebox{0.9}{
    \fbox{\begin{minipage}{0.975\textwidth}
    \centering
    \textbf{Common shorthand notations}\par\vspace{-0.6em}
    \begin{minipage}{0.32\textwidth}
        \centering
        \begin{equation}
            \tag{$\omega^{k}$}\label{oldwk}
            \input{shortcut/alphak.tikz} \defeq \input{shortcut/alphatimesk.tikz}
        \end{equation}
    \end{minipage}
    \begin{minipage}{0.39\textwidth}
        \centering
        \begin{equation}
            \tag{$U^k$}\label{oldUk}
            \input{shortcut/Uk.tikz} \defeq \input{shortcut/Utimesk.tikz}
        \end{equation}
    \end{minipage}\vspace{-0.6em}
    \begin{minipage}{0.45\textwidth}
        \centering
        \begin{equation}
            \label{oldbigCNOT}
            \input{shortcut/bigCNOT-left.tikz} \defeq \input{shortcut/bigCNOT-right.tikz}
        \end{equation}
    \end{minipage}
    \begin{minipage}{0.45\textwidth}
        \centering
        \begin{equation}
            \label{oldbigNOTC}
            \input{shortcut/bigNOTC-left.tikz} \defeq \input{shortcut/bigNOTC-right.tikz}
        \end{equation}
    \end{minipage}
    \end{minipage}}
    }
    \caption{Summary of Shortcuts used for all considered fragments used in the original PROs.}
    \label{fig:quantum-circuit-shortcuts-pros}
\end{figure}
\section{Auxiliary derivations for
  \texorpdfstring{$\QCclifford$}{QCclifford}}
\label{appendix:clifford-decoding-lemmas}


\begin{derivation}\label{Z2}
\begin{gather*}
  \input{identities/ZZ-00.tikz}
  \eqeqref{Z}\input{identities/ZZ-01.tikz}
  \eqeqref{S4}\input{identities/ZZ-02.tikz}
\end{gather*}
\end{derivation}

\begin{derivation}\label{X2}
\begin{gather*}
  \input{identities/XX-00.tikz}
  \eqeqref{X}\input{identities/XX-01.tikz}
  \eqeqref{H2}\input{identities/XX-02.tikz}
  \eqeqref{Z2}\input{identities/XX-03.tikz}
  \eqeqref{H2}\input{identities/XX-04.tikz}
\end{gather*}
\end{derivation}

\begin{derivation}\label{XSX}
\begin{gather*}
  \input{identities/XSX-00.tikz}
  \eqeqref{X}\input{identities/XSX-01.tikz}\\
  \eqeqref{E}\input{identities/XSX-02.tikz}
  \eqeqref{S4}\input{identities/XSX-03.tikz}\\
  \eqeqref{E}\input{identities/XSX-04.tikz}
  \eqdeuxeqref{S4}{H2}\input{identities/XSX-05.tikz}
\end{gather*}
\end{derivation}

\begin{derivation}\label{eulerH}
\begin{gather*}
  \input{identities/eulerH-00.tikz}
  \eqeqref{S4}\input{identities/eulerH-01.tikz}
  \eqeqref{H2}\input{identities/eulerH-02.tikz}
  \eqeqref{E}\input{identities/eulerH-03.tikz}
\end{gather*}
\end{derivation}

\begin{derivation}\label{CX2}
\begin{gather*}
  \input{cliffordaxioms/CNOTCNOT.tikz}
  \eqeqref{S4}\input{identities/CNOTCNOT-01.tikz}
  \eqeqref{Cs}\input{identities/CNOTCNOT-02.tikz}\\
  \eqeqref{B}\input{identities/CNOTCNOT-03.tikz}
  =\input{identities/CNOTCNOT-04.tikz}\\
  \eqeqref{B}\input{identities/CNOTCNOT-05.tikz}
  =\input{identities/CNOTCNOT-06.tikz}\\
  =\input{identities/CNOTCNOT-07.tikz}\\
  \eqeqref{Cs}\input{identities/CNOTCNOT-08.tikz}
  \eqeqref{S4}\input{cliffordaxioms/II.tikz}
\end{gather*}
\end{derivation}

\begin{derivation}\label{C}
\begin{gather*}
  \input{cliffordaxioms/SCNOT.tikz}
  \eqeqref{Cs}\input{identities/SCNOT-01.tikz}
  \eqeqref{CX2}\input{cliffordaxioms/CNOTS.tikz}
\end{gather*}
\end{derivation}

\begin{derivation}\label{phasegadget}
\begin{gather*}
  \input{identities/phasegadget-00.tikz}
  =\input{identities/phasegadget-01.tikz}
  \eqeqref{B}\input{identities/phasegadget-02.tikz}\\
  \eqeqref{C}\input{identities/phasegadget-03.tikz}
  \eqeqref{CX2}\input{identities/phasegadget-04.tikz}
\end{gather*}
\end{derivation}

\begin{derivation}\label{CNOTHH}
\begin{gather*}
  \input{identities/CNOTHH-00.tikz}
  \eqeqref{H2}\input{identities/CNOTHH-01.tikz}
  \eqeqref{CZ}\input{identities/CNOTHH-02.tikz}\\
  \eqeqref{phasegadget}\input{identities/CNOTHH-03.tikz}
  \eqeqref{CZ}\input{identities/CNOTHH-04.tikz}
  \eqeqref{H2}\input{identities/CNOTHH-05.tikz}
\end{gather*}
\end{derivation}

\begin{derivation}\label{czrev}
\begin{gather*}
  \input{identities/czrev-00.tikz}
  \eqeqref{H2}\input{identities/czrev-01.tikz}
  \eqeqref{CZ}\input{identities/czrev-02.tikz}
  \eqeqref{H2}\input{identities/czrev-03.tikz}
\end{gather*}
\end{derivation}

\begin{derivation}\label{XcommutCNOT}
\begin{gather*}
  \input{identities/XcommutCNOT-00.tikz}
  \eqtroiseqref{H2}{X}{Z}\input{identities/XcommutCNOT-01.tikz}
  \eqeqref{CNOTHH}\input{identities/XcommutCNOT-02.tikz}\\
  \eqeqref{C}\input{identities/XcommutCNOT-03.tikz}
  \eqeqref{CNOTHH}\input{identities/XcommutCNOT-04.tikz}
  \eqtroiseqref{H2}{X}{Z}\input{identities/XcommutCNOT-05.tikz}
\end{gather*}
\end{derivation}

\begin{derivation}\label{wCNOT}
\begin{gather*}
  \input{identities/wCNOT-00.tikz}
  \eqeqref{H2}\input{identities/wCNOT-01.tikz}
  \eqeqref{CZ}\input{identities/wCNOT-02.tikz}\\
  \eqeqref{phasegadget}\input{identities/wCNOT-03.tikz}
  \eqeqref{X}\input{identities/wCNOT-04.tikz}\\
  \eqdeuxeqref{Z}{C}\input{identities/wCNOT-05.tikz}
  \eqdeuxeqref{Z}{Sdagger}\input{identities/wCNOT-06.tikz}\\
  \eqeqref{XSX}\input{identities/wCNOT-07.tikz}\\
  \eqeqref{w8}\input{identities/wCNOT-08.tikz}\\
  \eqdeuxeqref{XcommutCNOT}{X2}\input{identities/wCNOT-09.tikz}
  \eqeqref{XcommutCNOT}\input{identities/wCNOT-10.tikz}\\
  \eqeqref{phasegadget}\input{identities/wCNOT-11.tikz}
  \eqeqref{CX2}\input{identities/wCNOT-12.tikz}\\
  \eqdeuxeqref{C}{phasegadget}\input{identities/wCNOT-13.tikz}
  \eqeqref{CZ}\input{identities/wCNOT-14.tikz}\\
  \eqeqref{CZ}\input{identities/wCNOT-15.tikz}
  \eqtroiseqref{S4}{CX2}{C}\input{identities/wCNOT-16.tikz}
\end{gather*}
\end{derivation}

\begin{derivation}\label{zCNOT}
\begin{gather*}
  \input{identities/zCNOT-00.tikz}
  \eqdeuxeqref{H2}{CNOTHH}\input{identities/zCNOT-01.tikz}
  \eqeqref{X}\input{identities/zCNOT-02.tikz}
  \eqeqref{wCNOT}\input{identities/zCNOT-03.tikz}\\
  \eqeqref{X}\input{identities/zCNOT-04.tikz}
  \eqdeuxeqref{H2}{CNOTHH}\input{identities/zCNOT-05.tikz}
\end{gather*}
\end{derivation}

\begin{derivation}\label{I}
\begin{gather*}
  \input{cliffordaxioms/CNOT13CNOT23.tikz}
  \eqeqref{CX2}\input{cliffordczaxioms/I2_00.tikz}
  \eqeqref{Inew}\input{cliffordczaxioms/I2_01.tikz}
  \eqeqref{CX2}\input{cliffordaxioms/CNOT12CNOT23CNOT12.tikz}
\end{gather*}
\end{derivation}

\begin{derivation}\label{3CNOTtargetcommut}
\begin{gather*}
  \input{identities/3CNOTtargetcommut-00.tikz}
  \eqeqref{CX2}\input{identities/3CNOTtargetcommut-01.tikz}
  \eqeqref{I}\input{identities/3CNOTtargetcommut-02.tikz}
  \eqeqref{CX2}\input{identities/3CNOTtargetcommut-03.tikz}
\end{gather*}
\end{derivation}

\begin{derivation}\label{3CNOTcontrolcommut}
\begin{gather*}
  \input{identities/3CNOTcontrolcommut-00.tikz}
  \eqeqref{CNOTHH}\input{identities/3CNOTcontrolcommut-01.tikz}
  \eqeqref{I}\input{identities/3CNOTcontrolcommut-02.tikz}
  \eqeqref{CNOTHH}\input{identities/3CNOTcontrolcommut-03.tikz}
\end{gather*}
\end{derivation}


\begin{derivation}\label{lem:dswap-unfolding}
\begin{gather*}
  \input{gates/SWAP.tikz}
  \eqeqref{H2}\input{cliffordczaxioms/DSWAPdef-00b.tikz}
  \eqeqref{CX2}\input{cliffordczaxioms/DSWAPdef-01.tikz}\\
  \eqeqref{B}\input{cliffordczaxioms/DSWAPdef-01b.tikz}
  \eqeqref{H2}\input{cliffordczaxioms/DSWAPdef-03.tikz}\\
  \eqeqref{czrev}\input{cliffordczaxioms/DSWAPdef-03b.tikz}
\end{gather*}
\end{derivation}

\begin{derivation}\label{lem:C4-dec-eq}
\begin{gather*}
  \input{cliffordczaxioms/C4-00.tikz}
  \eqeqref{E}\input{cliffordczaxioms/C4-01.tikz}\\
  \eqeqref{S4}\input{cliffordczaxioms/C4-02.tikz}
  \eqeqref{H2}\input{cliffordczaxioms/C4-03.tikz}
\end{gather*}
\end{derivation}

\begin{derivation}\label{lem:C5-dec-eq}
\begin{gather*}
  \input{cliffordczaxioms/C5-01.tikz}
  \eqeqref{H2}\input{cliffordczaxioms/C5-02.tikz}
  \eqeqref{CX2}\input{cliffordczaxioms/C5-03.tikz}
  \eqeqref{H2}\input{cliffordczaxioms/C5-04.tikz}
\end{gather*}
\end{derivation}

\begin{derivation}\label{lem:C6-dec-eq}
\begin{gather*}
  \input{cliffordczaxioms/C6-01.tikz}
  \eqeqref{C}\input{cliffordczaxioms/C6-02.tikz}
\end{gather*}
\end{derivation}

\begin{derivation}\label{lem:C7-dec-eq}
\begin{gather*}
  \input{cliffordczaxioms/C7-01.tikz}
  \eqeqref{czrev}\input{cliffordczaxioms/C7-02.tikz}\\
  \eqeqref{lem:C6-dec-eq}\input{cliffordczaxioms/C7-03.tikz}
  \eqeqref{czrev}\input{cliffordczaxioms/C7-04.tikz}
\end{gather*}
\end{derivation}

\begin{derivation}\label{lem:C8-dec-eq}
\begin{gather*}
  \input{cliffordczaxioms/C8-01.tikz}
  \eqeqref{X}\input{cliffordczaxioms/C8-02.tikz}
  \eqeqref{X2}\input{cliffordczaxioms/C8-03.tikz}\\
  \eqeqref{wCNOT}\input{cliffordczaxioms/C8-04.tikz}
  \eqdeuxeqref{X}{Z}\input{cliffordczaxioms/C8-05.tikz}\\
  \eqeqref{X}\input{cliffordczaxioms/C8-06.tikz}
  \eqeqref{Z}\input{cliffordczaxioms/C8-07.tikz}
\end{gather*}
\end{derivation}

\begin{derivation}\label{lem:C9-dec-eq}
\begin{gather*}
  \input{cliffordczaxioms/C9-01.tikz}
  \eqeqref{czrev} \input{cliffordczaxioms/C9-01b.tikz}\\
  = \input{cliffordczaxioms/C9-0ab.tikz}
  \eqeqref{lem:C8-dec-eq}\input{cliffordczaxioms/C9-0bb.tikz}\\
  = \input{cliffordczaxioms/C9-0cb.tikz}
  \eqeqref{czrev} \input{cliffordczaxioms/C9-08.tikz}
\end{gather*}
\end{derivation}

\begin{derivation}\label{lem:C10-dec-eq}
\begin{gather*}
  \input{cliffordczaxioms/C10-01.tikz}
  \eqeqref{eulerH}\input{cliffordczaxioms/C10-02.tikz}\\
  \eqeqref{C}\input{cliffordczaxioms/C10-03.tikz}
  \eqeqref{CNOTHH}\input{cliffordczaxioms/C10-04.tikz}\\
  \eqeqref{H2}\input{cliffordczaxioms/C10-05.tikz}
  \eqeqref{S4}\input{cliffordczaxioms/C10-06.tikz}\\
  \eqeqref{C}\input{cliffordczaxioms/C10-07.tikz}
  \eqeqref{phasegadget}\input{cliffordczaxioms/C10-08.tikz}\\
  \eqeqref{zCNOT}\input{cliffordczaxioms/C10-09.tikz}\\
  \eqdeuxeqref{Z}{Sdagger}\input{cliffordczaxioms/C10-10.tikz}\\
  \eqeqref{Z}\input{cliffordczaxioms/C10-11.tikz}\\
  \eqeqref{C}\input{cliffordczaxioms/C10-12.tikz}\\
  \eqeqref{E}\input{cliffordczaxioms/C10-13.tikz}
\end{gather*}
\end{derivation}

\begin{derivation}\label{lem:C11-dec-eq}
\begin{gather*}
  \input{cliffordczaxioms/C11-00b.tikz}
  \eqeqref{czrev}\input{cliffordczaxioms/C11-0ab.tikz}\\
  \eqeqref{lem:C10-dec-eq}\input{cliffordczaxioms/C11-0bb.tikz}\\
  \eqeqref{czrev}\input{cliffordczaxioms/C11-0cb.tikz}\\
  = \input{cliffordczaxioms/C11-01b.tikz}
\end{gather*}
\end{derivation}

\begin{derivation}\label{lem:C12-dec-eq}
\begin{gather*}
  \input{cliffordczaxioms/C12-01.tikz}
  \eqeqref{CNOTHH}\input{cliffordczaxioms/C12-02.tikz}
  \eqeqref{3CNOTtargetcommut}\input{cliffordczaxioms/C12-03.tikz}\\
  \eqeqref{H2}\input{cliffordczaxioms/C12-04.tikz}
  \eqeqref{CNOTHH}\input{cliffordczaxioms/C12-05.tikz}
\end{gather*}
\end{derivation}

\begin{derivation}\label{lem:C13-dec-eq}
\begin{gather*}
  \input{cliffordczaxioms/C13-01.tikz}\\
  \eqeqref{H2}\input{cliffordczaxioms/C13-02.tikz}
  \eqeqref{CNOTHH}\input{cliffordczaxioms/C13-03.tikz}\\
  \eqeqref{H2}\input{cliffordczaxioms/C13-04.tikz}
  \eqdeuxeqref{H2}{CNOTHH}\input{cliffordczaxioms/C13-05.tikz}\\
  \eqeqref{B}\input{cliffordczaxioms/C13-06.tikz}\\
  \eqeqref{3CNOTtargetcommut}\input{cliffordczaxioms/C13-07.tikz}
  \eqeqref{CX2}\input{cliffordczaxioms/C13-08.tikz}
  =\input{cliffordczaxioms/C13-09.tikz}\\
  \eqeqref{H2}\input{cliffordczaxioms/C13-10.tikz}
  =\input{cliffordczaxioms/C13-11.tikz}
  \eqdeuxeqref{H2}{CNOTHH}\input{cliffordczaxioms/C13-12.tikz}
  =\input{cliffordczaxioms/C13-13.tikz}\\
  \eqeqref{H2}\input{cliffordczaxioms/C13-14.tikz}
  =\input{cliffordczaxioms/C13-15.tikz}
  \eqeqref{CX2}\input{cliffordczaxioms/C13-16.tikz}\\
  \eqeqref{3CNOTtargetcommut}\input{cliffordczaxioms/C13-17.tikz}
  \eqeqref{B}\input{cliffordczaxioms/C13-18.tikz}\\
  \eqeqref{H2}\input{cliffordczaxioms/C13-22.tikz}\\
  \eqdeuxeqref{H2}{CNOTHH}\input{cliffordczaxioms/C13-23.tikz}\\
  \eqeqref{H2}\input{cliffordczaxioms/C13-24.tikz}
\end{gather*}
\end{derivation}

\begin{derivation}\label{lem:C14-dec-eq}
\begin{gather*}
  \scalebox{0.8}{\input{cliffordczaxioms/C14-01.tikz}}\\
  \eqeqref{H2}\input{cliffordczaxioms/C14-02.tikz}\\
  \eqdeuxeqref{H2}{CNOTHH}\input{cliffordczaxioms/C14-03.tikz}\\
  \eqeqref{B}\input{cliffordczaxioms/C14-04.tikz}\\
  \eqeqref{I}\input{cliffordczaxioms/C14-05.tikz}
  =\input{cliffordczaxioms/C14-06.tikz}\\
  \eqdeuxeqref{3CNOTtargetcommut}{3CNOTcontrolcommut}\input{cliffordczaxioms/C14-07.tikz}
  \eqeqref{CX2}\input{cliffordczaxioms/C14-08.tikz}
  \eqeqref{H2}\input{cliffordczaxioms/C14-09.tikz}
\end{gather*}
\end{derivation}

\begin{derivation}\label{lem:C15-dec-eq}
\begin{gather*}
  \scalebox{0.8}{\input{cliffordczaxioms/C15-00b.tikz}}\\
  \eqeqref{czrev}\scalebox{0.7}{\input{cliffordczaxioms/C15-0ab.tikz}}\\
  \eqeqref{lem:C14-dec-eq}\input{cliffordczaxioms/C15-0bb.tikz}
  = \input{cliffordczaxioms/C15-01.tikz}
\end{gather*}
\end{derivation}
\section{Auxiliary derivations for
  \texorpdfstring{$\QCrealclifford$}{QCrealclifford}}
\label{appendix:real-decoding-lemmas}

\begin{derivation}\label{realczrev}
  \begin{gather*}
    \input{identities/czrev-00.tikz}
    \eqeqref{realH2}\input{identities/czrev-01.tikz}
    \eqeqref{realXC}\input{identities/czrev-02.tikz}
    \eqeqref{realH2}\input{identities/czrev-03.tikz}
  \end{gather*}
\end{derivation}

\begin{derivation}\label{lem:R4-dec-eq}
  \begin{gather*}
    \input{realcliffordczaxioms/R4-00.tikz}\\
    \eqeqref{realF}\input{realcliffordczaxioms/R4-01.tikz}
    \eqdeuxeqref{realH2}{realZ2}\input{realcliffordczaxioms/R4-02.tikz}
  \end{gather*}
\end{derivation}

\begin{derivation}\label{lem:R5-dec-eq}
  \begin{gather*}
    \input{realcliffordczaxioms/R5-01.tikz}
    \eqeqref{realH2}\input{realcliffordczaxioms/R5-02.tikz}
    \eqeqref{realCX2}\input{realcliffordczaxioms/R5-03.tikz}
    \eqeqref{realH2}\input{realcliffordczaxioms/R5-04.tikz}
  \end{gather*}
\end{derivation}

\begin{derivation}\label{lem:R6-dec-eq}
  \begin{gather*}
    \input{realcliffordczaxioms/R6-01.tikz}
    \eqeqref{realZ2}\input{realcliffordczaxioms/R6-02.tikz}
    \eqeqref{realCX2}\input{realcliffordczaxioms/R6-03.tikz}\\
    \eqeqref{realZ2}\input{realcliffordczaxioms/R6-04.tikz}
    \eqeqref{realwC}\input{realcliffordczaxioms/R6-05.tikz}\\
    \eqdeuxeqref{realZ2}{realCX2}\input{realcliffordczaxioms/R6-06.tikz}
  \end{gather*}
\end{derivation}

\begin{derivation}\label{lem:R7-dec-eq}
  \begin{gather*}
    \input{realcliffordczaxioms/R7-00b.tikz}
    \eqeqref{realczrev}\input{realcliffordczaxioms/R7-0ab.tikz}
    =\input{realcliffordczaxioms/R7-0bb.tikz}\\
    \eqeqref{lem:R6-dec-eq}\input{realcliffordczaxioms/R7-0cb.tikz}
    \eqeqref{realczrev}\input{realcliffordczaxioms/R7-01b.tikz}
  \end{gather*}
\end{derivation}

\begin{derivation}\label{lem:R8-dec-eq}
  \begin{gather*}
    \input{realcliffordczaxioms/R8-01.tikz}
    \eqdeuxeqref{realH2}{realXC}\input{realcliffordczaxioms/R8-02.tikz}
    \eqeqref{realZ2}\input{realcliffordczaxioms/R8-03.tikz}\\
    \eqeqref{realwC}\input{realcliffordczaxioms/R8-04.tikz}
    \eqdeuxeqref{realH2}{realXC}\input{realcliffordczaxioms/R8-05.tikz}
  \end{gather*}
\end{derivation}

\begin{derivation}\label{lem:R9-dec-eq}
  \begin{gather*}
    \input{realcliffordczaxioms/R9-00b.tikz}
    \eqeqref{realczrev}\input{realcliffordczaxioms/R9-0ab.tikz}
    =\input{realcliffordczaxioms/R9-0bb.tikz}\\
    \eqeqref{lem:R8-dec-eq}\input{realcliffordczaxioms/R9-0cb.tikz}
    \eqeqref{realczrev}\input{realcliffordczaxioms/R9-01b.tikz}
  \end{gather*}
\end{derivation}

\begin{derivation}\label{lem:R10-dec-eq}
  \begin{gather*}
    \input{realcliffordczaxioms/R10-01.tikz}
    \eqdeuxeqref{realH2}{realXC}\input{realcliffordczaxioms/R10-02.tikz}\\[0.2cm]
    \eqeqref{realZ2}\input{realcliffordczaxioms/R10-03.tikz}
    \eqeqref{realwC}\input{realcliffordczaxioms/R10-04.tikz}\\[0.2cm]
    \eqeqref{lem:R6-dec-eq}\input{realcliffordczaxioms/R10-05.tikz}
    \eqdeuxeqref{realH2}{realXC}\input{realcliffordczaxioms/R10-06.tikz}\\[0.2cm]
    \eqeqref{realCF}\input{realcliffordczaxioms/R10-07.tikz}
    \eqdeuxeqref{realH2}{realZ2}\input{realcliffordczaxioms/R10-08.tikz}\\[0.2cm]
    \eqdeuxeqref{realH2}{realXC}\input{realcliffordczaxioms/R10-09.tikz}
  \end{gather*}
\end{derivation}

\begin{derivation}\label{lem:R11-dec-eq}
  \begin{gather*}
    \input{realcliffordczaxioms/R11-00b.tikz}
    \eqeqref{realczrev}\input{realcliffordczaxioms/R11-0ab.tikz}\\
    = \input{realcliffordczaxioms/R11-0bb.tikz}
    \eqeqref{lem:R10-dec-eq}\input{realcliffordczaxioms/R11-0cb.tikz}\\
    \eqeqref{realczrev}\input{realcliffordczaxioms/R11-01b.tikz}
  \end{gather*}
\end{derivation}

\begin{derivation}\label{lem:R12-dec-eq}
  \begin{gather*}
    \input{realcliffordczaxioms/R12-01.tikz}\\[0.2cm]
    \eqeqref{realH2}\input{realcliffordczaxioms/R12-02.tikz}
    \eqeqref{realXC}\input{realcliffordczaxioms/R12-03.tikz}
    \eqeqref{realB}\input{realcliffordczaxioms/R12-04.tikz}\\[0.2cm]
    \eqeqref{realCX2}\input{realcliffordczaxioms/R12-05.tikz}
    \eqeqref{realH2}\input{realcliffordczaxioms/R12-06.tikz}
    =\input{realcliffordczaxioms/R12-07.tikz}
    \eqeqref{realCX2}\input{realcliffordczaxioms/R12-08.tikz}\\[0.2cm]
    \eqeqref{realB}\input{realcliffordczaxioms/R12-09.tikz}
    \eqeqref{realXC}\input{realcliffordczaxioms/R12-10.tikz}\\[0.2cm]
    \eqeqref{realH2}\input{realcliffordczaxioms/R12-11.tikz}
  \end{gather*}
\end{derivation}

\section{Auxiliary derivations for
  \texorpdfstring{$\QCqutritclifford$}{QCqutritclifford}}
\label{appendix:qutrit-decoding-lemmas}


\begin{derivation}\label{qt-K2}
\begin{gather*}
  \input{qutritderivations/K2.tikz}
  = \input{qutritderivations/K2-00.tikz}
  \eqeqref{qt-K}\input{qutritderivations/K2-01.tikz}
  = \input{qutritderivations/K2-02.tikz}
  \eqeqref{qt-hpow4}\input{qutritderivations/I.tikz}
\end{gather*}
\end{derivation}

\begin{derivation}\label{qt-sppow3}
\begin{gather*}
  \input{qutritderivations/SpSpSp.tikz}
  =\input{qutritderivations/KSKKSKKSK.tikz}
  \eqeqref{qt-K2}\input{qutritderivations/KSSSK.tikz}
  \eqeqref{qt-spow3}\input{qutritderivations/K2.tikz}\\
  \eqeqref{qt-K2}\input{qutritderivations/I.tikz}
\end{gather*}
\end{derivation}

\begin{derivation}\label{qt-zpow3}
\begin{gather*}
  \input{qutritderivations/ZZZ.tikz}
  \eqeqref{qt-Z}\input{qutritderivations/SSpSpSSpSpSSpSp.tikz}
  \eqeqref{qt-ssprime}\input{qutritderivations/SSSSpSpSpSpSpSp.tikz}
  \eqdeuxeqref{qt-spow3}{qt-sppow3}\input{qutritderivations/I.tikz}
\end{gather*}
\end{derivation}

\begin{derivation}\label{qt-xpow3}
\begin{gather*}
  \input{qutritderivations/XXX.tikz}
  \eqeqref{qt-X}\input{qutritderivations/HZHHHHZHHHHZHHH.tikz}
  \eqeqref{qt-hpow4}\input{qutritderivations/HZZZHHH.tikz}\\
  \eqeqref{qt-zpow3}\input{qutritderivations/HHHH.tikz}
  \eqeqref{qt-hpow4}\input{qutritderivations/I.tikz}
\end{gather*}
\end{derivation}

\begin{derivation}\label{qt-SK}
\begin{gather*}
  \input{qutritderivations/SK.tikz}
  \eqeqref{qt-K2}\input{qutritderivations/KKSK.tikz}
  \eqeqref{qt-Sprime}\input{qutritderivations/KSp.tikz}
\end{gather*}
\end{derivation}

\begin{derivation}\label{qt-SpK}
\begin{gather*}
  \input{qutritderivations/SpK.tikz}
  \eqeqref{qt-Sprime}\input{qutritderivations/KSKK.tikz}
  \eqeqref{qt-K2}\input{qutritderivations/KS.tikz}
\end{gather*}
\end{derivation}

\begin{derivation}\label{qt-ZK}
\begin{gather*}
  \input{qutritderivations/ZK.tikz}
  \eqeqref{qt-Z}\input{qutritderivations/SSp2K.tikz}
  \eqeqref{qt-SpK}\input{qutritderivations/SKS2.tikz}
  \eqeqref{qt-SK}\input{qutritderivations/KSpS2.tikz}\\
  \eqeqref{qt-sppow3}\input{qutritderivations/KSp2Sp2SS.tikz}
  \eqeqref{qt-ssprime}\input{qutritderivations/KSSp2SSp2.tikz}
  \eqeqref{qt-Z}\input{qutritderivations/KZZ.tikz}
  =\input{qutritderivations/KZ2.tikz}
\end{gather*}
\end{derivation}

\begin{derivation}\label{qt-XK}
\begin{gather*}
  \input{qutritderivations/XK.tikz}
  \eqeqref{qt-X}\input{qutritderivations/HZHK.tikz}
  \eqeqref{qt-K}\input{qutritderivations/HZKH.tikz}
  \eqeqref{qt-ZK}\input{qutritderivations/HKZZH.tikz}\\
  \eqeqref{qt-K}\input{qutritderivations/KHZZH.tikz}
  \eqdeuxeqref{qt-hpow4}{qt-X}\input{qutritderivations/KXX.tikz}
\end{gather*}
\end{derivation}

\begin{derivation}\label{qt-XH}
\begin{gather*}
  \input{qutritderivations/XH.tikz}
  \eqeqref{qt-X}\input{qutritderivations/HZHHHH.tikz}
  \eqeqref{qt-hpow4}\input{qutritderivations/HZ.tikz}
\end{gather*}
\end{derivation}

\begin{derivation}\label{qt-ZZH}
\begin{gather*}
  \input{qutritderivations/ZZH.tikz}
  \eqdeuxeqref{qt-hpow4}{qt-K}\input{qutritderivations/ZZKHHH.tikz}
  \eqdeuxeqref{qt-ZK}{qt-zpow3}\input{qutritderivations/KZHHH.tikz}\\
  \eqeqref{qt-K}\input{qutritderivations/HHZHHH.tikz}
  \eqeqref{qt-X}\input{qutritderivations/HX.tikz}
\end{gather*}
\end{derivation}

\begin{derivation}\label{qt-XS}
\begin{gather*}
  \input{qutritderivations/XS.tikz}
  \eqeqref{qt-X}\input{qutritderivations/HZHS.tikz}
  \eqeqref{qt-Z}\input{qutritderivations/HSSSpHHHS.tikz}
  \eqdeuxeqref{qt-hpow4}{qt-Sprime}\input{qutritderivations/HSSHHSHS.tikz}\\
  \eqeqref{qt-shpow3}\input{qutritderivations/HSSHHHHHSSHHH.tikz}
  \eqeqref{qt-hpow4}\input{qutritderivations/HHHHHSSHHHHHSSHHH.tikz}\\
  \eqeqref{qt-shpow3}\input{qutritderivations/HHSHSHHSSHHH.tikz}
  \eqdeuxeqref{qt-SK}{qt-Sprime}\input{qutritderivations/HHSHHHSpSSHHH.tikz}\\
  \eqtroiseqref{qt-shpow3}{qt-hpow4}{qt-Sprime}\input{qutritderivations/SpHSpSSHHH.tikz}
  \eqdeuxeqref{qt-Z}{qt-X}\input{qutritderivations/SpX.tikz}\\
  \eqtroiseqref{qt-Z}{qt-spow3}{qt-sppow3}\input{qutritderivations/ZSX.tikz}
  \eqeqref{qt-ssprime}\input{qutritderivations/SZX.tikz}
\end{gather*}
\end{derivation}

\begin{derivation}\label{qt-XSp}
\begin{gather*}
  \input{qutritderivations/XSp.tikz}
  \eqeqref{qt-X}\input{qutritderivations/HZHHHSp.tikz}
  \eqeqref{qt-Z}\input{qutritderivations/HSpSSHHHSp.tikz}
  \eqeqref{qt-SK}\input{qutritderivations/HSpSSHSHH.tikz}\\
  \eqeqref{qt-shpow3}\input{qutritderivations/HSpSSSSHHHSSH.tikz}
  \eqeqref{qt-spow3}\input{qutritderivations/HSpSHHHSSH.tikz}\\
  \eqeqref{qt-ssprime}\input{qutritderivations/HSSpHHHSSH.tikz}
  \eqdeuxeqref{qt-hpow4}{qt-Sprime}\input{qutritderivations/HSHHSHSSH.tikz}\\
  \eqeqref{qt-shpow3}\input{qutritderivations/SSHHHSSHSHSSH.tikz}
  \eqdeuxeqref{qt-shpow3}{qt-omegapow12}\input{qutritderivations/SSHHHSSHHHHSSHHHSH.tikz}\\
  \eqdeuxeqref{qt-spow3}{qt-hpow4}\input{qutritderivations/SSHHHSHHHSH.tikz}
  \eqdeuxeqref{qt-shpow3}{qt-omegapow12}\input{qutritderivations/SSHHSSHHHSSHHSH.tikz}\\
  \eqdeuxeqref{qt-SK}{qt-K}\input{qutritderivations/SSHHHHSpSpHSSSpHHH.tikz}
  \eqtroiseqref{qt-hpow4}{qt-Z}{qt-X}\input{qutritderivations/SpSpSSX.tikz}\\
  \eqeqref{qt-Z}\input{qutritderivations/SpZX.tikz}
\end{gather*}
\end{derivation}

\begin{derivation}\label{qt-XZ}
\begin{gather*}
  \input{qutritderivations/XZ.tikz}
  \eqeqref{qt-Z}\input{qutritderivations/XSSSp.tikz}
  \eqeqref{qt-XS}\input{qutritderivations/SZSZXSp.tikz}\\
  \eqeqref{qt-XSp}\input{qutritderivations/SZSZSpZX.tikz}
  \eqdeuxeqref{qt-Z}{qt-ssprime}\input{qutritderivations/ZZZZX.tikz}
  \eqeqref{qt-zpow3}\input{qutritderivations/ZX.tikz}
\end{gather*}
\end{derivation}

\begin{derivation}\label{qt-SXSXSX}
\begin{gather*}
  \input{qutritderivations/SXSXSX.tikz}
  \eqeqref{qt-XS}\input{qutritderivations/SSZXXSX.tikz}
  \eqeqref{qt-XS}\input{qutritderivations/SSZSZXZXX.tikz}\\
  \eqdeuxeqref{qt-ssprime}{qt-XZ}\input{qutritderivations/SSSZZZXXX.tikz}
  \eqtroiseqref{qt-xpow3}{qt-zpow3}{qt-spow3}\input{qutritderivations/I-phomega2.tikz}
\end{gather*}
\end{derivation}

\begin{derivation}\label{qt-SXKSXK}
\begin{gather*}
  \input{qutritderivations/SXKSXK.tikz}
  \eqeqref{qt-XK}\input{qutritderivations/SKXXSXK.tikz}
  \eqeqref{qt-XS}\input{qutritderivations/SKXSZXXK.tikz}\\
  \eqeqref{qt-XS}\input{qutritderivations/SKSZXZXXK.tikz}
  \eqeqref{qt-XZ}\input{qutritderivations/SKSZZXXXK.tikz}\\
  \eqeqref{qt-xpow3}\input{qutritderivations/SKSZZK.tikz}
  \eqdeuxeqref{qt-ZK}{qt-zpow3}\input{qutritderivations/SKSKZ.tikz}\\
  \eqeqref{qt-Sprime}\input{qutritderivations/SSpZ.tikz}
  \eqtroiseqref{qt-Z}{qt-ssprime}{qt-spow3}\input{qutritderivations/SpSp.tikz}
\end{gather*}
\end{derivation}

\begin{derivation}\label{qt-KSXKSX}
\begin{gather*}
  \input{qutritderivations/KSXKSX.tikz}
  \eqeqref{qt-XK}\input{qutritderivations/KSKXXSX.tikz}
  =\input{qutritderivations/SpXXSX.tikz}\\
  \eqeqref{qt-XS}\input{qutritderivations/SpXSZXX.tikz}
  \eqeqref{qt-XS}\input{qutritderivations/SpSZXZXX.tikz}
  \eqeqref{qt-XZ}\input{qutritderivations/SpSZZXXX.tikz}\\
  \eqeqref{qt-xpow3}\input{qutritderivations/SpSZZ.tikz}
  \eqtroiseqref{qt-Z}{qt-ssprime}{qt-spow3}\input{qutritderivations/SS.tikz}
\end{gather*}
\end{derivation}

\begin{derivation}\label{qt-HSXH}
\begin{gather*}
  \input{qutritderivations/SXHSX.tikz}
  \eqeqref{qt-XH}\input{qutritderivations/SHZSX.tikz}
  \eqdeuxeqref{qt-Z}{qt-ssprime}\input{qutritderivations/SHSZX.tikz}\\
  \eqeqref{qt-shpow3}\input{qutritderivations/HHHSSHHHZX.tikz}
  \eqdeuxeqref{qt-XH}{qt-XK}\input{qutritderivations/HHHSSXXHHHX.tikz}\\
  \eqtroiseqref{qt-ZZH}{qt-ZK}{qt-zpow3}\input{qutritderivations/HHHSSXXZHHH.tikz}\\
  \eqdeuxeqref{qt-XZ}{qt-omegapow12}\input{qutritderivations/HHHSSZXXHHH.tikz}\\
  \eqeqref{qt-XS}\input{qutritderivations/HHHSXSXHHH.tikz}
\end{gather*}
\end{derivation}


\begin{derivation}\label{qt-S-CNOT}
\begin{gather*}
  \input{qutritderivations/S-CNOT.tikz}
  \eqeqref{qt-spow3}\input{qutritderivations/S-CNOT-00.tikz}
  \eqeqref{qt-K2}\input{qutritderivations/S-CNOT-02.tikz}
  = \input{qutritderivations/S-CNOT-03.tikz}
  \eqeqref{qt-cnotremove}\input{qutritderivations/S-CNOT-04.tikz}\\
  = \input{qutritderivations/S-CNOT-05.tikz}
  \eqeqref{qt-cnotremove}\input{qutritderivations/S-CNOT-06.tikz}\\
  \eqeqref{qt-cnotremove}\input{qutritderivations/S-CNOT-07.tikz}
  \eqeqref{qt-cnotremove}\input{qutritderivations/S-CNOT-08.tikz}\\
  = \input{qutritderivations/S-CNOT-09.tikz}
  \eqdeuxeqref{qt-K2}{qt-spow3}\input{qutritderivations/CNOT-S.tikz}
\end{gather*}
\end{derivation}

\begin{derivation}\label{qt-CKC}
\begin{gather*}
  \input{qutritderivations/CNOT-K-CNOT.tikz}
  \eqeqref{qt-spow3}\input{qutritderivations/CNOT-K-CNOT-00.tikz}
  \eqeqref{qt-S-CNOT}\input{qutritderivations/CNOT-K-CNOT-01.tikz}
  \eqeqref{qt-cnotremove}\input{qutritderivations/CNOT-K-CNOT-02.tikz}
  \eqeqref{qt-spow3}\input{qutritderivations/I-tens-K.tikz}
\end{gather*}
\end{derivation}

\begin{derivation}\label{qt-KbCnot}
\begin{gather*}
  \input{qutritderivations/Kb-CNOT.tikz}
  \eqeqref{qt-CKC}\input{qutritderivations/Kb-CNOT-00.tikz}
  \eqeqref{qt-K2}\input{qutritderivations/Kb-CNOT-01.tikz}\\
  \eqeqref{qt-kcnot}\input{qutritderivations/Kb-CNOT-02.tikz}
  \eqeqref{qt-K2}\input{qutritderivations/Kb-CNOT-04.tikz}
  \eqeqref{qt-CKC}\input{qutritderivations/Kb-CNOT-05.tikz}\\
  \eqeqref{qt-kcnot}\input{qutritderivations/Kb-CNOT-06.tikz}
  \eqeqref{qt-K2}\input{qutritderivations/CNOT-CNOT-Kb.tikz}
\end{gather*}
\end{derivation}

\begin{derivation}\label{qt-CCC}
\begin{gather*}
  \input{qutritderivations/CNOT-CNOT-CNOT.tikz}
  \eqeqref{qt-K2}\input{qutritderivations/CNOT-CNOT-CNOT-00.tikz}
  \eqeqref{qt-KbCnot}\input{qutritderivations/CNOT-CNOT-CNOT-01.tikz}\\
  \eqeqref{qt-CKC}\input{qutritderivations/CNOT-CNOT-CNOT-02.tikz}
  \eqeqref{qt-K2}\input{identities/I2.tikz}
\end{gather*}
\end{derivation}

\begin{derivation}\label{qt-SWAP}
\begin{gather*}
  \input{gates/SWAP.tikz}
  \eqeqref{qt-CCC}\input{qutritderivations/SWAP-00.tikz}
  \eqeqref{qt-swapdecomp}\input{qutritderivations/SWAP-01.tikz}
  \eqeqref{qt-kcnot}\input{qutritderivations/SWAP-02.tikz}
\end{gather*}
\end{derivation}

\begin{derivation}\label{qt-PhGadget}
\begin{gather*}
  \input{qutritderivations/PhGadget.tikz}
  = \input{qutritderivations/PhGadget-00.tikz}\\
  \eqeqref{qt-SWAP}\input{qutritderivations/PhGadget-01.tikz}\\
  \eqeqref{qt-kcnot}\input{qutritderivations/PhGadget-02.tikz}\\
  \eqeqref{qt-CCC}\input{qutritderivations/PhGadget-03.tikz}\\
  \eqeqref{qt-S-CNOT}\input{qutritderivations/PhGadget-04.tikz}\\
  \eqeqref{qt-CKC}\input{qutritderivations/PhGadget-05.tikz}
  \eqeqref{qt-CKC}\input{qutritderivations/PhGadget-06.tikz}\\
  \eqeqref{qt-kcnot}\input{qutritderivations/PhGadget-07.tikz}
  \eqeqref{qt-K2}\input{qutritderivations/PhGadget-08.tikz}
\end{gather*}
\end{derivation}

\begin{derivation}\label{qt-CzRev}
\begin{gather*}
  \input{qutritderivations/CZ-rev.tikz}
  = \input{qutritderivations/CZ-rev-00.tikz}
  \eqeqref{qt-czdecomp}\input{qutritderivations/CZ-rev-01.tikz}\\
  \eqeqref{qt-PhGadget}\input{qutritderivations/CZ-rev-02.tikz}
  \eqeqref{qt-czdecomp}\input{qutritderivations/CZ-rev-03.tikz}
  = \input{qutritderivations/CZ-rev-04.tikz}
\end{gather*}
\end{derivation}

\begin{derivation}\label{qt-S-NOTC}
\begin{gather*}
  \input{qutritderivations/S-NOTC.tikz}
  \eqeqref{qt-K2}\input{qutritderivations/S-NOTC-00.tikz}
  \eqeqref{qt-CCC}\input{qutritderivations/S-NOTC-01.tikz}\\
  \eqeqref{qt-kcnot}\input{qutritderivations/S-NOTC-02.tikz}
  \eqeqref{qt-spow3}\input{qutritderivations/S-NOTC-03.tikz}\\
  \eqeqref{qt-czdecomp}\input{qutritderivations/S-NOTC-04.tikz}
  \eqeqref{qt-S-CNOT}\input{qutritderivations/S-NOTC-05.tikz}\\
  \eqeqref{qt-PhGadget}\input{qutritderivations/S-NOTC-06.tikz}
  \eqeqref{qt-S-CNOT}\input{qutritderivations/S-NOTC-07.tikz}\\
  \eqeqref{qt-K}\input{qutritderivations/S-NOTC-08.tikz}
  \eqeqref{qt-Sdagger}\input{qutritderivations/S-NOTC-09.tikz}
\end{gather*}
\end{derivation}

\begin{derivation}\label{qt-CZ-CNOT}
\begin{gather*}
  \input{qutritderivations/CZ-CNOT.tikz}
  \eqeqref{qt-KbCnot}\input{qutritderivations/CZ-CNOT-00.tikz}\\
  \eqdeuxeqref{qt-hpow4}{qt-H3}\input{qutritderivations/CZ-CNOT-01.tikz}
  \eqeqref{qt-spow3}\input{qutritderivations/CZ-CNOT-02.tikz}\\
  \eqeqref{qt-S-CNOT}\input{qutritderivations/CZ-CNOT-03.tikz}
  \eqeqref{qt-CzRev}\input{qutritderivations/CZ-CNOT-04.tikz}\\
  \eqeqref{qt-czdecomp}\input{qutritderivations/CZ-CNOT-05.tikz}\\
  \eqeqref{qt-spow3}\input{qutritderivations/CZ-CNOT-06.tikz}\\
  \eqeqref{qt-S-CNOT}\input{qutritderivations/CZ-CNOT-07.tikz}
  \eqeqref{qt-CCC}\input{qutritderivations/CZ-CNOT-08.tikz}\\
  \eqeqref{qt-S-NOTC}\input{qutritderivations/CZ-CNOT-09.tikz}
  \eqeqref{qt-spow3}\input{qutritderivations/CZ-CNOT-10.tikz}
\end{gather*}
\end{derivation}

\begin{derivation}\label{qt-SwCZ}
\begin{gather*}
  \input{qutritderivations/SWAP-CZ-00.tikz}
  \eqdeuxeqref{qt-CzRev}{qt-hpow4}\input{qutritderivations/SWAP-CZ-01.tikz}\\
  \eqeqref{qt-KbCnot}\input{qutritderivations/SWAP-CZ-02.tikz}
  \eqeqref{qt-swapdecomp}\input{qutritderivations/SWAP-CZ-03.tikz}\\
  \eqeqref{qt-KbCnot}\input{qutritderivations/SWAP-CZ-04.tikz}
  \eqeqref{qt-CCC}\input{qutritderivations/SWAP-CZ-05.tikz}
  \eqeqref{qt-hpow4}\input{gates/SWAP.tikz}
\end{gather*}
\end{derivation}

\begin{derivation}\label{qt-CComm}
\begin{gather*}
  \input{qutritderivations/CNOT13CNOT23.tikz}
  \eqeqref{qt-CCC}\input{qutritderivations/CNOT12CNOT12CNOT23CNOT23CNOT23CNOT12CNOT13CNOT23.tikz}
  \eqeqref{qt-I}\input{qutritderivations/CNOT12CNOT12CNOT23CNOT23CNOT12CNOT23CNOT23.tikz}\\
  \eqdeuxeqref{qt-K2}{qt-KbCnot}\input{qutritderivations/CNOT12CNOT12CNOT23CNOT12CNOT23.tikz}
  \eqeqref{qt-I}\input{qutritderivations/CNOT12CNOT12CNOT23CNOT23CNOT12CNOT13.tikz}\\
  \eqtroiseqref{qt-K2}{qt-KbCnot}{qt-CCC}\input{qutritderivations/CNOT12CNOT12CNOT23CNOT12CNOT13CNOT13.tikz}
  \eqeqref{qt-I}\input{qutritderivations/CNOT12CNOT12CNOT12CNOT23CNOT13.tikz}
  \eqeqref{qt-CCC}\input{qutritderivations/CNOT23CNOT13.tikz}
\end{gather*}
\end{derivation}

\begin{derivation}\label{qt-CzComm}
\begin{gather*}
  \input{qutritderivations/CZ12CZ23.tikz}
  \eqeqref{qt-CzRev}\input{qutritderivations/CZ12CZ23r.tikz}
  \eqeqref{qt-hpow4}\input{qutritderivations/HCNOTNOTCH.tikz}\\
  \eqeqref{qt-CComm}\input{qutritderivations/HNOTCCNOTH.tikz}
  \eqeqref{qt-hpow4}\input{qutritderivations/CZ23rCZ12.tikz}
  \eqeqref{qt-CzRev}\input{qutritderivations/CZ23CZ12.tikz}
\end{gather*}
\end{derivation}


\begin{derivation}\label{lem:qt-C4-dec-eq}
\begin{gather*}
  \input{qutritcliffordaxioms/s2hs2hs2hphases.tikz}
  \eqdeuxeqref{qt-spow3}{qt-omegapow12}\input{qutritderivations/shshsh.tikz}\\
  \eqeqref{qt-shpow3}\input{qutritderivations/shshsh-01.tikz}
  \eqdeuxeqref{qt-spow3}{qt-hpow4}\input{qutritcliffordaxioms/D-Iphase-omega.tikz}
\end{gather*}
\end{derivation}

\begin{proposition}\label{qt-prop:SSpnecessity}
    The axiom \eqref{qt-ssprime} is necessary in $\QCqutritclifford$.
\end{proposition}
\begin{proof}
    The interpretation $\interp{\cdot}_{\mathit{Cliff3},\sim}^{\gS:=\gSX}$ respects all axioms of $\QCqutritclifford$ acting on at most one qutrit except \cref{qt-ssprime}. The checks are not merely syntactic: for instance, \eqref{qt-spow3} and \eqref{qt-shpow3} remain sound because the relevant images reduce as follows.
    \begin{gather*}
    \interp{\tikzfigM{./qutritcliffordaxioms/spow3}}_{\mathit{Cliff3},\sim}^{\gS:=\gSX}
    = \interp{\tikzfigM{./qutritderivations/SXSXSX}} 
    \eqeqref{qt-SXSXSX} \interp{\tikzfigM{./qutritderivations/I-phomega2}} \\
    \sim \interp{\tikzfigM{./cliffordaxioms/I}} 
    = \interp{\tikzfigM{./cliffordaxioms/I}}_{\mathit{Cliff3},\sim}^{\gS:=\gSX}
    \end{gather*}
    
    \begin{gather*}
    \interp{\tikzfigM{./qutritcliffordaxioms/shs}}_{\mathit{Cliff3},\sim}^{\gS:=\gSX}
    = \interp{\tikzfigM{./qutritderivations/SXHSX}} 
    \eqeqref{qt-HSXH} \interp{\tikzfigM{./qutritderivations/HHHSXSXHHH}} \\
    \sim \interp{\tikzfigM{./qutritcliffordaxioms/h3s2h3}}_{\mathit{Cliff3},\sim}^{\gS:=\gSX}
    \end{gather*}

    The axiom \eqref{qt-ssprime} is the point where the interpretation separates the theory.  Its two sides are sent to inequivalent elements of $\cat{Qudit}_3^{\sim}(1,1)$:
    \begin{eqnarray*}
        \interp{\tikzfigM{./qutritcliffordaxioms/ssprime}}_{\mathit{Cliff3},\sim}^{\gS:=\gSX}
        &=&
        \interp{\tikzfigM{./qutritderivations/SXKSXK}}
        \;\;\eqeqref{qt-SXKSXK}\;\;
        \interp{\tikzfigM{./qutritderivations/SpSp}},\\
        \interp{\tikzfigM{./qutritcliffordaxioms/sprimes}}_{\mathit{Cliff3},\sim}^{\gS:=\gSX}
        &=&
        \interp{\tikzfigM{./qutritderivations/KSXKSX}}
        \;\;\eqeqref{qt-KSXKSX}\;\;
        \interp{\tikzfigM{./qutritderivations/SS}}.
    \end{eqnarray*}

    All other axioms acting on at most one qutrit remain sound under this interpretation.  Removing \cref{qt-ssprime} would therefore identify too little and leave the resulting theory incomplete.
\end{proof}
\section{Auxiliary derivations for
  \texorpdfstring{$\QCcliffordplust$}{QCcliffordplust}}
\label{appendix:t-decoding-lemmas}


\begin{derivation}\label{t-Z2}
\begin{gather*}
    \input{identities/ZZ-00.tikz}
    \eqeqref{t-Z}\input{identities/ZZ-01.tikz}
    \eqeqref{t-S}\input{identities/ZZ-01b.tikz}
    \eqeqref{t-Tpow8}\input{identities/ZZ-02.tikz}
\end{gather*}
\end{derivation}

\begin{derivation}\label{t-X2}
\begin{gather*}
    \input{identities/XX-00.tikz}
    \eqeqref{t-X}\input{identities/XX-01.tikz}
    \eqeqref{t-H2}\input{identities/XX-02.tikz}
    \eqeqref{t-Z2}\input{identities/XX-03.tikz}
    \eqeqref{t-H2}\input{identities/XX-04.tikz}
\end{gather*}
\end{derivation}

\begin{derivation}\label{t-HX}
\begin{gather*}
    \input{identities/HX.tikz}
    \eqeqref{t-X}\input{identities/HHZH.tikz}
    \eqeqref{t-H2}\input{identities/ZH.tikz}
\end{gather*}
\end{derivation}

\begin{derivation}\label{t-XSX}
\begin{gather*}
    \input{identities/XSX-00.tikz}
    \eqeqref{t-X}\input{identities/XSX-01.tikz}\\
    \eqeqref{t-E}\input{identities/XSX-02.tikz}
    \eqeqref{t-Tpow8}\input{identities/XSX-03.tikz}\\
    \eqeqref{t-E}\input{identities/XSX-04.tikz}
    \eqdeuxeqref{t-Tpow8}{t-H2}\input{identities/XSX-05.tikz}
\end{gather*}
\end{derivation}

\begin{derivation}\label{t-CX2}
\begin{gather*}
    \input{cliffordaxioms/CNOTCNOT.tikz}
    \eqeqref{t-Tpow8}\input{identities/CNOTCNOT-01b.tikz}
    \eqeqref{t-Ct}\input{identities/CNOTCNOT-02b.tikz}\\
    \eqeqref{t-B}\input{identities/CNOTCNOT-03b.tikz}
    =\input{identities/CNOTCNOT-04b.tikz}\\
    \eqeqref{t-B}\input{identities/CNOTCNOT-05b.tikz}
    =\input{identities/CNOTCNOT-06b.tikz}\\
    =\input{identities/CNOTCNOT-07b.tikz}
    \eqeqref{t-Ct}\input{identities/CNOTCNOT-08b.tikz}
    \eqeqref{t-Tpow8}\input{identities/I2.tikz}
\end{gather*}
\end{derivation}

\begin{derivation}\label{t-C}
\begin{gather*}
    \input{cliffordplustaxioms/TCNOT.tikz}
    \eqeqref{t-Ct}\input{cliffordplustaxioms/TCNOT-01.tikz}
    \eqeqref{t-CX2}\input{cliffordplustaxioms/CNOTT.tikz}
\end{gather*}
\end{derivation}

\begin{derivation}\label{t-phasegadget}
\begin{gather*}
    \input{identities/phasegadget-00.tikz}
    =\input{identities/phasegadget-01.tikz}
    \eqeqref{t-B}\input{identities/phasegadget-02.tikz}\\
    \eqeqref{t-C}\input{identities/phasegadget-03.tikz}
    \eqeqref{t-CX2}\input{identities/phasegadget-04.tikz}
\end{gather*}
\end{derivation}

\begin{derivation}\label{t-phasegadgetCS}
\begin{gather*}
    \input{identities/phasegadgetCS-00.tikz}
    =\input{identities/phasegadgetCS-01.tikz}
    \eqeqref{t-B}\input{identities/phasegadgetCS-02.tikz}\\
    \eqeqref{t-C}\input{identities/phasegadgetCS-03.tikz}
    \eqeqref{t-CX2}\input{identities/phasegadgetCS-04.tikz}
\end{gather*}
\end{derivation}

\begin{derivation}\label{t-CNOTHH}
\begin{gather*}
    \input{identities/CNOTHH-00.tikz}
    \eqeqref{t-H2}\input{identities/CNOTHH-01.tikz}
    \eqeqref{t-CZ}\input{identities/CNOTHH-02.tikz}\\
    \eqeqref{t-phasegadget}\input{identities/CNOTHH-03.tikz}
    \eqeqref{t-CZ}\input{identities/CNOTHH-04.tikz}
    \eqeqref{t-H2}\input{identities/CNOTHH-05.tikz}
\end{gather*}
\end{derivation}

\begin{derivation}\label{t-czrev}
\begin{gather*}
    \input{identities/czrev-00.tikz}
    \eqeqref{t-H2}\input{identities/czrev-01.tikz}
    \eqeqref{t-CZ}\input{identities/czrev-02.tikz}
    \eqeqref{t-H2}\input{identities/czrev-03.tikz}
\end{gather*}
\end{derivation}

\begin{derivation}\label{t-XcommutCNOT}
\begin{gather*}
    \input{identities/XcommutCNOT-00.tikz}
    \eqtroiseqref{t-H2}{t-X}{t-Z}\input{identities/XcommutCNOT-01.tikz}
    \eqeqref{t-CNOTHH}\input{identities/XcommutCNOT-02.tikz}\\
    \eqeqref{t-C}\input{identities/XcommutCNOT-03.tikz}
    \eqeqref{t-CNOTHH}\input{identities/XcommutCNOT-04.tikz}
    \eqtroiseqref{t-H2}{t-X}{t-Z}\input{identities/XcommutCNOT-05.tikz}
\end{gather*}
\end{derivation}

\begin{derivation}\label{t-wCNOT}
\begin{gather*}
    \input{identities/wCNOT-00.tikz}
    \eqeqref{t-H2}\input{identities/wCNOT-01.tikz}
    \eqeqref{t-CZ}\input{identities/wCNOT-02.tikz}\\
    \eqeqref{t-phasegadget}\input{identities/wCNOT-03.tikz}
    \eqeqref{t-X}\input{identities/wCNOT-04.tikz}\\
    \eqdeuxeqref{t-Z}{t-C}\input{identities/wCNOT-05.tikz}
    \eqdeuxeqref{t-Z}{t-Sdagger}\input{identities/wCNOT-06.tikz}\\
    \eqeqref{t-XSX}\input{identities/wCNOT-07.tikz}\\
    \eqeqref{t-w8}\input{identities/wCNOT-08.tikz}\\
    \eqdeuxeqref{t-XcommutCNOT}{t-X2}\input{identities/wCNOT-09.tikz}
    \eqeqref{t-XcommutCNOT}\input{identities/wCNOT-10.tikz}\\
    \eqeqref{t-phasegadget}\input{identities/wCNOT-11.tikz}\\
    \eqeqref{t-CX2}\input{identities/wCNOT-12.tikz}\\
    \eqdeuxeqref{t-C}{t-phasegadget}\input{identities/wCNOT-13.tikz}\\
    \eqeqref{t-CZ}\input{identities/wCNOT-14.tikz}\\
    \eqeqref{t-CZ}\input{identities/wCNOT-15.tikz}\\
    \eqtroiseqref{t-Tpow8}{t-CX2}{t-C}\input{identities/wCNOT-16.tikz}
\end{gather*}
\end{derivation}

\begin{derivation}\label{t-zCNOT}
\begin{gather*}
    \input{identities/zCNOT-00.tikz}
    \eqdeuxeqref{t-H2}{t-CNOTHH}\input{identities/zCNOT-01.tikz}
    \eqeqref{t-X}\input{identities/zCNOT-02.tikz}
    \eqeqref{t-wCNOT}\input{identities/zCNOT-03.tikz}\\
    \eqeqref{t-X}\input{identities/zCNOT-04.tikz}
    \eqdeuxeqref{t-H2}{t-CNOTHH}\input{identities/zCNOT-05.tikz}
\end{gather*}
\end{derivation}

\begin{derivation}\label{t-HCh}
\begin{gather*}
    \input{cliffordplustaxioms/HCh-00.tikz}
    \eqeqref{t-ChWhite}\input{cliffordplustaxioms/HCh-01.tikz}
    \eqeqref{t-Ch}\input{cliffordplustaxioms/HCh-02.tikz}\\
    \eqeqref{t-E}\input{cliffordplustaxioms/HCh-03.tikz}
    \eqeqref{t-Tpow8}\input{cliffordplustaxioms/HCh-04.tikz}\\
    \eqeqref{t-Z}\input{cliffordplustaxioms/HCh-05.tikz}
    \eqeqref{t-HX}\input{cliffordplustaxioms/HCh-06.tikz}\\
    \eqeqref{t-TX}\input{cliffordplustaxioms/HCh-07.tikz}
    \eqeqref{t-XcommutCNOT}\input{cliffordplustaxioms/HCh-08.tikz}
    \eqeqref{t-Ch}\input{cliffordplustaxioms/HCh-09.tikz}\\
    \eqeqref{t-H2}\input{cliffordplustaxioms/HCh-10.tikz}
    \eqeqref{t-ChWhite}\input{cliffordplustaxioms/HCh-11.tikz}
\end{gather*}
\end{derivation}

\begin{derivation}\label{t-CNOTCH}
\begin{gather*}
    \input{cliffordplustaxioms/CNOTCH-00.tikz}
    \eqeqref{t-H2}\input{cliffordplustaxioms/CNOTCH-01.tikz}
    \eqeqref{t-HCh}\input{cliffordplustaxioms/CNOTCH-02.tikz}\\
    \eqeqref{t-CZ}\input{cliffordplustaxioms/CNOTCH-03.tikz}
    \eqeqref{t-S}\input{cliffordplustaxioms/CNOTCH-04.tikz}\\
    \eqdeuxeqref{t-S}{t-Tpow8}\input{cliffordplustaxioms/CNOTCH-05.tikz}\\
    \eqdeuxeqref{t-Sdagger}{t-Tpow8}\input{cliffordplustaxioms/CNOTCH-06.tikz}\\
    \eqdeuxeqref{t-C}{t-phasegadgetCS}\input{cliffordplustaxioms/CNOTCH-07.tikz}
    \eqeqref{t-Cs}\input{cliffordplustaxioms/CNOTCH-08.tikz}\\
    \eqeqref{t-CsCh}\input{cliffordplustaxioms/CNOTCH-09.tikz}
    \eqeqref{t-HCh}\input{cliffordplustaxioms/CNOTCH-10.tikz}\\
    \eqeqref{t-Cs}\input{cliffordplustaxioms/CNOTCH-11.tikz}\\
    \eqtroiseqref{t-C}{t-phasegadgetCS}{t-S}\input{cliffordplustaxioms/CNOTCH-12.tikz}\\
    \eqeqref{t-CX2}\input{cliffordplustaxioms/CNOTCH-13.tikz}\\
    \eqtroiseqref{t-S}{t-Tpow8}{t-Sdagger}\input{cliffordplustaxioms/CNOTCH-14.tikz}\\
    \eqeqref{t-CZ}\input{cliffordplustaxioms/CNOTCH-15.tikz}
    \eqeqref{t-H2}\input{cliffordplustaxioms/CNOTCH-16.tikz}
\end{gather*}
\end{derivation}

\begin{derivation}\label{t-CtrlH}
\begin{gather*}
    \input{cliffordplustaxioms/CtrlH.tikz}
    \eqeqref{t-ChWhite}\input{cliffordplustaxioms/CtrlH-00.tikz}
    \eqeqref{t-Ch}\input{cliffordplustaxioms/CtrlH-01.tikz}\\
    \eqeqref{t-E}\input{cliffordplustaxioms/CtrlH-02.tikz}\\
    \eqdeuxeqref{t-Z}{t-Tpow8}\input{cliffordplustaxioms/CtrlH-03.tikz}\\
    \eqeqref{t-HX}\input{cliffordplustaxioms/CtrlH-04.tikz}\\
    \eqeqref{t-TX}\input{cliffordplustaxioms/CtrlH-05.tikz}
    \eqdeuxeqref{t-Sdagger}{t-S}\input{cliffordplustaxioms/CtrlH-06.tikz}\\
    \eqeqref{t-CnotWhite}\input{cliffordplustaxioms/CtrlHFinal.tikz}
\end{gather*}
\end{derivation}

\begin{derivation}\label{t-XCHX}
\begin{gather*}
    \input{cliffordplustaxioms/XCHX.tikz}
    \eqeqref{t-CtrlH}\input{cliffordplustaxioms/XCHX-00.tikz}
    \eqeqref{t-CnotWhite}\input{cliffordplustaxioms/XCHX-01.tikz}\\
    \eqdeuxeqref{t-wCNOT}{t-X2}\input{shortcut/ChDecomp.tikz}
    \eqeqref{t-Ch}\input{shortcut/Ch.tikz}
\end{gather*}
\end{derivation}

\begin{derivation}\label{t-CHCH}
\begin{gather*}
    \input{cliffordplustaxioms/CHCH_00.tikz}
    \eqeqref{t-Ch}\input{cliffordplustaxioms/CHCH_01.tikz}\\
    \eqtroiseqref{t-S}{t-Tpow8}{t-H2}\input{cliffordplustaxioms/CHCH_02.tikz}
    \eqeqref{t-CX2}\input{cliffordplustaxioms/CHCH_03.tikz}
    \eqtroiseqref{t-S}{t-Tpow8}{t-H2}\input{identities/I2.tikz}
\end{gather*}
\end{derivation}

\begin{derivation}\label{t-CSCSCSCS}
\begin{gather*}
    \input{cliffordplustaxioms/CSCSCSCS_00.tikz}
    \eqeqref{t-Cs}\input{cliffordplustaxioms/CSCSCSCS_01.tikz}\\
    \eqdeuxeqref{t-C}{t-Z}\input{cliffordplustaxioms/CSCSCSCS_02.tikz}\\
    \eqeqref{t-phasegadgetCS}\input{cliffordplustaxioms/CSCSCSCS_03.tikz}\\
    \eqdeuxeqref{t-C}{t-Z}\input{cliffordplustaxioms/CSCSCSCS_04.tikz}\\
    \eqeqref{t-CX2}\input{cliffordplustaxioms/CSCSCSCS_05.tikz}
    \eqdeuxeqref{t-Tpow8}{t-Z}\input{cliffordplustaxioms/CSCSCSCS_06.tikz}
    \eqdeuxeqref{t-zCNOT}{t-Z2}\input{cliffordplustaxioms/CSCSCSCS_07.tikz}
    \eqeqref{t-CX2}\input{identities/I2.tikz}
\end{gather*}
\end{derivation}

\begin{derivation}\label{t-CSCS}
\begin{gather*}
    \input{cliffordplustaxioms/CSCS_00.tikz}
    \eqeqref{t-Cs}\input{cliffordplustaxioms/CSCS_01.tikz}
    \eqdeuxeqref{t-C}{t-S}\input{cliffordplustaxioms/CSCS_02.tikz}\\
    \eqeqref{t-phasegadgetCS}\input{cliffordplustaxioms/CSCS_03.tikz}
    \eqdeuxeqref{t-C}{t-S}\input{cliffordplustaxioms/CSCS_04.tikz}\\
    \eqeqref{t-CX2}\input{cliffordplustaxioms/CSCS_05.tikz}\\
    \eqdeuxeqref{t-Tpow8}{t-S}\input{cliffordplustaxioms/CSCS_06.tikz}
    \eqeqref{t-phasegadget}\input{cliffordplustaxioms/CSCS_07.tikz}
    \eqeqref{t-CZ}\input{cliffordplustaxioms/CSCS_08.tikz}
\end{gather*}
\end{derivation}

\begin{derivation}\label{t-SdCS}
\begin{gather*}
    \input{cliffordplustaxioms/SdCS_00.tikz}
    \eqeqref{t-Cs}\input{cliffordplustaxioms/SdCS_01.tikz}
    \eqeqref{t-Tdagger}\input{cliffordplustaxioms/SdCS_02.tikz}
    \eqeqref{t-phasegadgetCS}\input{cliffordplustaxioms/SdCS_03.tikz}\\
    \eqeqref{t-XcommutCNOT}\input{cliffordplustaxioms/SdCS_04.tikz}
    \eqdeuxeqref{t-TX}{t-X2}\input{cliffordplustaxioms/SdCS_05.tikz}\\
    \eqdeuxeqref{t-Z}{t-C}\input{cliffordplustaxioms/SdCS_06.tikz}\\
    \eqeqref{t-zCNOT}\input{cliffordplustaxioms/SdCS_07.tikz}
    \eqdeuxeqref{t-C}{t-Tdagger}\input{cliffordplustaxioms/SdCS_08.tikz}\\
    \eqeqref{t-CX2}\input{cliffordplustaxioms/SdCS_09.tikz}\\
    \eqeqref{t-C}\input{cliffordplustaxioms/SdCS_10.tikz}\\
    \eqeqref{t-phasegadgetCS}\input{cliffordplustaxioms/SdCS_11.tikz}\\
    \eqeqref{t-C}\input{cliffordplustaxioms/SdCS_12.tikz}
    \eqdeuxeqref{t-Cs}{t-Csdagger}\input{cliffordplustaxioms/SdCS_13.tikz}
\end{gather*}
\end{derivation}

\begin{derivation}\label{t-SCSd}
\begin{gather*}
    \input{cliffordplustaxioms/SCSd_00.tikz}
    \eqdeuxeqref{t-Csdagger}{t-Tpow8}\input{cliffordplustaxioms/SCSd_01.tikz}
    \eqtroiseqref{t-C}{t-Cs}{t-phasegadgetCS}\input{cliffordplustaxioms/SCSd_02.tikz}\\
    \eqeqref{t-SdCS}\input{cliffordplustaxioms/SCSd_03.tikz}
    \eqeqref{t-X2}\input{cliffordplustaxioms/SCSd_04.tikz}\\
    \eqdeuxeqref{t-Csdagger}{t-CSCSCSCS}\input{cliffordplustaxioms/SCSd_05.tikz}
\end{gather*}
\end{derivation}

\begin{derivation}\label{t-HSTH}
\begin{gather*}
    \input{cliffordplustaxioms/HSTH_00.tikz}
    \eqeqref{t-CSCSCSCS}\input{cliffordplustaxioms/HSTH_01.tikz}\\
    \eqeqref{t-SCSd}\input{cliffordplustaxioms/HSTH_02.tikz}\\
    \eqeqref{t-XCHX}\input{cliffordplustaxioms/HSTH_03.tikz}
    \eqeqref{t-CsCh}\input{cliffordplustaxioms/HSTH_04.tikz}\\
    \eqeqref{t-XCHX}\input{cliffordplustaxioms/HSTH_05.tikz}\\
    \eqtroiseqref{t-Cs}{t-phasegadgetCS}{t-C}\input{cliffordplustaxioms/HSTH_06.tikz}
    \eqeqref{t-CsCh}\input{cliffordplustaxioms/HSTH_07.tikz}\\
    \eqtroiseqref{t-Cs}{t-phasegadgetCS}{t-C}\input{cliffordplustaxioms/HSTH_08.tikz}
    \eqeqref{t-SdCS}\input{cliffordplustaxioms/HSTH_09.tikz}\\
    \eqeqref{t-HTHT}\input{cliffordplustaxioms/HSTH_10.tikz}
    \eqeqref{t-XCHX}\input{cliffordplustaxioms/HSTH_11.tikz}\\
    \eqeqref{t-CsCh}\input{cliffordplustaxioms/HSTH_12.tikz}
    \eqeqref{t-XCHX}\input{cliffordplustaxioms/HSTH_13.tikz}\\
    \eqdeuxeqref{t-SdCS}{t-SCSd}\input{cliffordplustaxioms/HSTH_14.tikz}\\
    \eqtroiseqref{t-Cs}{t-phasegadgetCS}{t-C}\input{cliffordplustaxioms/HSTH_15.tikz}
    \eqeqref{t-CsCh}\input{cliffordplustaxioms/HSTH_16.tikz}\\
    \eqeqref{t-CSCSCSCS}\input{cliffordplustaxioms/HSTH_17.tikz}
\end{gather*}
\end{derivation}


\begin{derivation}\label{lem:t-C6-dec-eq}
\begin{gather*}
    \input{cliffordplustaxioms/TTHTTHTTH.tikz}
    \eqeqref{t-E}\input{cliffordplustaxioms/TTHHTTTTTTHH.tikz}
    \eqdeuxeqref{t-H2}{t-Tpow8}\input{cliffordczaxioms/C4-03.tikz}
\end{gather*}
\end{derivation}

\begin{derivation}\label{lem:t-C7-dec-eq}
\begin{gather*}
    \input{cliffordplustaxioms/C7-00.tikz}
    \eqeqref{t-H2}\input{cliffordplustaxioms/C7-01.tikz}
    \eqeqref{t-CX2}\input{cliffordplustaxioms/C7-02.tikz}
    \eqeqref{t-H2}\input{cliffordczaxioms/C5-04.tikz}
\end{gather*}
\end{derivation}

\begin{derivation}\label{lem:t-C10-dec-eq}
\begin{gather*}
    \input{cliffordczaxioms/C8-01.tikz}
    \eqeqref{t-X}\input{cliffordczaxioms/C8-02.tikz}
    \eqeqref{t-X2}\input{cliffordczaxioms/C8-03.tikz}\\
    \eqeqref{t-wCNOT}\input{cliffordczaxioms/C8-04.tikz}
    \eqdeuxeqref{t-X}{t-Z}\input{cliffordczaxioms/C8-05.tikz}\\
    \eqeqref{t-X}\input{cliffordczaxioms/C8-06.tikz}
    \eqeqref{t-Z}\input{cliffordczaxioms/C8-07.tikz}
\end{gather*}
\end{derivation}

\begin{derivation}\label{lem:t-C11-dec-eq}
\begin{gather*}
    \input{cliffordczaxioms/C9-01.tikz}
    \eqeqref{t-czrev}\input{cliffordczaxioms/C9-01b.tikz}\\
    =\input{cliffordczaxioms/C9-0ab.tikz}
    \eqeqref{lem:t-C10-dec-eq}\input{cliffordczaxioms/C9-0bb.tikz}\\
    \eqeqref{t-czrev}\input{cliffordczaxioms/C9-0cb.tikz}
    =\input{cliffordczaxioms/C9-08.tikz}
\end{gather*}
\end{derivation}

\begin{derivation}\label{lem:t-C12-dec-eq}
\begin{gather*}
    \input{cliffordczaxioms/C10-01.tikz}
    \eqeqref{t-E}\input{cliffordczaxioms/C10-02.tikz}\\
    \eqeqref{t-C}\input{cliffordczaxioms/C10-03.tikz}
    \eqeqref{t-CNOTHH}\input{cliffordczaxioms/C10-04.tikz}\\
    \eqeqref{t-H2}\input{cliffordczaxioms/C10-05.tikz}
    \eqeqref{t-Tpow8}\input{cliffordczaxioms/C10-06.tikz}\\
    \eqeqref{t-C}\input{cliffordczaxioms/C10-07.tikz}
    \eqeqref{t-phasegadget}\input{cliffordczaxioms/C10-08.tikz}\\
    \eqeqref{t-zCNOT}\input{cliffordczaxioms/C10-09.tikz}\\
    \eqdeuxeqref{t-Z}{t-Sdagger}\input{cliffordczaxioms/C10-10.tikz}\\
    \eqeqref{t-Z}\input{cliffordczaxioms/C10-11.tikz}\\
    \eqeqref{t-C}\input{cliffordczaxioms/C10-12.tikz}\\
    \eqeqref{t-E}\input{cliffordczaxioms/C10-13.tikz}
\end{gather*}
\end{derivation}

\begin{derivation}\label{lem:t-C13-dec-eq}
\begin{gather*}
    \input{cliffordczaxioms/C11-00b.tikz}
    \eqeqref{t-czrev}\input{cliffordczaxioms/C11-0ab.tikz}\\
    \eqeqref{lem:t-C12-dec-eq}\input{cliffordczaxioms/C11-0bb.tikz}\\
    \eqeqref{t-czrev}\input{cliffordczaxioms/C11-0cb.tikz}\\
    =\input{cliffordczaxioms/C11-01b.tikz}
\end{gather*}
\end{derivation}

\begin{derivation}\label{lem:t-C17-dec-eq}
\begin{gather*}
    \input{cliffordplustaxioms/C17-01.tikz}
    \eqeqref{t-H2}\input{cliffordplustaxioms/C17-02.tikz}\\
    \eqeqref{t-czrev}\input{cliffordplustaxioms/C17-03.tikz}
    \eqeqref{t-H2}\input{cliffordplustaxioms/C17-04.tikz}\\
    \eqeqref{t-B}\input{cliffordplustaxioms/C17-05.tikz}
    \eqeqref{t-C}\input{cliffordplustaxioms/C17-06.tikz}
    \eqeqref{t-B}\input{cliffordplustaxioms/C17-07.tikz}\\
    \eqeqref{t-H2}\input{cliffordplustaxioms/C17-08.tikz}
    \eqeqref{t-czrev}\input{cliffordplustaxioms/C17-09.tikz}
\end{gather*}
\end{derivation}

\begin{derivation}\label{lem:t-C18-dec-eq}
\begin{gather*}
    \input{cliffordplustaxioms/C18l.tikz}
    \eqdeuxeqref{t-XcommutCNOT}{t-X2}\input{cliffordplustaxioms/C18-01.tikz}\\
    \eqeqref{t-TX}\input{cliffordplustaxioms/C18-02.tikz}\\
    \eqeqref{t-HX}\input{cliffordplustaxioms/C18-02b.tikz}\\
    \eqtroiseqref{t-Z}{t-Tpow8}{t-CtrlH}\input{cliffordplustaxioms/C18-03.tikz}
    \eqeqref{t-Z}\input{cliffordplustaxioms/C18-04.tikz}\\
    \eqeqref{t-zCNOT}\input{cliffordplustaxioms/C18-05.tikz}
    \eqeqref{t-CX2}\input{cliffordplustaxioms/C18-06.tikz}\\
    \eqeqref{t-Tpow8}\input{cliffordplustaxioms/C18-07.tikz}
    \eqeqref{t-Cs}\input{cliffordplustaxioms/C18-08.tikz}\\
    \eqdeuxeqref{t-CsCh}{t-CNOTCH}\input{cliffordplustaxioms/C18-09.tikz}
    \eqeqref{t-C}\input{cliffordplustaxioms/C18-10.tikz}\\
    \eqeqref{t-Cs}\input{cliffordplustaxioms/C18-11.tikz}
    \eqdeuxeqref{t-Tpow8}{t-CX2}\input{cliffordplustaxioms/C18-12.tikz}\\
    \eqtroiseqref{t-Tpow8}{t-Z}{t-Sdagger}\input{cliffordplustaxioms/C18-13.tikz}
    \eqeqref{t-zCNOT}\input{cliffordplustaxioms/C18-14.tikz}\\
    \eqeqref{t-CtrlH}\input{cliffordplustaxioms/C18-15.tikz}\\
    \eqdeuxeqref{t-Z}{t-Tpow8}\input{cliffordplustaxioms/C18-16.tikz}\\
    \eqeqref{t-HX}\input{cliffordplustaxioms/C18-17.tikz}
    \eqeqref{t-TX}\input{cliffordplustaxioms/C18-18.tikz}\\
    \eqdeuxeqref{t-XcommutCNOT}{t-X2}\input{cliffordplustaxioms/C18-19.tikz}
\end{gather*}
\end{derivation}

\begin{derivation}\label{lem:t-C19-dec-eq}
\begin{gather*}
    \input{cliffordplustaxioms/C19l.tikz}\\
    \eqdeuxeqref{t-XcommutCNOT}{t-X2}\input{cliffordplustaxioms/C19_01.tikz}\\
    \eqeqref{t-TX}\input{cliffordplustaxioms/C19_02.tikz}\\
    \eqeqref{t-HX}\input{cliffordplustaxioms/C19_03.tikz}\\
    \eqtroiseqref{t-Z}{t-Tpow8}{t-CtrlH}\input{cliffordplustaxioms/C19_04.tikz}\\
    \eqeqref{t-CX2}\input{cliffordplustaxioms/C19_05.tikz}\\
    \eqdeuxeqref{t-Cs}{t-Tpow8}\input{cliffordplustaxioms/C19_06.tikz}\\
    \eqeqref{t-H2}\input{cliffordplustaxioms/C19_07.tikz}\\
    \eqeqref{t-czrev}\input{cliffordplustaxioms/C19_08.tikz}\\
    \eqeqref{t-C}\input{cliffordplustaxioms/C19_09.tikz}\\
    \eqeqref{t-czrev}\input{cliffordplustaxioms/C19_10.tikz}\\
    \eqeqref{t-CSCS}\input{cliffordplustaxioms/C19_11.tikz}\\
    \eqeqref{t-CsCh}\input{cliffordplustaxioms/C19_12.tikz}\\
    \eqeqref{t-C}\input{cliffordplustaxioms/C19_13.tikz}\\
    \eqeqref{t-CSCS}\input{cliffordplustaxioms/C19_14.tikz}\\
    \eqdeuxeqref{t-czrev}{t-C}\input{cliffordplustaxioms/C19_15.tikz}\\
    \eqdeuxeqref{t-H2}{t-C}\input{cliffordplustaxioms/C19_16.tikz}\\
    \eqeqref{t-Tpow8}\input{cliffordplustaxioms/C19_17.tikz}\\
    \eqeqref{t-Cs}\input{cliffordplustaxioms/C19_18.tikz}\\
    \eqeqref{t-S}\input{cliffordplustaxioms/C19_19.tikz}\\
    \eqeqref{t-CHCH}\input{cliffordplustaxioms/C19_20.tikz}\\
    \eqdeuxeqref{t-HCh}{t-ChWhite}\input{cliffordplustaxioms/C19_22.tikz}\\
    \eqeqref{t-HSTH}\input{cliffordplustaxioms/C19_23.tikz}\\
    \eqeqref{t-CsCh}\input{cliffordplustaxioms/C19_24.tikz}\\
    \eqtroiseqref{t-Cs}{t-phasegadgetCS}{t-C}\input{cliffordplustaxioms/C19_25a.tikz}\\
    \eqeqref{t-CsCh}\input{cliffordplustaxioms/C19_25b.tikz}\\
    \eqtroiseqref{t-Cs}{t-phasegadgetCS}{t-C}\input{cliffordplustaxioms/C19_25.tikz}\\
    \eqeqref{t-ChWhite}\input{cliffordplustaxioms/C19_26.tikz}\\
    \eqeqref{t-CsCh}\input{cliffordplustaxioms/C19_27.tikz}\\
    \eqeqref{t-CSCSCSCS}\input{cliffordplustaxioms/C19_28.tikz}\\
    \eqdeuxeqref{t-ChWhite}{t-HCh}\input{cliffordplustaxioms/C19_29.tikz}\\
    \eqeqref{t-HSTH}\input{cliffordplustaxioms/C19_30.tikz}\\
    \eqeqref{t-CHCH}\input{cliffordplustaxioms/C19_31.tikz}\\
    \eqtroiseqref{t-Z}{t-Tpow8}{t-CtrlH}\input{cliffordplustaxioms/C19_32.tikz}\\
    \eqeqref{t-HX}\input{cliffordplustaxioms/C19_33.tikz}\\
    \eqeqref{t-TX}\input{cliffordplustaxioms/C19_34.tikz}\\
    \eqdeuxeqref{t-XcommutCNOT}{t-X2}\input{cliffordplustaxioms/C19r.tikz}
\end{gather*}
\end{derivation}

\begin{derivation}\label{lem:t-dswap-unfolding}
\begin{gather*}
    \input{cliffordczaxioms/DSWAPdef-00.tikz}
    \eqeqref{t-H2}\input{cliffordczaxioms/DSWAPdef-00b.tikz}
    \eqeqref{t-CX2}\input{cliffordczaxioms/DSWAPdef-01.tikz}
    \eqeqref{t-B}\input{cliffordczaxioms/DSWAPdef-01b.tikz}\\
    \eqeqref{t-H2}\input{cliffordczaxioms/DSWAPdef-03.tikz}\\
    \eqeqref{t-czrev}\input{cliffordczaxioms/DSWAPdef-03b.tikz}
\end{gather*}
\end{derivation}
\section{Auxiliary derivations for
  \texorpdfstring{$\QCcliffordpluscs$}{QCcliffordpluscs}}
\label{appendix:cs-decoding-lemmas}

\begin{derivation}\label{lem:cs-C4-dec-eq}
\begin{gather*}
    \input{cliffordpluscsderivations/C4-00.tikz} 
    \eqeqref{cs-E} \input{cliffordpluscsderivations/C4-01.tikz}\\
    \eqeqref{cs-S4}\input{cliffordpluscsderivations/C4-02.tikz}
    \eqeqref{cs-H2}\input{cliffordpluscsderivations/w6.tikz}
\end{gather*}
\end{derivation}

\begin{derivation}\label{lem:cs-C5-dec-eq}
\begin{gather*}
    \input{cliffordpluscsderivations/C5-00.tikz} 
    \eqeqref{cs-H2} \input{cliffordpluscsderivations/C5-01.tikz}\\
    \eqeqref{cs-CNOT}\input{cliffordpluscsderivations/C5-02.tikz}
    \eqdeuxeqref{cs-S4}{cs-H2}\input{cliffordpluscsderivations/C5-03.tikz}\\
    \eqeqref{cs-C}\input{cliffordpluscsderivations/C5-04.tikz}\\
    \eqeqref{cs-H2}\input{cliffordpluscsderivations/C5-05.tikz}\\
    \eqeqref{cs-CNOT}\input{cliffordpluscsderivations/C5-06.tikz}\\
    \eqeqref{cs-B}\input{cliffordpluscsderivations/C5-07.tikz}
    = \input{cliffordpluscsderivations/C5-08.tikz}\\
    \eqeqref{cs-B}\input{cliffordpluscsderivations/C5-09.tikz}
    = \input{cliffordpluscsderivations/C5-10.tikz}\\
    \eqeqref{cs-CNOT}\input{cliffordpluscsderivations/C5-11.tikz}\\
    \eqeqref{cs-H2}\input{cliffordpluscsderivations/C5-12.tikz}\\
    \eqeqref{cs-C}\input{cliffordpluscsderivations/C5-13.tikz}
    \eqeqref{cs-S4} \input{cliffordpluscsderivations/C5-15.tikz}
    \eqeqref{cs-H2} \input{identities/I2.tikz}
\end{gather*}
\end{derivation}

\begin{derivation}\label{lem:cs-C6-dec-eq}
\begin{gather*}
    \input{cliffordpluscsaxioms/ScommTopLeft.tikz}
    \eqeqref{cs-C}\input{cliffordpluscsderivations/C6-01.tikz}
    \eqeqref{lem:cs-C5-dec-eq} \input{cliffordpluscsaxioms/ScommTopRight.tikz}
\end{gather*}
\end{derivation}

\begin{derivation}\label{lem:cs-C7-dec-eq}
\begin{gather*}
    \input{cliffordpluscsaxioms/ScommBotLeft.tikz}
    \eqeqref{cs-C}\input{cliffordpluscsderivations/C7-01.tikz}\\
    \eqeqref{cs-CSrev}\input{cliffordpluscsderivations/C7-02.tikz}
    \eqeqref{lem:cs-C5-dec-eq} \input{cliffordpluscsderivations/C7-03.tikz}
    = \input{cliffordpluscsaxioms/ScommBotRight.tikz}
\end{gather*}
\end{derivation}

\begin{derivation}\label{lem:cs-C8-dec-eq}
\begin{gather*}
    \input{cliffordpluscsderivations/C8-00.tikz}
    \eqeqref{cs-X}\input{cliffordpluscsderivations/C8-01.tikz}\\
    \eqeqref{cs-XCS}\input{cliffordpluscsderivations/C8-02.tikz}
    \eqeqref{cs-X}\input{cliffordpluscsderivations/C8-03.tikz}
\end{gather*}
\end{derivation}

\begin{derivation}\label{lem:cs-C9-dec-eq}
\begin{gather*}
    \input{cliffordpluscsderivations/C9-00.tikz}
    \eqeqref{cs-CSrev}\input{cliffordpluscsderivations/C9-01.tikz}\\
    \eqeqref{lem:cs-C8-dec-eq}\input{cliffordpluscsderivations/C9-02.tikz}
    \eqeqref{cs-CSrev}\input{cliffordpluscsderivations/C9-03.tikz}
\end{gather*}
\end{derivation}

\begin{derivation}\label{lem:cs-C11-dec-eq}
\begin{gather*}
    \input{cliffordpluscsderivations/C11-00.tikz}
    \eqeqref{cs-CSrev}\input{cliffordpluscsderivations/C11-01.tikz}\\
    \eqeqref{cs-SHCHC}\input{cliffordpluscsderivations/C11-02.tikz}
    \eqeqref{cs-CSrev}\input{cliffordpluscsderivations/C11-03.tikz}
\end{gather*}
\end{derivation}

\begin{derivation}\label{lem:cs-C12-dec-eq}
\begin{gather*}
    \input{cliffordpluscsaxioms/ICSCSI.tikz}
    \eqeqref{cs-H2}\input{cliffordpluscsderivations/C12-00.tikz}
    \eqeqref{cs-U}\input{cliffordpluscsderivations/C12-01.tikz}
    \eqeqref{cs-H2}\input{cliffordpluscsaxioms/CSIICS.tikz}
\end{gather*}
\end{derivation}

\begin{derivation}\label{lem:cs-C13-dec-eq}
\begin{gather*}
    \input{cliffordpluscsaxioms/C13l.tikz}
    \eqeqref{cs-B}\input{cliffordpluscsderivations/C13-02.tikz}
    \eqeqref{lem:cs-C12-dec-eq}\input{cliffordpluscsderivations/C13-03.tikz}\\
    \eqeqref{cs-CNOT}\input{cliffordpluscsderivations/C13-04.tikz}
    \eqeqref{cs-H2}\input{cliffordpluscsderivations/C13-05.tikz}\\
    \eqeqref{lem:cs-C5-dec-eq}\input{cliffordpluscsderivations/C13-06.tikz}
    \eqeqref{cs-H2}\input{cliffordpluscsderivations/C13-07.tikz}
    = \input{cliffordpluscsderivations/C13-08.tikz}
    \eqeqref{cs-H2}\input{cliffordpluscsderivations/C13-09.tikz}\\
    \eqeqref{lem:cs-C5-dec-eq}\input{cliffordpluscsderivations/C13-10.tikz}
    \eqeqref{lem:cs-C12-dec-eq}\input{cliffordpluscsderivations/C13-11.tikz}\\
    \eqeqref{cs-H2}\input{cliffordpluscsderivations/C13-12.tikz}
    \eqeqref{cs-CNOT}\input{cliffordpluscsderivations/C13-13.tikz}\\
    \eqeqref{cs-B}\input{cliffordpluscsaxioms/C13r.tikz}
\end{gather*}
\end{derivation}

\begin{derivation}\label{lem:cs-C14-dec-eq}
\begin{gather*}
    \input{cliffordpluscsaxioms/C14l.tikz}
    \eqdeuxeqref{cs-CNOT}{cs-I}\input{cliffordpluscsderivations/C14-01.tikz}\\
    \eqeqref{lem:cs-C5-dec-eq}\input{cliffordpluscsderivations/C14-02.tikz}\\
    \eqtroiseqref{lem:cs-C5-dec-eq}{cs-H2}{cs-CNOT}\input{cliffordpluscsderivations/C14-03.tikz}\\
    \eqdeuxeqref{cs-CNOT}{cs-U}\input{cliffordpluscsderivations/C14-04.tikz}\\
    \eqdeuxeqref{lem:cs-C12-dec-eq}{lem:cs-C5-dec-eq}\input{cliffordpluscsderivations/C14-05.tikz}\\
    \eqtroiseqref{lem:cs-C12-dec-eq}{cs-H2}{lem:cs-C5-dec-eq}\input{cliffordpluscsderivations/C14-06.tikz}
    \eqdeuxeqref{lem:cs-C12-dec-eq}{cs-I}\input{cliffordpluscsaxioms/C14r.tikz}\\
\end{gather*}
\end{derivation}

\begin{derivation}\label{lem:cs-C15-dec-eq}
\begin{gather*}
    \input{cliffordpluscsaxioms/C15l.tikz}
    \eqeqref{cs-H2}\input{cliffordpluscsderivations/C15-01.tikz}
    \eqeqref{cs-CNOT}\input{cliffordpluscsderivations/C15-02.tikz}\\
    \eqeqref{cs-B}\input{cliffordpluscsderivations/C15-03.tikz}
    \eqtroiseqref{cs-CNOT}{cs-H2}{lem:cs-C5-dec-eq}\input{cliffordpluscsderivations/C15-04.tikz}
    =\input{cliffordpluscsderivations/C15-05.tikz}\\
    \eqtroiseqref{cs-CNOT}{cs-H2}{lem:cs-C5-dec-eq}\input{cliffordpluscsderivations/C15-06.tikz}
    \eqeqref{cs-I}\input{cliffordpluscsderivations/C15-07.tikz}\\
    \eqdeuxeqref{cs-CNOT}{cs-H2}\input{cliffordpluscsaxioms/C15r.tikz}
\end{gather*}
\end{derivation}

\begin{derivation}\label{lem:cs-C17-dec-eq}
\begin{gather*}
    \input{cliffordpluscsderivations/C17-00.tikz}
    \eqeqref{lem:cs-C5-dec-eq} \input{cliffordpluscsderivations/C17-01.tikz}\\
    \eqeqref{cs-U} \input{cliffordpluscsderivations/C17-02.tikz}\\
    \eqdeuxeqref{lem:cs-C5-dec-eq}{lem:cs-C12-dec-eq} \input{cliffordpluscsderivations/C17-final.tikz}
\end{gather*}
\end{derivation}

\begin{derivation}\label{lem:cs-dswap-unfolding}
\begin{gather*}
    \input{gates/SWAP.tikz}
    \eqdeuxeqref{cs-H2}{lem:cs-C5-dec-eq}\input{cliffordpluscsderivations/DSWAP-01.tikz}\\
    \eqeqref{cs-CNOT}\input{cliffordpluscsderivations/DSWAP-02.tikz}
    \eqeqref{cs-B}\input{cliffordpluscsderivations/DSWAP-03.tikz}\\
    \eqdeuxeqref{cs-CNOT}{cs-w8}\input{cliffordpluscsderivations/DSWAP-dfinal.tikz}
\end{gather*}
\end{derivation}

\section[Proofs for Lemma~\ref{lem:decenc}]{Proofs for \cref{lem:decenc}}\label{app:proofs-lem1}

\subsection{Clifford}\label{app:proofs-lem1-clifford}
\(
  \begin{aligned}
    D\bigl(E(\input{gates/CNOT.tikz})\bigr)
    &= D\bigl(\input{completeness/H2CZH2.tikz}\bigr)
    = \input{completeness/HH2CNOTHH2.tikz}
    \eqeqref{H2} \input{gates/CNOT.tikz}.
  \end{aligned}
\)

\subsection{Real Clifford}\label{app:proofs-lem1-realclifford}
\(
  \begin{aligned}
    D\bigl(E(\input{gates/CNOT.tikz})\bigr)
    &= D\bigl(\input{completeness/H2CZH2.tikz}\bigr)
    = \input{completeness/HH2CNOTHH2.tikz}
    \eqeqref{realH2} \input{gates/CNOT.tikz}.
  \end{aligned}
\)

\subsection{Qutrit}\label{app:proofs-lem1-qutrit}
\(
  \begin{aligned}
    D\bigl(E(\input{gates/CNOT.tikz})\bigr)
    &= D\bigl(\input{completeness/H2CZHHH2.tikz}\bigr)
    = \input{completeness/HHHH2CNOTHHHH2.tikz}
    \eqdeuxeqref{qt-hpow4}{qt-omegapow12} \input{gates/CNOT.tikz},\\
    D\bigl(E(\input{gates/H.tikz})\bigr)
    &= D\bigl(\input{gates/Hbisp.tikz}\bigr)
    = \input{gates/Hbist.tikz}
    \eqeqref{qt-omegapow12} \input{gates/H.tikz},\\
    D\bigl(E(\input{gates/S.tikz})\bigr)
    &= D\bigl(\input{gates/Sbis.tikz}\bigr)
    = \input{gates/St.tikz}
    \eqdeuxeqref{qt-omegapow12}{qt-spow3} \input{gates/S.tikz}.
  \end{aligned}
\)

\subsection{Clifford+T}\label{app:proofs-lem1-cliffordt}
\(
  \begin{aligned}
    D\bigl(E(\input{gates/CNOT.tikz})\bigr)
    &= D\bigl(\input{completeness/H2CZH2.tikz}\bigr)
    = \input{completeness/HH2CNOTHH2.tikz}
    \eqeqref{t-H2} \input{gates/CNOT.tikz}.
  \end{aligned}
\)

\subsection{Clifford+CS}\label{app:proofs-lem1-cliffordcs}
\(
  \begin{aligned}
    D\bigl(E(\input{gates/H.tikz})\bigr)
    &= D\bigl(\input{gates/Kw.tikz}\bigr)
    = \input{gates/Hww.tikz}
    \eqeqref{cs-w8} \input{gates/H.tikz}.
  \end{aligned}
\)
\section{Proof of \cref{lem:decrelations}}\label{appendix:decodage}
\subsection{Clifford decoding}
\begin{proof}[Proof of the decoding of \cref{DSWAPdef}]
    \begin{gather*}
        D\left(\input{gates/SWAP.tikz}\right)
        = \input{gates/SWAP.tikz}\\
        \eqeqref{lem:dswap-unfolding} \input{cliffordczaxioms/DSWAPdef-03b.tikz}
        = D\left(\input{cliffordczaxioms/DSWAPdef-02.tikz}\right)
    \end{gather*}
\end{proof}

\begin{proof}[Proof of the decoding of \cref{C1}, \cref{C2}, \cref{C3}]
    Since the decoding only affects the $\gCZ$ gate, these are trivial.
\end{proof}

\begin{proof}[Proof of the decoding of \cref{C4}]
    \begin{gather*}
        D\left(\input{cliffordczaxioms/C4-00.tikz}\right)
        = \input{cliffordczaxioms/C4-00.tikz}\\
        \eqeqref{lem:C4-dec-eq} \input{cliffordczaxioms/C4-03.tikz}
        = D\left(\input{cliffordczaxioms/C4-03.tikz}\right)
    \end{gather*}
\end{proof}

\begin{proof}[Proof of the decoding of \cref{C5}]
    \begin{gather*}
        D\left(\input{cliffordczaxioms/C5-00.tikz}\right)
        = \input{cliffordczaxioms/C5-01.tikz}
        \eqeqref{lem:C5-dec-eq} \input{cliffordczaxioms/C5-04.tikz}
        = D\left(\input{identities/I2.tikz}\right)
    \end{gather*}
\end{proof}

\begin{proof}[Proof of the decoding of \cref{C6}]
    \begin{gather*}
        D\left(\input{cliffordczaxioms/C6-00.tikz}\right)
        = \input{cliffordczaxioms/C6-01.tikz}
        \eqeqref{lem:C6-dec-eq} \input{cliffordczaxioms/C6-02.tikz}
        = D\left(\input{cliffordczaxioms/C6-03.tikz}\right)
    \end{gather*}
\end{proof}

\begin{proof}[Proof of the decoding of \cref{C7}]
    \begin{gather*}
        D\left(\input{cliffordczaxioms/C7-00.tikz}\right)
        = \input{cliffordczaxioms/C7-01.tikz}
        \eqeqref{lem:C7-dec-eq} \input{cliffordczaxioms/C7-04.tikz}
        = D\left(\input{cliffordczaxioms/C7-05.tikz}\right)
    \end{gather*}
\end{proof}

\begin{proof}[Proof of the decoding of \cref{C8}]
    \begin{gather*}
        D\left(\input{cliffordczaxioms/C8-00.tikz}\right)
        = \input{cliffordczaxioms/C8-01.tikz}\\
        \eqeqref{lem:C8-dec-eq} \input{cliffordczaxioms/C8-07.tikz}
        = D\left(\input{cliffordczaxioms/C8-08.tikz}\right)
    \end{gather*}
\end{proof}

\begin{proof}[Proof of the decoding of \cref{C9}]
    \begin{gather*}
        D\left(\input{cliffordczaxioms/C9-00.tikz}\right)
        = \input{cliffordczaxioms/C9-01.tikz}\\
        \eqeqref{lem:C9-dec-eq} \input{cliffordczaxioms/C9-08.tikz}
        = D\left(\input{cliffordczaxioms/C9-09.tikz}\right)
    \end{gather*}
\end{proof}

\begin{proof}[Proof of the decoding of \cref{C10}]
    \begin{gather*}
        D\left(\input{cliffordczaxioms/C10-00.tikz}\right)
        = \input{cliffordczaxioms/C10-01.tikz}\\
        \eqeqref{lem:C10-dec-eq} \input{cliffordczaxioms/C10-13.tikz}
        = D\left(\input{cliffordczaxioms/C10-14.tikz}\right)
    \end{gather*}
\end{proof}

\begin{proof}[Proof of the decoding of \cref{C11}]
    \begin{gather*}
        D\left(\input{cliffordczaxioms/C11-00.tikz}\right)
        = \input{cliffordczaxioms/C11-00b.tikz}\\
        \eqeqref{lem:C11-dec-eq} \input{cliffordczaxioms/C11-01b.tikz}
        = D\left(\input{cliffordczaxioms/C11-01.tikz}\right)
    \end{gather*}
\end{proof}

\begin{proof}[Proof of the decoding of \cref{C12}]
    \begin{gather*}
        D\left(\input{cliffordczaxioms/C12-00.tikz}\right)
        = \input{cliffordczaxioms/C12-01.tikz}
        \eqeqref{lem:C12-dec-eq} \input{cliffordczaxioms/C12-05.tikz}
        = D\left(\input{cliffordczaxioms/C12-06.tikz}\right)
    \end{gather*}
\end{proof}

\begin{proof}[Proof of the decoding of \cref{C13}]
    \begin{gather*}
        D\left(\input{cliffordczaxioms/C13-00.tikz}\right)\\
        = \input{cliffordczaxioms/C13-01.tikz}\\
        \eqeqref{lem:C13-dec-eq} \input{cliffordczaxioms/C13-24.tikz}\\
        = D\left(\input{cliffordczaxioms/C13-25.tikz}\right)
    \end{gather*}
\end{proof}

\begin{proof}[Proof of the decoding of \cref{C14}]
    \begin{gather*}
        D\left(\input{cliffordczaxioms/C14-00.tikz}\right)\\
        = \scalebox{0.8}{\input{cliffordczaxioms/C14-01.tikz}}\\
        \eqeqref{lem:C14-dec-eq} \input{cliffordczaxioms/C14-09.tikz}
        = D\left(\input{cliffordczaxioms/C14-09.tikz}\right)
    \end{gather*}
\end{proof}

\begin{proof}[Proof of the decoding of \cref{C15}]
    \begin{gather*}
        D\left(\input{cliffordczaxioms/C15-00.tikz}\right)\\
        = \scalebox{0.8}{\input{cliffordczaxioms/C15-00b.tikz}}\\
        \eqeqref{lem:C15-dec-eq} \input{cliffordczaxioms/C15-01.tikz}
        = D\left(\input{cliffordczaxioms/C15-01.tikz}\right)
    \end{gather*}
\end{proof}

\subsection{Real-Clifford decoding}

\begin{proof}[Proof of the decoding of \cref{R1}, \cref{R2}, \cref{R3}]
    Since the decoding only affects the $\gCZ$ gate, these are trivial.
\end{proof}

\begin{proof}[Proof of the decoding of \cref{R4}]
    \begin{gather*}
        D\left(\input{realcliffordczaxioms/R4-00.tikz}\right)
        = \input{realcliffordczaxioms/R4-00.tikz}\\
        \eqeqref{lem:R4-dec-eq} \input{realcliffordczaxioms/R4-02.tikz}
        = D\left(\input{realcliffordczaxioms/R4-02.tikz}\right)
    \end{gather*}
\end{proof}

\begin{proof}[Proof of the decoding of \cref{R5}]
    \begin{gather*}
        D\left(\input{cliffordczaxioms/C5-00.tikz}\right)
        = \input{realcliffordczaxioms/R5-01.tikz}
        \eqeqref{lem:R5-dec-eq} \input{realcliffordczaxioms/R5-04.tikz}
        = D\left(\input{identities/I2.tikz}\right)
    \end{gather*}
\end{proof}

\begin{proof}[Proof of the decoding of \cref{R6}]
    \begin{gather*}
        D\left(\input{realcliffordczaxioms/R6-00.tikz}\right)
        = \input{realcliffordczaxioms/R6-01.tikz}
        \eqeqref{lem:R6-dec-eq} \input{realcliffordczaxioms/R6-06.tikz}
        = D\left(\input{realcliffordczaxioms/R6-07.tikz}\right)
    \end{gather*}
\end{proof}

\begin{proof}[Proof of the decoding of \cref{R7}]
    \begin{gather*}
        D\left(\input{realcliffordczaxioms/R7-00.tikz}\right)
        = \input{realcliffordczaxioms/R7-00b.tikz}
        \eqeqref{lem:R7-dec-eq} \input{realcliffordczaxioms/R7-01b.tikz}
        = D\left(\input{realcliffordczaxioms/R7-01.tikz}\right)
    \end{gather*}
\end{proof}

\begin{proof}[Proof of the decoding of \cref{R8}]
    \begin{gather*}
        D\left(\input{realcliffordczaxioms/R8-00.tikz}\right)
        = \input{realcliffordczaxioms/R8-01.tikz}\\
        \eqeqref{lem:R8-dec-eq} \input{realcliffordczaxioms/R8-05.tikz}
        = D\left(\input{realcliffordczaxioms/R8-06.tikz}\right)
    \end{gather*}
\end{proof}

\begin{proof}[Proof of the decoding of \cref{R9}]
    \begin{gather*}
        D\left(\input{realcliffordczaxioms/R9-00.tikz}\right)
        = \input{realcliffordczaxioms/R9-00b.tikz}\\
        \eqeqref{lem:R9-dec-eq} \input{realcliffordczaxioms/R9-01b.tikz}
        = D\left(\input{realcliffordczaxioms/R9-01.tikz}\right)
    \end{gather*}
\end{proof}

\begin{proof}[Proof of the decoding of \cref{R10}]
    \begin{gather*}
        D\left(\input{realcliffordczaxioms/R10-00.tikz}\right)
        = \input{realcliffordczaxioms/R10-01.tikz}\\
        \eqeqref{lem:R10-dec-eq} \input{realcliffordczaxioms/R10-09.tikz}
        = D\left(\input{realcliffordczaxioms/R10-10.tikz}\right)
    \end{gather*}
\end{proof}

\begin{proof}[Proof of the decoding of \cref{R11}]
    \begin{gather*}
        D\left(\input{realcliffordczaxioms/R11-00.tikz}\right)
        = \input{realcliffordczaxioms/R11-00b.tikz}\\
        \eqeqref{lem:R11-dec-eq} \input{realcliffordczaxioms/R11-01b.tikz}
        = D\left(\input{realcliffordczaxioms/R11-01.tikz}\right)
    \end{gather*}
\end{proof}

\begin{proof}[Proof of the decoding of \cref{R12}]
    \begin{gather*}
        D\left(\input{realcliffordczaxioms/R12-00.tikz}\right)
        = \input{realcliffordczaxioms/R12-01.tikz}\\
        \eqeqref{lem:R12-dec-eq} \input{realcliffordczaxioms/R12-11.tikz}
        = D\left(\input{realcliffordczaxioms/R12-12.tikz}\right)
    \end{gather*}
\end{proof}

\begin{proof}[Proof of Equations \eqref{real-DSWAPdef} \eqref{R13}, \eqref{R14}, \eqref{R15} and \eqref{R16}]
    Equations \eqref{real-DSWAPdef}, \eqref{R13}, \eqref{R14}, \eqref{R15} and \eqref{R16} of $\QCrealcliffordCZ$ respectively correspond to Equations \eqref{DSWAPdef}, \eqref{C12}, \eqref{C13}, \eqref{C14} and \eqref{C15} of $\QCcliffordCZ$. One could straightforwardly mimic the proofs of Equations $\textup{C}_i$ to prove Equations $\textup{R}_j$.
\end{proof}

\subsection{Qutrit-Clifford decoding}

\begin{proof}[Proof of the decoding of \cref{qt-C1}]
    \begin{gather*}
        D\left(\input{qutritcliffordaxioms/omegapow6.tikz}\right)
        = \input{qutritcliffordaxioms/omegapow60.tikz}
        \eqeqref{qt-omegapow12} \input{gates/empty.tikz}
        = D\left(\input{gates/empty.tikz}\right)
    \end{gather*}
\end{proof}

\begin{proof}[Proof of the decoding of \cref{qt-C2}]
    \begin{gather*}
        D\left(\input{qutritcliffordaxioms/hadpow4.tikz}\right)
        = \input{qutritcliffordaxioms/hadpow4phases.tikz}
        \eqdeuxeqref{qt-hpow4}{qt-omegapow12} \input{gates/Id.tikz}
        = D\left(\input{gates/Id.tikz}\right)
    \end{gather*}
\end{proof}

\begin{proof}[Proof of the decoding of \cref{qt-C3}]
    \begin{gather*}
        D\left(\input{qutritcliffordaxioms/spow3.tikz}\right)
        = \input{qutritcliffordaxioms/spow3phases.tikz}
        \eqdeuxeqref{qt-spow3}{qt-omegapow12} \input{gates/Id.tikz}
        = D\left(\input{gates/Id.tikz}\right)
    \end{gather*}
\end{proof}

\begin{proof}[Proof of the decoding of \cref{qt-C4}]
    \begin{gather*}
        D\left(\input{qutritcliffordaxioms/s2hs2hs2h.tikz}\right)
        = \input{qutritcliffordaxioms/s2hs2hs2hphases.tikz}
        \eqeqref{lem:qt-C4-dec-eq} \input{qutritcliffordaxioms/D-Iphase-omega.tikz}
        = D\left(\input{qutritcliffordaxioms/Iphase-omega.tikz}\right)
    \end{gather*}
\end{proof}

\begin{proof}[Proof of the decoding of \cref{qt-C5}]
    \begin{gather*}
        D\left(\input{qutritcliffordaxioms/ssprime.tikz}\right)
        = \input{qutritcliffordaxioms/ssprimeb.tikz}
        \eqeqref{qt-ssprime} \input{qutritcliffordaxioms/sprimesb.tikz}
        = D\left(\input{qutritcliffordaxioms/sprimes.tikz}\right)
    \end{gather*}
\end{proof}

\begin{proof}[Proof of the decoding of \cref{qt-C6}]
    \begin{gather*}
        D\left(\input{qutritderivations/CZ-CZ-CZ.tikz}\right)
        = \input{qutritderivations/CX-CX-CX.tikz}
        \eqdeuxeqref{qt-hpow4}{qt-CCC} \input{identities/I2.tikz}
        = D\left(\input{identities/I2.tikz}\right)
    \end{gather*}
\end{proof}

\begin{proof}[Proof of the decoding of \cref{qt-C7}]
    \begin{gather*}
        D\left(\input{qutritderivations/S-CZ.tikz}\right)
        = \input{qutritderivations/S2-CZ.tikz}
        \eqeqref{qt-S-CNOT} \input{qutritderivations/CZ-S2.tikz}
        = D\left(\input{qutritderivations/CZ-S.tikz}\right)
    \end{gather*}
\end{proof}

\begin{proof}[Proof of the decoding of \cref{qt-C8}]
    \begin{gather*}
        D\left(\input{qutritderivations/K-CZ.tikz}\right)
        = \input{qutritderivations/Kp-CZ.tikz}\\
        \eqdeuxeqref{qt-kcnot}{qt-hpow4} \input{qutritderivations/CZ-CZ-Kp.tikz}
        = D\left(\input{qutritderivations/CZ-CZ-K.tikz}\right)
    \end{gather*}
\end{proof}

\begin{proof}[Proof of the decoding of \cref{qt-C10}]
    \begin{gather*}
        D\left(\input{qutritderivations/S-NOTC-origin.tikz}\right)
        = \input{qutritderivations/S-NOTCb.tikz}\\
        \eqtroiseqref{qt-hpow4}{qt-S-NOTC}{qt-omegapow12} \input{qutritderivations/S-NOTC-09b.tikz}
        = D\left(\input{qutritderivations/NOTC-CZ-SSp.tikz}\right)
    \end{gather*}
\end{proof}

\begin{proof}[Proof of the decoding of \cref{qt-C11}]
    \begin{gather*}
        D\left(\input{qutritderivations/CZ-CNOT-b.tikz}\right)
        = \input{qutritderivations/CZ-CNOT.tikz}\\
        \eqtroiseqref{qt-hpow4}{qt-CZ-CNOT}{qt-omegapow12} \input{qutritderivations/CZ-CNOT-10b.tikz}
        = D\left(\input{qutritderivations/NOTC-CZ-SSp-b.tikz}\right)
    \end{gather*}
\end{proof}

\begin{proof}[Proof of the decoding of \cref{qt-C13}]
    \begin{gather*}
        D\left(\input{qutritderivations/ICZ-CZI.tikz}\right)
        = \input{qutritderivations/CZ12CZ23.tikz}
        \eqeqref{qt-CzComm} \input{qutritderivations/CZ23CZ12.tikz}
        = D\left(\input{qutritderivations/CZI-ICZ.tikz}\right)
    \end{gather*}
\end{proof}

\begin{proof}[Proof of the decoding of \cref{qt-C16}]
    \begin{gather*}
        D\left(\input{qutritderivations/CNI-ICZ.tikz}\right)
        = \input{qutritcliffordaxioms/CNOT12CNOT23b.tikz}\\
        \eqdeuxeqref{qt-I}{qt-hpow4} \input{qutritcliffordaxioms/CNOT23CNOT12CNOT13b.tikz}
        = D\left(\input{qutritderivations/ICZ-CNI-CIZ.tikz}\right)
    \end{gather*}
\end{proof}

\begin{proof}[Proof of the decoding of \cref{qt-Swap}]
    \begin{gather*}
        D\left(\input{qutritderivations/SWAP-CZ.tikz}\right)
        = \input{qutritderivations/SWAP-CZ-00b.tikz}
        \eqdeuxeqref{qt-SwCZ}{qt-hpow4} \input{gates/SWAP.tikz}
        = D\left(\input{gates/SWAP.tikz}\right)
    \end{gather*}
\end{proof}
\begin{proof}[Proof of Equations \eqref{qt-C9}, \eqref{qt-C14}, \eqref{qt-C15}, \eqref{qt-SWAPSWAP}, \eqref{qt-SSWAP}, \eqref{qt-CZSWAP}]
    \cref{qt-C9}, \cref{qt-C14}, \cref{qt-C15}, \cref{qt-SWAPSWAP}, \cref{qt-SSWAP}, \cref{qt-CZSWAP} comes from the naturality of the swap and the new rule \cref{qt-Swap}
\end{proof}

\subsection{Clifford+T decoding}
\begin{proof}[Proof of the decoding of \cref{t-C3}, \cref{t-C4}, \cref{t-C15}]
    Since the decoding only affects the $\gS$ gate and $\gCZ$ gate, these are trivial.
\end{proof}

\begin{proof}[Proof of the decoding of \cref{t-C5}]
    \begin{gather*}
        D\left(\input{cliffordaxioms/SSSS.tikz}\right)
        = \input{cliffordplustaxioms/Tpow8.tikz}
        \eqeqref{t-Tpow8} \input{cliffordaxioms/I.tikz}
        = D\left(\input{cliffordaxioms/I.tikz}\right)
    \end{gather*}
\end{proof}

\begin{proof}[Proof of the decoding of \cref{t-C6}]
    \begin{gather*}
        D\left(\input{cliffordczaxioms/C4-00.tikz}\right)
        = \input{cliffordplustaxioms/TTHTTHTTH.tikz}\\
        \eqeqref{lem:t-C6-dec-eq} \input{cliffordczaxioms/C4-03.tikz}
        = D\left(\input{cliffordczaxioms/C4-03.tikz}\right)
    \end{gather*}
\end{proof}

\begin{proof}[Proof of the decoding of \cref{t-C14}]
    \begin{gather*}
        D\left(\input{shortcut/S.tikz}\right)
        = \input{shortcut/TT.tikz}
        = D\left(\input{shortcut/TT.tikz}\right)
    \end{gather*}
\end{proof}

\begin{proof}[Proof of the decoding of \cref{t-C7}]
    \begin{gather*}
        D\left(\input{cliffordczaxioms/C5-00.tikz}\right)
        = \input{cliffordplustaxioms/C7-00.tikz}
        \eqeqref{lem:t-C7-dec-eq} \input{cliffordczaxioms/C5-04.tikz}
        = D\left(\input{cliffordczaxioms/C5-04.tikz}\right)
    \end{gather*}
\end{proof}

\begin{proof}[Proof of the decoding of \cref{t-C8}]
    \begin{gather*}
        D\left(\input{cliffordczaxioms/C6-00.tikz}\right)
        = \input{cliffordczaxioms/C6-01.tikz}
        \eqeqref{t-C} \input{cliffordczaxioms/C6-02.tikz}
        = D\left(\input{cliffordczaxioms/C6-03.tikz}\right)
    \end{gather*}
\end{proof}

\begin{proof}[Proof of the decoding of \cref{t-C9}]
    \begin{gather*}
        D\left(\input{cliffordczaxioms/C7-00.tikz}\right)
        = \input{cliffordczaxioms/C7-01.tikz}
        \eqdeuxeqref{t-C}{t-czrev} \input{cliffordczaxioms/C7-02.tikz}
        = D\left(\input{cliffordczaxioms/C7-05.tikz}\right)
    \end{gather*}   
\end{proof}

\begin{proof}[Proof of the decoding of \cref{t-C10}]
    \begin{gather*}
        D\left(\input{cliffordczaxioms/C8-00.tikz}\right)
        = \input{cliffordczaxioms/C8-01.tikz}\\
        \eqeqref{lem:t-C10-dec-eq} \input{cliffordczaxioms/C8-07.tikz}
        = D\left(\input{cliffordczaxioms/C8-08.tikz}\right)
    \end{gather*}
\end{proof}

\begin{proof}[Proof of the decoding of \cref{t-C11}]
    \begin{gather*}
        D\left(\input{cliffordczaxioms/C9-00.tikz}\right)
        = \input{cliffordczaxioms/C9-01.tikz}\\
        \eqeqref{lem:t-C11-dec-eq} \input{cliffordczaxioms/C9-08.tikz}
        = D\left(\input{cliffordczaxioms/C9-09.tikz}\right)
    \end{gather*}
\end{proof}

\begin{proof}[Proof of the decoding of \cref{t-C12}]
    \begin{gather*}
        D\left(\input{cliffordczaxioms/C10-00.tikz}\right)
        = \input{cliffordczaxioms/C10-01.tikz}\\
        \eqeqref{lem:t-C12-dec-eq} \input{cliffordczaxioms/C10-13.tikz}
        = D\left(\input{cliffordczaxioms/C10-14.tikz}\right)
    \end{gather*}
\end{proof}

\begin{proof}[Proof of the decoding of \cref{t-C13}]
    \begin{gather*}
        D\left(\input{cliffordczaxioms/C11-00.tikz}\right)
        = \input{cliffordczaxioms/C11-00b.tikz}\\
        \eqeqref{lem:t-C13-dec-eq} \input{cliffordczaxioms/C11-01b.tikz}
        = D\left(\input{cliffordczaxioms/C11-01.tikz}\right)
    \end{gather*}
\end{proof}

\begin{proof}[Proof of the decoding of \cref{t-C16}]
    \begin{gather*}
        D\left(\input{cliffordplustaxioms/C16l.tikz}\right)
        =  \input{cliffordplustaxioms/C16-01.tikz}
        \eqeqref{t-C} \input{cliffordplustaxioms/C16-02.tikz}
        = D\left(\input{cliffordplustaxioms/C16r.tikz}\right)
    \end{gather*}
\end{proof}

\begin{proof}[Proof of the decoding of \cref{t-C17}]
    \begin{gather*}
        D\left(\input{cliffordplustaxioms/C17l.tikz}\right)
        = \input{cliffordplustaxioms/C17-01.tikz}\\
        \eqeqref{lem:t-C17-dec-eq} \input{cliffordplustaxioms/C17-09.tikz}
        = D\left(\input{cliffordplustaxioms/C17r.tikz}\right)
    \end{gather*}
\end{proof}

\begin{proof}[Proof of the decoding of \cref{t-C18}]
    \begin{gather*}
        D\left(\input{cliffordplustaxioms/C18l.tikz}\right)
        = \input{cliffordplustaxioms/C18l.tikz}\\
        \eqeqref{lem:t-C18-dec-eq} \input{cliffordplustaxioms/C18-19.tikz}
        = D\left(\input{cliffordplustaxioms/C18r.tikz}\right)
    \end{gather*}
\end{proof}

\begin{proof}[Proof of the decoding of \cref{old-t-C19}]
    \begin{gather*}
        D\left(\input{cliffordplustaxioms/C19l.tikz}\right)\\
        = \input{cliffordplustaxioms/C19l.tikz}\\
        \eqeqref{lem:t-C19-dec-eq} \input{cliffordplustaxioms/C19r.tikz}\\
        = D\left(\input{cliffordplustaxioms/C19r.tikz}\right)
    \end{gather*}
\end{proof}

\begin{proof}[Proof of the decoding of \cref{old-t-C20}]
    \begin{gather*}
        D\left(\input{cliffordplustaxioms/C20l.tikz}\right)
        = \input{cliffordplustaxioms/C20l.tikz}\\
        \eqeqref{t-C20} \input{cliffordplustaxioms/C20r.tikz}
        = D\left(\input{cliffordplustaxioms/C20r.tikz}\right)
    \end{gather*}
\end{proof}

\begin{proof}[Proof of the decoding of \cref{t-DSWAPdef}]
    \begin{gather*}
        D\left(\input{gates/SWAP.tikz}\right)
        = \input{gates/SWAP.tikz}\\
        \eqeqref{lem:t-dswap-unfolding} \input{cliffordczaxioms/DSWAPdef-03b.tikz}
        = D\left(\input{cliffordczaxioms/DSWAPdef-02.tikz}\right)
    \end{gather*}
\end{proof}

\subsection{Clifford+CS decoding}
\begin{proof}[Proof of the decoding of \cref{DSWAPdef}]
    \begin{gather*}
        D\left(\input{gates/SWAP.tikz}\right)
        = \input{gates/SWAP.tikz}\\
        \eqeqref{lem:cs-dswap-unfolding} \input{cliffordpluscsderivations/DSWAP-dfinal.tikz}
        = D\left(\input{cliffordpluscsderivations/DSWAP-final.tikz}\right)
    \end{gather*}
\end{proof}

\begin{proof}[Proof of the decoding of \cref{cs-C3}]
    Since the decoding does not affect the $\gS$ and $\gCS$ gates, this is trivial.
\end{proof}

\begin{proof}[Proof of the decoding of \cref{cs-C1}]
    \begin{gather*}
        D\left(\input{qutritcliffordaxioms/ipow4.tikz}\right)
        = \input{cliffordaxioms/w.tikz}
        \eqeqref{cs-w8} \input{gates/empty.tikz}
        = D\left(\input{gates/empty.tikz}\right)
    \end{gather*}
\end{proof}

\begin{proof}[Proof of the decoding of \cref{cs-C2}]
    \begin{gather*}
        D\left(\input{cliffordpluscsaxioms/K2.tikz}\right)
        = \input{cliffordpluscsderivations/H2w6.tikz}
        \eqdeuxeqref{cs-H2}{cs-w8} \input{cliffordpluscsderivations/w6.tikz}
        = D\left(\input{cliffordpluscsaxioms/ipow3.tikz}\right)
    \end{gather*}
\end{proof}

\begin{proof}[Proof of the decoding of \cref{cs-C4}]
    \begin{gather*}
        D\left(\input{cliffordpluscsaxioms/SKSKSK.tikz}\right)
        = \input{cliffordpluscsderivations/C4-00b.tikz}\\
        \eqdeuxeqref{lem:cs-C4-dec-eq}{cs-w8}  \input{cliffordpluscsderivations/w6.tikz}
        = D\left(\input{cliffordpluscsaxioms/ipow3.tikz}\right)
    \end{gather*}
\end{proof}

\begin{proof}[Proof of the decoding of \cref{cs-C5}]
    \begin{gather*}
        D\left(\input{cliffordpluscsaxioms/CSpow4.tikz}\right)
        = \input{cliffordpluscsaxioms/CSpow4.tikz}
        \eqeqref{lem:cs-C5-dec-eq} \input{identities/I2.tikz}
        = D\left(\input{identities/I2.tikz}\right)
    \end{gather*}
\end{proof}

\begin{proof}[Proof of the decoding of \cref{cs-C6}]
    \begin{gather*}
        D\left(\input{cliffordpluscsaxioms/ScommTopLeft.tikz}\right)
        = \input{cliffordpluscsaxioms/ScommTopLeft.tikz}
        \eqeqref{lem:cs-C6-dec-eq} \input{cliffordpluscsaxioms/ScommTopRight.tikz}
        = D\left(\input{cliffordpluscsaxioms/ScommTopRight.tikz}\right)
    \end{gather*}
\end{proof}

\begin{proof}[Proof of the decoding of \cref{cs-C7}]
    \begin{gather*}
        D\left(\input{cliffordpluscsaxioms/ScommBotLeft.tikz}\right)
        = \input{cliffordpluscsaxioms/ScommBotLeft.tikz}
        \eqeqref{lem:cs-C7-dec-eq} \input{cliffordpluscsaxioms/ScommBotRight.tikz}
        = D\left(\input{cliffordpluscsaxioms/ScommBotRight.tikz}\right)
    \end{gather*}
\end{proof}

\begin{proof}[Proof of the decoding of \cref{cs-C8}]
    \begin{gather*}
        D\left(\input{cliffordpluscsaxioms/XTopLeft.tikz}\right)
        = \input{cliffordpluscsderivations/C8-00b.tikz}\\
        \eqeqref{lem:cs-C8-dec-eq} \input{cliffordpluscsderivations/C8-03b.tikz}
        = D\left(\input{cliffordpluscsaxioms/XTopRight.tikz}\right)
    \end{gather*}
\end{proof}

\begin{proof}[Proof of the decoding of \cref{cs-C9}]
    \begin{gather*}
        D\left(\input{cliffordpluscsaxioms/XBotLeft.tikz}\right)
        = \input{cliffordpluscsderivations/C9-00b.tikz}\\
        \eqeqref{lem:cs-C9-dec-eq} \input{cliffordpluscsderivations/C9-04b.tikz}
        = D\left(\input{cliffordpluscsaxioms/XBotRight.tikz}\right)
    \end{gather*}
\end{proof}

\begin{proof}[Proof of the decoding of \cref{cs-C10}]
    \begin{gather*}
        D\left(\input{cliffordpluscsaxioms/KTopLeft.tikz}\right)
        = \input{cliffordpluscsderivations/C10-00.tikz}\\
        \eqeqref{cs-SHCHC} \input{cliffordpluscsderivations/C10-01.tikz}
        = D\left(\input{cliffordpluscsaxioms/KTopRight.tikz}\right)
    \end{gather*}
\end{proof}

\begin{proof}[Proof of the decoding of \cref{cs-C11}]
    \begin{gather*}
        D\left(\input{cliffordpluscsaxioms/KBotLeft.tikz}\right)
        = \input{cliffordpluscsderivations/C11-00.tikz}\\
        \eqeqref{lem:cs-C11-dec-eq} \input{cliffordpluscsderivations/C11-03.tikz}
        = D\left(\input{cliffordpluscsaxioms/KBotRight.tikz}\right)
    \end{gather*}
\end{proof}

\begin{proof}[Proof of the decoding of \cref{cs-C12}]
    \begin{gather*}
        D\left(\input{cliffordpluscsaxioms/ICSCSI.tikz}\right)
        = \input{cliffordpluscsaxioms/ICSCSI.tikz}
        \eqeqref{lem:cs-C12-dec-eq} \input{cliffordpluscsaxioms/CSIICS.tikz}
        = D\left(\input{cliffordpluscsaxioms/CSIICS.tikz}\right)
    \end{gather*}
\end{proof}

\begin{proof}[Proof of the decoding of \cref{cs-C13}]
    \begin{gather*}
        D\left(\input{cliffordpluscsaxioms/C13lred.tikz}\right)
        = \input{cliffordpluscsaxioms/C13lb.tikz}\\
        \eqeqref{lem:cs-C13-dec-eq} \input{cliffordpluscsaxioms/C13rb.tikz}
        = D\left(\input{cliffordpluscsaxioms/C13rred.tikz}\right)
    \end{gather*}
\end{proof}

\begin{proof}[Proof of the decoding of \cref{cs-C14-old}]
    \begin{gather*}
        D\left(\input{cliffordpluscsaxioms/C14lred.tikz}\right)
        = \input{cliffordpluscsaxioms/C14lb.tikz}\\
        \eqeqref{lem:cs-C14-dec-eq} \input{cliffordpluscsaxioms/C14rb.tikz}
        = D\left(\input{cliffordpluscsaxioms/C14rred.tikz}\right)
    \end{gather*}
\end{proof}

\begin{proof}[Proof of the decoding of \cref{cs-C15}]
    \begin{gather*}
        D\left(\input{cliffordpluscsaxioms/C15lred.tikz}\right)
        = \input{cliffordpluscsaxioms/C15lb.tikz}\\
        \eqeqref{lem:cs-C15-dec-eq} \input{cliffordpluscsaxioms/C15rb.tikz}
        = D\left(\input{cliffordpluscsaxioms/C15rred.tikz}\right)
    \end{gather*}
\end{proof}

\begin{proof}[Proof of the decoding of \cref{cs-C16-old}]
    \begin{gather*}
        D\left(\input{cliffordpluscsaxioms/C16l.tikz}\right)
        = \input{cliffordpluscsderivations/C16-00b.tikz}\\
        \eqeqref{cs-U} \input{cliffordpluscsderivations/C16-finalb.tikz}
        = D\left(\input{cliffordpluscsaxioms/C16r.tikz}\right)
    \end{gather*}
\end{proof}

\begin{proof}[Proof of the decoding of \cref{cs-C17-old}]
    \begin{gather*}
        D\left(\input{cliffordpluscsaxioms/C17l.tikz}\right)
        = \input{cliffordpluscsderivations/C17-00b.tikz}\\
        \eqeqref{lem:cs-C17-dec-eq} \input{cliffordpluscsderivations/C17-finalb.tikz}
        = D\left(\input{cliffordpluscsaxioms/C17r.tikz}\right)
    \end{gather*}
\end{proof}
\section{CNOT-dihedral circuits}\label{app:CNOT-dihedral}

For CNOT-dihedral completeness, the fragment-specific work is narrower than in the other cases.
No encoding/decoding transfer is needed: the source presentation of
\ref{appendix:oldgraphicallanguage5} already lives in a PROP setting with the same scalar
convention as our master signature.

The remaining check is conservative coverage of the source presentation.  Every axiom of
\ref{appendix:oldgraphicallanguage5} omitted from the smaller rule set $\QCdihedral$ must be
derivable from $\QCdihedral$ itself.  The lemmas in this section provide those derivations, so the
imported completeness theorem applies unchanged.

\begin{remark}\label{lem:CNOTdihedral-R4}
\cref{cnot-old-R4} follows by the same rewriting used in
the section for Clifford+T.
\end{remark}

\begin{derivation}\label{lem:CNOTdihedral-R2}
  \cref{cnot-old-R2}
  \begin{gather*}
    \input{identities/CXC_00.tikz}
    \eqeqref{new-R3}\input{identities/CXC_01.tikz}
    \eqeqref{cnot-R1}\input{identities/CXC_02.tikz}\\
    \eqeqref{new-R3}\input{identities/CXC_03.tikz}
    \eqeqref{cnot-R1}\input{identities/CXC_04.tikz}
    \eqeqref{lem:CNOTdihedral-R4}\input{identities/CXC_05.tikz}
    \eqeqref{cnot-R1}\input{identities/CXC_06.tikz}
  \end{gather*}
\end{derivation}

\begin{derivation}\label{lem:CNOTdihedral-src-R3}
  \cref{cnot-old-R3}
  \begin{gather*}
    \input{cliffordczaxioms/CNOTXCNOT.tikz}
    \eqeqref{new-R3}\input{cliffordczaxioms/CNOTCNOTXX.tikz}
    \eqeqref{lem:CNOTdihedral-R4}\input{cliffordczaxioms/XX.tikz}
  \end{gather*}
\end{derivation}

\begin{derivation}\label{lem:CNOTdihedral-src-R5}
  \cref{cnot-old-R5}
  \begin{gather*}
    \input{cliffordczaxioms/CNOTswCNOTswCNOT.tikz}
    \eqeqref{new-R5}\input{cliffordczaxioms/CnCnNcSwCn.tikz}
    \eqeqref{lem:CNOTdihedral-R4}\input{cliffordczaxioms/NcSwCn.tikz}\\
     = \input{cliffordczaxioms/NcNcSw.tikz}
    \eqeqref{lem:CNOTdihedral-R4} \input{gates/SWAP.tikz}
  \end{gather*}
\end{derivation}

\begin{derivation}\label{lem:CNOTdihedral-src-R8}
  \cref{cnot-old-R8}
  \begin{gather*}
    \input{cliffordczaxioms/R8-old.tikz}
    \eqeqref{lem:CNOTdihedral-R4}\input{cliffordczaxioms/CNOTZCNOT.tikz}\\
    \eqeqref{cnot-R7}\input{cliffordczaxioms/ZZCNOTZCNOT.tikz}
    \eqeqref{new-R8} \input{cliffordczaxioms/CNOTCNOTZZ.tikz}
    \eqeqref{lem:CNOTdihedral-R4} \input{cliffordczaxioms/ZZ.tikz}
  \end{gather*}
\end{derivation}

\begin{derivation}\label{lem:CNOTdihedral-src-R11}
  \cref{cnot-old-R11}
  \begin{gather*}
    \input{cnotdihedral/R11oldL.tikz}
    \eqeqref{cnot-R11}\input{cliffordczaxioms/R11oldmid.tikz}
    \eqeqref{cnot-R1} \input{cnotdihedral/R11oldR.tikz}
  \end{gather*}
\end{derivation}

\section{Scalar refinement}\label{app:scalar-refinement}

The scalar-refinement step in \cref{subsec:scalar-refinement} addresses a strict-semantics issue:
global phases are part of the data of the fragment.  Thus two syntactic presentations of the
\emph{same} gate fragment can denote different semantic subPROPs when they generate different
subgroups of global phases.  Before importing a completeness theorem stated with another scalar
convention, we align the scalar subgroups by conservatively extending the syntax with an additional
scalar generator.

The main technical claim proved here is that, under mild assumptions satisfied by all our
unitary fragments, adjoining a new root-of-unity scalar and relating it to an existing one is
a \emph{conservative extension}: it does not introduce any new equalities between circuits
that do not mention the new scalar.  Equivalently, the induced interpretation remains faithful.


\begin{definition}\label{def:scalar-in-prop}
Let $\cat{P}$ be a PROP. Its \emph{scalars} are the endomorphisms of the tensor unit,
\(
  S(\cat{P}) \coloneqq \cat{P}(0,0).
\)
Since $0$ is the monoidal unit, both $\circ$ and $\otimes$ restrict to binary operations
\(
  S(\cat{P})\times S(\cat{P})\to S(\cat{P})
\),
and they coincide by interchange.  In particular, $S(\cat{P})$ is a commutative monoid
(and a group whenever $\cat{P}$ is a groupoid).
\end{definition}

\begin{definition}\label{def:global-phase}
Let $ \cat{C}$ be a (strict) symmetric monoidal subcategory of $\cat{FdHilb}$ and assume
$ \cat{C}$ is \emph{endomorphism-only}:
$ \cat{C}(n,m)=\varnothing$ for $n\neq m$.  For $n\in\N$, a unitary
$U\in \cat{C}(n,n)$ is a \emph{global phase on $n$ wires} if it is a scalar multiple of
the identity, namely \(U=\lambda \id_n\) for some
\(\lambda\in\mathrm{U}(1)=\{z\in\C\mid |z|=1\}\).
In words: $U$ acts trivially on all states and only contributes a phase factor $\lambda$.
\end{definition}

\begin{definition}\label{def:visible-scalars}
Let $ \cat{C}\subseteq \cat{FdHilb}$ be a symmetric monoidal subcategory.
The \emph{visible scalar subgroup} of $ \cat{C}$ is
\(
  S( \cat{C}) \coloneqq  \cat{C}(0,0)\subseteq\mathrm{U}(1).
\)
Every element $\lambda\in S( \cat{C})$ yields a global phase on $n$ wires by tensoring:
$\lambda\otimes \id_n = \lambda \id_n\in  \cat{C}(n,n)$.
\end{definition}

\begin{definition}\label{def:hidden-phase}
Let $ \cat{C}\subseteq \cat{FdHilb}$ be an endomorphism-only symmetric monoidal category.
A phase $\lambda\in\mathrm{U}(1)$ is called a \emph{hidden phase of $ \cat{C}$} if there
exists some arity $n\ge 1$ such that $\lambda \id_n\in \cat{C}(n,n)$ while
$\lambda\notin S( \cat{C})$.

Equivalently, $\lambda$ is hidden if it occurs as a global phase at some positive wire count,
but cannot be produced as a scalar morphism $0\to 0$ inside~$ \cat{C}$.
\end{definition}

\begin{example}\label{ex:hidden-phase}
Consider the subgroup $ \cat{C}\subseteq \cat{Qubit}$ generated (as a symmetric monoidal
subcategory) by a single 1-qubit gate $R_X(2\pi)$, but \emph{with no scalar generators}.
Semantically, $R_X(2\pi)=-\id_1$, hence $-\id_1\in \cat{C}(1,1)$.  However,
$S( \cat{C})=\{1\}$ because there are no nontrivial $0\to 0$ morphisms available.
Thus $\lambda=-1$ is a hidden phase in the sense of \cref{def:hidden-phase}.

If we were to adjoin a new scalar generator $\omega$ interpreted as $-1$, then $\omega\otimes\id_1$
and the existing circuit $R_X(2\pi)$ would have the same semantics; without additional axioms
linking them, this would typically destroy faithfulness of the interpretation.  The ``no hidden phases''
assumption in \cref{lem:scalar-refinement} precisely rules out this situation.
\end{example}


The conservativity argument relies on the observation that scalars commute with all endomorphisms,
so new scalar generators can always be ``pulled out'' as a single global phase factor.

\begin{lemma}\label{lem:scalar-centrality}
Let $ \cat{C}$ be a strict symmetric monoidal category.  For any scalar $a:0\to 0$ and any
endomorphism $f:n\to n$ one has
\begin{equation}\label{eq:scalar-centrality-app}
  (a\otimes \id_n)\circ f = f\circ (a\otimes \id_n).
\end{equation}
For scalars $a,b:0\to 0$, composition and tensor coincide:
\(
  a\circ b = a\otimes b.
\)
\end{lemma}
\begin{proof}
Since the tensor is bifunctorial and the category is strict,
\((a\otimes \id_n)\circ f
= (a\circ \id_0)\otimes(\id_n\circ f)
= (\id_0\circ a)\otimes(f\circ \id_n)
= f\circ(a\otimes \id_n)\), using the interchange law.
For scalars, the same computation with $n=0$ yields $a \circ b = a \otimes b$.
\end{proof}


The scalar-refinement step is an extension of a presentation by adjoining a new scalar generator of
higher order.

\begin{definition}\label{def:mu_m}
For $m\ge 1$ write
\(
  \mu_m \coloneqq \{e^{2\pi i k/m}\mid 0\le k<m\}\subseteq \mathrm{U}(1)
\)
for the cyclic group of $m$-th roots of unity.
\end{definition}

\begin{definition}\label{def:scalar-refined-presentation}
Let $\cat{P}^{\mathrm{src}}/\mathcal{R}^{\mathrm{src}}$ be a presented PROP and let
$s:0\to 0$ be a chosen scalar in the \emph{src} signature.
Fix integers $m\ge 1$ and $\ell\ge 1$.
Define $\cat{P}^{\mathrm{src},\sharp}/\mathcal{R}^{\mathrm{src},\sharp}$ to be the
presentation obtained by adjoining a new scalar generator
\(
  \omega:0\to 0
\)
and adding the two relations
\begin{equation}\label{eq:refinement-relations}
  \omega^{m\ell}=\id_0,
  \qquad
  \omega^{r}=s,
\end{equation}
where $\omega^k$ denotes the $k$-fold composite of $\omega$ (equivalently, $k$-fold tensor),
and $r$ is an integer chosen so that semantically $\omega^r$ matches $s$ as in
\eqref{eq:choose-r}.
\end{definition}

\begin{remark}\label{rem:refinement-intuition}
The relations \eqref{eq:refinement-relations} assert that $\omega$ has order dividing $m\ell$,
and that the \emph{src} scalar $s$ is the $r$-th power of $\omega$.  Concretely, this is the
standard ``adjoin a root'' construction at the level of scalars.  The point of
\cref{lem:scalar-refinement} is that, in a PROP setting, adjoining such a root is
conservative provided there are no hidden phases.
\end{remark}


The centrality of scalars yields a simple normal form in the refined syntax: all occurrences of
$\omega$ can be extracted into a single scalar factor.

\begin{lemma}\label{lem:scalar-extraction}
In the refined PROP $\cat{P}^{\mathrm{src},\sharp}/\mathcal{R}^{\mathrm{src},\sharp}$,
every morphism $C:n\to n$ is derivably equal to one of the form
\begin{equation}\label{eq:scalar-normal-form}
  C = \omega^{a}\otimes C',
\end{equation}
where $a\in\{0,1,\dots,m\ell-1\}$ and $C'$ is a morphism built using only the
\emph{src} generators (i.e.\ $C'$ contains no occurrence of $\omega$).
\end{lemma}
\begin{proof}
Work in the free PROP on the refined signature, modulo PROP coherence and the relations
$\mathcal{R}^{\mathrm{src}}$ and \eqref{eq:refinement-relations}.  Because $\omega$ is a scalar,
every occurrence of $\omega$ in a diagram is a $0\to 0$ component tensored somewhere with an
endomorphism.  By \cref{lem:scalar-centrality}, we can slide each such component past all other
boxes and past all compositions/tensors so as to collect all $\omega$-factors together as a single
power $\omega^{a}$ tensored with a diagram $C'$ that contains no $\omega$.
Finally, use $\omega^{m\ell}=\id_0$ to reduce $a$ modulo $m\ell$ into the stated range.
\end{proof}


The main lemma used in \cref{subsec:scalar-refinement} is the resulting conservativity statement.

\begin{lemma}\label{lem:scalar-refinement}
Let $\cat{P}^{\mathrm{src}}/\mathcal{R}^{\mathrm{src}}$ be a graphical language equipped with a
\emph{faithful} (i.e.\ complete) strict symmetric monoidal interpretation
\(
  \interp{\cdot}_{\mathrm{src}} :
  \cat{P}^{\mathrm{src}}/\mathcal{R}^{\mathrm{src}}
  \longrightarrow
   \cat{C}^{\mathrm{src}}
  \subseteq \cat{FdHilb}.
\)
Assume the following semantic properties.
\begin{enumerate}
\item\label{it:scalar-cyclic-app}
(finite cyclic visible scalars)
The scalar group $S( \cat{C}^{\mathrm{src}})= \cat{C}^{\mathrm{src}}(0,0)$ is finite cyclic of order $m$
and generated by $\interp{s}_{\mathrm{src}}$ for some chosen src scalar $s:0\to 0$.
Equivalently, $S( \cat{C}^{\mathrm{src}})\cong \mu_m$.

\item\label{it:no-hidden-phases-app}
(no hidden phases)
$ \cat{C}^{\mathrm{src}}$ has no hidden phases: for all $n\in\N$ and all $\lambda\in\mathrm{U}(1)$,
if $\lambda \id_n\in \cat{C}^{\mathrm{src}}(n,n)$ then $\lambda\in S( \cat{C}^{\mathrm{src}})$
(cf.\ \cref{def:hidden-phase}).

\item\label{it:invertible-app}
(invertibility)
Every morphism in $ \cat{C}^{\mathrm{src}}$ is invertible, i.e.\ each hom-set
$ \cat{C}^{\mathrm{src}}(n,n)$ is a group under composition.\footnote{%
In all applications in this paper, $ \cat{C}^{\mathrm{src}}$ consists of unitaries, hence this holds automatically.}
\end{enumerate}

Fix $\ell\ge 1$ and choose a primitive root of unity $\zeta\in\mathrm{U}(1)$ of order $m\ell$.
Choose an integer $r$ such that
\begin{equation}\label{eq:choose-r}
  \zeta^{r}=\interp{s}_{\mathrm{src}}.
\end{equation}
Form the scalar-refined presentation
$\cat{P}^{\mathrm{src},\sharp}/\mathcal{R}^{\mathrm{src},\sharp}$ as in
\cref{def:scalar-refined-presentation}.

Let $ \cat{C}^{\sharp}\subseteq\cat{FdHilb}$ be the subPROP generated by $ \cat{C}^{\mathrm{src}}$
together with the scalar $\zeta:0\to 0$ (viewed as the unitary map $\C\to\C$, $1\mapsto\zeta$).

Then the interpretation functor extends uniquely to a strict symmetric monoidal functor
\(
  \interp{\cdot}^{\sharp} :
  \cat{P}^{\mathrm{src},\sharp}/\mathcal{R}^{\mathrm{src},\sharp}
  \longrightarrow
   \cat{C}^{\sharp}
\)
sending $\omega$ to $\zeta$, and $\interp{\cdot}^{\sharp}$ is faithful.

In particular, the refined language is complete for $ \cat{C}^{\sharp}$ and
$S( \cat{C}^{\sharp})\cong \mu_{m\ell}$.
\end{lemma}

\begin{proof}
(1) Soundness and existence of the refined interpretation.
By construction, $\interp{\cdot}^{\sharp}$ agrees with $\interp{\cdot}_{\mathrm{src}}$ on all src generators.
We interpret the new generator $\omega$ by the scalar $\zeta$.
The additional relations \eqref{eq:refinement-relations} hold semantically because
$\zeta^{m\ell}=1$ and $\zeta^{r}=\interp{s}_{\mathrm{src}}$ by \eqref{eq:choose-r}.
Therefore the universal property of presented PROPs yields a unique induced functor
$\interp{\cdot}^{\sharp}$ into $ \cat{C}^{\sharp}$.

(2) Scalar extraction in the refined syntax.
By \cref{lem:scalar-extraction}, every refined circuit $C:n\to n$ is derivably equal to
$\omega^{a}\otimes C'$ with $0\le a<m\ell$ and $C'$ an src circuit.

(3) Faithfulness.
Let $C_1,C_2:n\to n$ be refined circuits such that
$\interp{C_1}^{\sharp}=\interp{C_2}^{\sharp}$.
Write both in scalar-extraction normal form:
\(
  C_i = \omega^{a_i}\otimes C_i'
  \qquad (0\le a_i<m\ell),
\)
with $C_i'$ containing no $\omega$.
Let $U_i\coloneqq \interp{C_i'}_{\mathrm{src}}\in \cat{C}^{\mathrm{src}}(n,n)$.
The semantic equality becomes
\begin{equation}\label{eq:semantic-equality-expanded}
  \zeta^{a_1}   U_1 = \zeta^{a_2}  U_2.
\end{equation}
By invertibility \ref{it:invertible-app}, we can right-multiply by $U_2^{-1}$ to obtain
\begin{equation}\label{eq:phase-in-src-app}
  \zeta^{a_2-a_1} \id_n = U_1\circ U_2^{-1}
  \in  \cat{C}^{\mathrm{src}}(n,n).
\end{equation}
Thus $\zeta^{a_2-a_1}$ appears as a global phase in $ \cat{C}^{\mathrm{src}}$ on $n$ wires.
By the no-hidden-phases assumption \ref{it:no-hidden-phases-app}, this implies
\(
  \zeta^{a_2-a_1}\in S( \cat{C}^{\mathrm{src}}).
\)
By \ref{it:scalar-cyclic-app}, there exists $t\in\{0,\dots,m-1\}$ such that
\begin{equation}\label{eq:phase-as-src-scalar}
  \zeta^{a_2-a_1} = \interp{s}_{\mathrm{src}}^{ t}.
\end{equation}
Combining \eqref{eq:semantic-equality-expanded} with \eqref{eq:phase-as-src-scalar} gives
\(
  \interp{C_1'}_{\mathrm{src}}
  =
  \interp{s^{t}\otimes C_2'}_{\mathrm{src}}.
\)
Since $\interp{\cdot}_{\mathrm{src}}$ is faithful by hypothesis,
\(
  \cat{P}^{\mathrm{src}}/\mathcal{R}^{\mathrm{src}}
  \vdash
  C_1' = s^{t}\otimes C_2'.
\)
The refined presentation contains all src relations, so the same derivation is valid in
$\cat{P}^{\mathrm{src},\sharp}/\mathcal{R}^{\mathrm{src},\sharp}$.
Now use the refinement relation $\omega^{r}=s$ to rewrite $s^{t}=\omega^{rt}$, yielding
\(
  C_1
  =
  \omega^{a_1}\otimes C_1'
  =
  \omega^{a_1}\otimes \omega^{rt}\otimes C_2'
  =
  \omega^{a_1+rt}\otimes C_2'.
\)

It remains to compare the exponents.  From \eqref{eq:phase-as-src-scalar} and \eqref{eq:choose-r} we have
$\zeta^{a_2-a_1}=\zeta^{rt}$, so $\zeta^{a_2-a_1-rt}=1$.  Since $\zeta$ has exact order $m\ell$, this implies
\(
  a_2-a_1-rt \equiv 0 \pmod{m\ell}.
\)
Using $\omega^{m\ell}=\id_0$ in the refined syntax, we conclude $\omega^{a_1+rt}=\omega^{a_2}$, hence
\(
  C_1=\omega^{a_2}\otimes C_2' = C_2.
\)
Therefore $\interp{\cdot}^{\sharp}$ is faithful.

(4) Scalars of $ \cat{C}^{\sharp}$.
By definition, $S( \cat{C}^{\sharp})$ is generated by $S( \cat{C}^{\mathrm{src}})=\langle \interp{s}_{\mathrm{src}}\rangle$
together with $\zeta$.  Since $\interp{s}_{\mathrm{src}}=\zeta^r\in\langle \zeta\rangle$, we obtain
$S( \cat{C}^{\sharp})=\langle \zeta\rangle\cong \mu_{m\ell}$.
\end{proof}

\begin{corollary}\label{cor:scalar-refinement-conservative}
With the hypotheses of \cref{lem:scalar-refinement}, let $C_1,C_2$ be \emph{src} circuits
(i.e.\ containing no occurrence of $\omega$).  If
\(
  \cat{P}^{\mathrm{src},\sharp}/\mathcal{R}^{\mathrm{src},\sharp}\vdash C_1=C_2,
\)
then already
\(
  \cat{P}^{\mathrm{src}}/\mathcal{R}^{\mathrm{src}}\vdash C_1=C_2.
\)
\end{corollary}
\begin{proof}
If the refined theory derives $C_1=C_2$, then in particular $\interp{C_1}_{\mathrm{src}}=\interp{C_2}_{\mathrm{src}}$
(because the refined interpretation restricts to the src one on src syntax).  Faithfulness of
$\interp{\cdot}_{\mathrm{src}}$ implies the src derivability.
\end{proof}


The two scalar refinements applied in the main text are the following.

\begin{itemize}
\item \textbf{Qutrit Clifford.}
The imported presentation of \cite{qutritclifford} uses an order-$6$ scalar
$\input{gates/minw.tikz}$.  Our strict qutrit convention uses an order-$12$ scalar subgroup
(\cref{sec:circuits-relations}).  We refine the source syntax by adjoining an order-$12$ scalar
generator $\input{gates/wb.tikz}$ and adding the relations
\(\input{gates/wb.tikz}^{12}=\id_0\) and
\(\input{gates/wb.tikz}^{10}=\input{gates/minw.tikz}\).
This is the $\mu_6\to\mu_{12}$ refinement referenced in \cref{subsec:scalar-refinement}.

\item \textbf{Clifford+CS.}
The imported presentation of \cite{SelingerCliffordPlusCS} uses an order-$4$ scalar
$\input{gates/i.tikz}$, whereas our strict qubit convention uses an order-$8$ scalar subgroup.
We refine the source syntax by adjoining an order-$8$ scalar generator $\input{gates/wprime.tikz}$
and adding the relations \(\input{gates/w.tikz}^{8}=\id_0\) and
\(\input{gates/w.tikz}^{2}=\input{gates/i.tikz}\).
This is the $\mu_4\to\mu_8$ refinement referenced in \cref{subsec:scalar-refinement}.
\end{itemize}

In both cases, the hypotheses of \cref{lem:scalar-refinement} are satisfied for the unitary
semantics considered in this paper: invertibility is automatic, visible scalars are finite cyclic by
construction, and the imported normal forms expose every scalar multiple of an identity as one of
those visible scalars, so no hidden phases appear relative to the chosen scalar conventions.
Therefore, scalar refinement preserves faithfulness, and the refined imported presentations may be
treated as complete presentations for the strict semantics used throughout.


\begin{figure}[ht]
    \begin{threeparttable}
    \begin{adjustbox}{width=\textwidth,center}
    \renewcommand{\arraystretch}{1.2}
    \begin{tabular}{|c|c|c|c|c|c|c|c|}
    \hline
    \textbf{Interpretation} & \textbf{Cliff} & \textbf{Real Cliff} & \textbf{Qutrit Cliff} & \textbf{Cliff+T} & \textbf{Cliff+CS} & \textbf{CNOT-dih} &  \textbf{Arity} \\
    \hline
    $?\gw$ & \eqref{w8} & — & — & \eqref{t-w8}  & \eqref{cs-w8} & \eqref{cnot-R10} & 0\\
    $?\gminus$ & — & \eqref{realminus2} & — & —  & — & — & 0\\
    $?\gww$ & — & — & \eqref{qt-omegapow12} & —  & — & — & 0\\
    \hline
    $?\gH$ & \eqref{H2} & \eqref{realH2} & \eqref{qt-hpow4} & \eqref{t-H2}  & \eqref{cs-H2} & — & 1 \\
    $?\gS$ & \eqref{S4} & — & \eqref{qt-spow3} & — & \eqref{cs-S4} & — & 1 \\
    $?\gZ$ & — & \eqref{realZ2} & — & — & — & — & 1 \\
    $?\gT$ & — & — & — & \eqref{t-Tpow8} & — & \eqref{cnot-R7} & 1 \\
    $?\gX$ & — & — & — & — & — & \eqref{cnot-R1} & 1 \\
    $\#\{\gH\}_{[2]}$ & \eqref{E} & — & \eqref{qt-shpow3} & \eqref{t-E} & \eqref{cs-E} & — & 1 \\
    $\#\{\gw\}_{[2]}$ & — & — & — & — & — & \eqref{cnot-R11} & 1 \\
    $\#\{\gminus\}_{[2]}$ & — & \eqref{realF} & — & — & — & — & 1 \\
    $\#\{\gH, \gw\}_{[2]}$ & — & — & — & \eqref{t-TX} & — & — & 1 \\
    $\#\{\gX\}_{[2]}$ & — & — & — & — & — & \eqref{new-R3} & 1 \\
    $[\gS:=\gSX]_{\sim}$ & — & — & \eqref{qt-ssprime}\tnote{a} & — & — & — & 1 \\
    \hline
    $?\gCNOT$ & \eqref{Cs} & \eqref{realCX2} & \eqref{qt-cnotremove} & \eqref{t-Ct} & — & \eqref{cnot-R12} & 2 \\
    $?\gCS$ & — & — & — & — & \eqref{cs-C} & — & 2 \\
    $\#\{\gSWAP\}_{[2]}$ & \eqref{B} & \eqref{realB} & \eqref{qt-swapdecomp} & \eqref{t-B} & \eqref{cs-B} & \eqref{new-R5} & 2 \\
    $\#\{\gCNOT,\gSWAP\}_{[2]}$ & \eqref{CZ} & —  & — & \eqref{t-CZ} & — & — & 2 \\
    $\#\{\gS\}_{[3]}$ & — & — & \eqref{qt-czdecomp} & — & — & — & 2 \\
    $\#\{\gZ\}_{[2]}$ & — & \eqref{realwC} & — & — & — & — & 2 \\
    $[\gZ:=\gI]_{\sim}$ & — & \eqref{realCF} & — & — & — & — & 2 \\
    $[\gH:=\gI]_{\sim}$ & — & \eqref{realXC} & — & — & — & — & 2 \\
    $\#\{\gCNOT\}_{[2]}$ & — & — & \eqref{qt-kcnot} & — & — & — & 2 \\
    $\#\{\gS,\gH\}_{[2]}$ & — & — & — & — & \eqref{cs-XCS} & — & 2 \\
    $[\gCS:=\gCSz]_{\sim}$ & — & — & — & — & \eqref{cs-CSrev} & — & 2 \\
    $[\gCS:=\gCSzz]_{\sim}$ & — & — & — & — & \eqref{cs-SHCHC} & — & 2 \\
    $\#\{\gT,\gw,\gw\}_{[8]}$ & — & — & — & — & — & \eqref{new-R8} & 2 \\
    \hline
    $\arg\det_{2}$ & \eqref{Inew} & \eqref{realInew} & — & — & \eqref{cs-I} & — & 3 \\
    $\#\{\gCNOT,\gSWAP\}_{[2]}$ & — & —  & — & — & — & \eqref{cnot-R6} & 3 \\
    $\#\{\gT,\gw,\gw\}_{[4]}$ & — & — & — & — & — & \eqref{cnot-R9} & 3 \\
    \hline
    $\arg\det_{3}$ & — & — & — & — & — & \eqref{cnot-R13} & 4 \\
    \hline
    \end{tabular}
    \end{adjustbox}
    \begin{tablenotes}
    \footnotesize
    \item[a] See \cref{qt-prop:SSpnecessity}.
    \end{tablenotes}
    \end{threeparttable}
  \caption{Summary of minimality interpretations across all fragments}
  \label{fig:minimality-summary}
\end{figure}

\begin{figure}[ht]
    \fbox{\begin{minipage}{.975\textwidth}\begin{center}
        \vspace{-1em}
        \begin{subfigure}{0.21\textwidth}
            \begin{align}\label{C1}\tag{C1}\input{cliffordaxioms/w.tikz}=\input{gates/empty.tikz}\end{align}
        \end{subfigure}
        \begin{subfigure}{0.29\textwidth}
            \begin{align}\label{C2}\tag{C2}\input{cliffordaxioms/HH.tikz}=\input{gates/Id.tikz}\end{align}
        \end{subfigure}
        \begin{subfigure}{0.37\textwidth}
            \begin{align}\label{C3}\tag{C3}\input{cliffordaxioms/SSSS.tikz}=\input{gates/Id.tikz}\end{align}
        \end{subfigure}\vspace{-0.5em}
        \begin{subfigure}{0.52\textwidth}
            \begin{align}\label{C4}\tag{C4}\input{cliffordczaxioms/C4-00.tikz}=\input{cliffordczaxioms/C4-03.tikz}\end{align}
        \end{subfigure}
        \begin{subfigure}{0.28\textwidth}
            \begin{align}\label{C5}\tag{C5}\input{cliffordczaxioms/C5-00.tikz}=\input{identities/I2.tikz}\end{align}
        \end{subfigure}\vspace{-0.5em}
        \begin{subfigure}{0.32\textwidth}
            \begin{align}\label{C6}\tag{C6}\input{cliffordczaxioms/C6-00.tikz}=\input{cliffordczaxioms/C6-03.tikz}\end{align}
        \end{subfigure}
        \begin{subfigure}{0.32\textwidth}
            \begin{align}\label{C7}\tag{C7}\input{cliffordczaxioms/C7-00.tikz}=\input{cliffordczaxioms/C7-05.tikz}\end{align}
        \end{subfigure}\vspace{-0.5em}
        \begin{subfigure}{0.61\textwidth}
            \begin{align}\label{C8}\tag{C8}\input{cliffordczaxioms/C8-00.tikz}=\input{cliffordczaxioms/C8-08.tikz}\end{align}
        \end{subfigure}\vspace{-0.5em}
        \begin{subfigure}{0.61\textwidth}
            \begin{align}\label{C9}\tag{C9}\input{cliffordczaxioms/C9-00.tikz}=\input{cliffordczaxioms/C9-09.tikz}\end{align}
        \end{subfigure}\vspace{-0.5em}
        \begin{subfigure}{0.61\textwidth}
            \begin{align}\label{C10}\tag{C10}\input{cliffordczaxioms/C10-00.tikz}=\input{cliffordczaxioms/C10-14.tikz}\end{align}
        \end{subfigure}\vspace{-0.5em}
        \begin{subfigure}{0.61\textwidth}
            \begin{align}\label{C11}\tag{C11}\input{cliffordczaxioms/C11-00.tikz}=\input{cliffordczaxioms/C11-01.tikz}\end{align}
        \end{subfigure}
        \begin{subfigure}{0.35\textwidth}
            \begin{align}\label{C12}\tag{C12}\input{cliffordczaxioms/C12-00.tikz}=\input{cliffordczaxioms/C12-06.tikz}\end{align}
        \end{subfigure}
        \begin{subfigure}{.85\textwidth}
            \begin{align}\label{C13}\tag{C13}\input{cliffordczaxioms/C13-00.tikz}=\input{cliffordczaxioms/C13-25.tikz}\end{align}
        \end{subfigure}
        \begin{subfigure}{.82\textwidth}
            \begin{align}\label{C14}\tag{C14}\input{cliffordczaxioms/C14-00.tikz}=\input{cliffordczaxioms/C14-09.tikz}\end{align}
        \end{subfigure}
        \begin{subfigure}{.82\textwidth}
            \begin{align}\label{C15}\tag{C15}\input{cliffordczaxioms/C15-00.tikz}=\input{cliffordczaxioms/C15-01.tikz}\end{align}
        \end{subfigure}
        \begin{subfigure}{.44\textwidth}
            \begin{align}\label{DSWAPdef}\tag{S}\input{gates/SWAP.tikz}=\input{cliffordczaxioms/DSWAPdef-02.tikz}\end{align}
        \end{subfigure}
    \end{center}\end{minipage}}
    \caption{Set of relations $\QCcliffordCZ$. Equations $(\textup{C}_i)$ are the ones of the complete set of relations for the pro considered in \cite{SelingerStabilizer}.}
    \label{appendix:oldgraphicallanguage1}
\end{figure}

\begin{figure}[ht]
    \fbox{\begin{minipage}{.975\textwidth}\begin{center}
        \vspace{-1em}
        \begin{subfigure}{0.23\textwidth}
            \begin{align}\label{R1}\tag{R1}\input{realcliffordaxioms/2old.tikz} = \input{gates/empty.tikz}\end{align}
        \end{subfigure}
        \begin{subfigure}{0.28\textwidth}
            \begin{align}\label{R2}\tag{R2}\input{realcliffordaxioms/Z2.tikz} = \input{gates/Id.tikz}\end{align}
        \end{subfigure}
        \begin{subfigure}{0.28\textwidth}
            \begin{align}\label{R3}\tag{R3}\input{realcliffordaxioms/H2.tikz} = \input{gates/Id.tikz}\end{align}
        \end{subfigure}\vspace{-0.5em}
        \begin{subfigure}{0.63\textwidth}
            \begin{align}\label{R4}\tag{R4}\input{realcliffordczaxioms/R4-00.tikz}=\input{realcliffordczaxioms/R4-02.tikz}\end{align}
        \end{subfigure}
        \begin{subfigure}{0.27\textwidth}
            \begin{align}\label{R5}\tag{R5}\input{cliffordczaxioms/C5-00.tikz}=\input{identities/I2.tikz}\end{align}
        \end{subfigure}\vspace{-0.5em}
        \begin{subfigure}{0.34\textwidth}
            \begin{align}\label{R6}\tag{R6}\input{realcliffordczaxioms/R6-00.tikz}=\input{realcliffordczaxioms/R6-07.tikz}\end{align}
        \end{subfigure}
        \begin{subfigure}{0.34\textwidth}
            \begin{align}\label{R7}\tag{R7}\input{realcliffordczaxioms/R7-00.tikz}=\input{realcliffordczaxioms/R7-01.tikz}\end{align}
        \end{subfigure}\vspace{-0.5em}
        \begin{subfigure}{0.52\textwidth}
            \begin{align}\label{R8}\tag{R8}\input{realcliffordczaxioms/R8-00.tikz}=\input{realcliffordczaxioms/R8-06.tikz}\end{align}
        \end{subfigure}\vspace{-0.5em}
        \begin{subfigure}{0.52\textwidth}
            \begin{align}\label{R9}\tag{R9}\input{realcliffordczaxioms/R9-00.tikz}=\input{realcliffordczaxioms/R9-01.tikz}\end{align}
        \end{subfigure}\vspace{-0.5em}
        \begin{subfigure}{0.5\textwidth}
            \begin{align}\label{R10}\tag{R10}\input{realcliffordczaxioms/R10-00.tikz}=\input{realcliffordczaxioms/R10-10.tikz}\end{align}
        \end{subfigure}\vspace{-0.5em}
        \begin{subfigure}{0.5\textwidth}
            \begin{align}\label{R11}\tag{R11}\input{realcliffordczaxioms/R11-00.tikz}=\input{realcliffordczaxioms/R11-01.tikz}\end{align}
        \end{subfigure}\vspace{-0.5em}
        \begin{subfigure}{0.64\textwidth}
            \begin{align}\label{R12}\tag{R12}\input{realcliffordczaxioms/R12-00.tikz}=\input{realcliffordczaxioms/R12-12.tikz}\end{align}
        \end{subfigure}
        \begin{subfigure}{0.34\textwidth}
            \begin{align}\label{R13}\tag{R13}\input{realcliffordczaxioms/R13-00.tikz}=\input{realcliffordczaxioms/R13-99.tikz}\end{align}
        \end{subfigure}\vspace{-0.5em}
        \begin{subfigure}{1\textwidth}
            \begin{align}\label{R14}\tag{R14}\input{realcliffordczaxioms/R14-00.tikz}=\input{realcliffordczaxioms/R14-99.tikz}\end{align}
        \end{subfigure}\vspace{-0.5em}
        \begin{subfigure}{.76\textwidth}
            \begin{align}\label{R15}\tag{R15}\input{realcliffordczaxioms/R15-00.tikz}=\input{realcliffordczaxioms/R15-99.tikz}\end{align}
        \end{subfigure}\vspace{-0.5em}
        \begin{subfigure}{.76\textwidth}
            \begin{align}\label{R16}\tag{R16}\input{realcliffordczaxioms/R16-00.tikz}=\input{realcliffordczaxioms/R16-99.tikz}\end{align}
        \end{subfigure}\vspace{-0.5em}
        \begin{subfigure}{.46\textwidth}
            \begin{align}\label{real-DSWAPdef}\tag{S}\input{gates/SWAP.tikz}=\input{cliffordczaxioms/DSWAPdef-04.tikz}\end{align}
        \end{subfigure}
    \end{center}\end{minipage}}
    \caption{Set of relations $\QCrealcliffordCZ$. Equations $(\textup{R}_i)$ are the ones of the complete set of relations for the pro considered in \cite{SelingerRealStabilizer}.}
    \label{appendix:oldgraphicallanguage2}
\end{figure}

\begin{figure}[ht]
    \fbox{\begin{minipage}{.975\textwidth}\begin{center}
        \vspace{-1em}
        \begin{subfigure}{0.26\textwidth}
            \begin{align}\label{qt-C1}\tag{C1}
                \input{qutritcliffordaxioms/omegapow6.tikz}
                =
                \input{gates/empty.tikz}
            \end{align}
        \end{subfigure}
        \begin{subfigure}{0.36\textwidth}
            \begin{align}\label{qt-C2}\tag{C2}
                \input{qutritcliffordaxioms/hadpow4.tikz}
                =
                \input{gates/Id.tikz}
            \end{align}
        \end{subfigure}
        \begin{subfigure}{0.32\textwidth}
            \begin{align}\label{qt-C3}\tag{C3}
                \input{qutritcliffordaxioms/spow3.tikz}
                =
                \input{gates/Id.tikz}
            \end{align}
        \end{subfigure}\vspace{-0.5em}
        \begin{subfigure}{0.48\textwidth}
            \begin{align}\label{qt-C4}\tag{C4}
                \input{qutritcliffordaxioms/s2hs2hs2h.tikz}
                =
                \input{qutritcliffordaxioms/Iphase-omega.tikz}
            \end{align}
        \end{subfigure}
        \begin{subfigure}{0.49\textwidth}
            \begin{align}\label{qt-C5}\tag{C5}
                \input{qutritcliffordaxioms/ssprime.tikz}
                =
                \input{qutritcliffordaxioms/sprimes.tikz}
            \end{align}
        \end{subfigure}\vspace{-0.5em}
        \begin{subfigure}{0.36\textwidth}
            \begin{align}\label{qt-C6}\tag{C6}
                \input{qutritderivations/CZ-CZ-CZ.tikz} = \input{identities/I2.tikz}
            \end{align}
        \end{subfigure}
        \begin{subfigure}{0.44\textwidth}
            \begin{align}\label{qt-SWAPSWAP}\tag{C7}
                \input{qutritderivations/SWAP-SWAP.tikz} = \input{identities/I2.tikz}
            \end{align}
        \end{subfigure}\vspace{-0.5em}
        \begin{subfigure}{0.36\textwidth}
            \begin{align}\label{qt-C7}\tag{C8}
                \input{qutritderivations/S-CZ.tikz} = \input{qutritderivations/CZ-S.tikz}
            \end{align}
        \end{subfigure}
        \begin{subfigure}{0.41\textwidth}
            \begin{align}\label{qt-C8}\tag{C9}
                \input{qutritderivations/K-CZ.tikz}
                =
                \input{qutritderivations/CZ-CZ-K.tikz}
            \end{align}
        \end{subfigure}\vspace{-0.5em}
        \begin{subfigure}{0.47\textwidth}
            \begin{align}\label{qt-SSWAP}\tag{C10}
                \input{qutritderivations/S-SWAP.tikz} = \input{qutritderivations/SWAP-S.tikz}
            \end{align}
        \end{subfigure}
        \begin{subfigure}{0.47\textwidth}
            \begin{align}\label{qt-C9}\tag{C11}
                \input{qutritderivations/H-SWAP.tikz} = \input{qutritderivations/SWAP-H.tikz}
            \end{align}
        \end{subfigure}\vspace{-0.5em}
        \begin{subfigure}{0.47\textwidth}
            \begin{align}\label{qt-CZSWAP}\tag{C12}
                \input{qutritderivations/CZ-SWAP.tikz} = \input{qutritderivations/SWAP-CZz.tikz}
            \end{align}
        \end{subfigure}
        \begin{subfigure}{0.5\textwidth}
            \begin{align}\label{qt-C10}\tag{C13}
                \input{qutritderivations/S-NOTC-origin.tikz} = \input{qutritderivations/NOTC-CZ-SSp.tikz}
            \end{align}
        \end{subfigure}
        \begin{subfigure}{0.54\textwidth}
            \begin{align}\label{qt-C11}\tag{C14}
                \input{qutritderivations/CZ-CNOT-b.tikz} = \input{qutritderivations/NOTC-CZ-SSp-b.tikz}
            \end{align}
        \end{subfigure}
        \begin{subfigure}{0.38\textwidth}
            \begin{align}\label{qt-C13}\tag{C15}
                \input{qutritderivations/ICZ-CZI.tikz} = \input{qutritderivations/CZI-ICZ.tikz}
            \end{align}
        \end{subfigure}
        \begin{subfigure}{0.67\textwidth}
            \begin{align}\label{qt-C14}\tag{C16}
                \input{qutritderivations/ISW-SWI-ISW.tikz} = \input{qutritderivations/SWI-ISW-SWI.tikz}
            \end{align}
        \end{subfigure}
        \begin{subfigure}{0.63\textwidth}
            \begin{align}\label{qt-C15}\tag{C17}
                \input{qutritderivations/ICZ-SWI-ISW.tikz} = \input{qutritderivations/SWI-ISW-CZI.tikz}
            \end{align}
        \end{subfigure}
        \begin{subfigure}{0.41\textwidth}
            \begin{align}\label{qt-C16}\tag{C18}
                \input{qutritderivations/CNI-ICZ.tikz} = \input{qutritderivations/ICZ-CNI-CIZ.tikz}
            \end{align}
        \end{subfigure}
        \begin{subfigure}{0.38\textwidth}
            \begin{align}\label{qt-Swap}\tag{S}
                \input{qutritderivations/SWAP-CZ.tikz} = \input{gates/SWAP.tikz}
            \end{align}
        \end{subfigure}
        
    \end{center}\end{minipage}}
    \caption{Set of relations $\QCqutritcliffordCZ$. Equations $(\textup{C}_i)$ are the ones of the complete set of relations for the pro considered in \cite{qutritclifford}.}
    \label{appendix:oldgraphicallanguage3}
\end{figure}

\begin{figure}[ht]
    \fbox{\begin{minipage}{.975\textwidth}\begin{center}
        \vspace{-1em}
        \begin{subfigure}{0.21\textwidth}
            \begin{align}\label{t-C3}\tag{C3}\input{cliffordaxioms/w.tikz}=\input{gates/empty.tikz}\end{align}
        \end{subfigure}
        \begin{subfigure}{0.29\textwidth}
            \begin{align}\label{t-C4}\tag{C4}\input{cliffordaxioms/HH.tikz}=\input{gates/Id.tikz}\end{align}
        \end{subfigure}
        \begin{subfigure}{0.37\textwidth}
            \begin{align}\label{t-C5}\tag{C5}\input{cliffordaxioms/SSSS.tikz}=\input{gates/Id.tikz}\end{align}
        \end{subfigure}\vspace{-0.5em}
        \begin{subfigure}{0.52\textwidth}
            \begin{align}\label{t-C6}\tag{C6}\input{cliffordczaxioms/C4-00.tikz}=\input{cliffordczaxioms/C4-03.tikz}\end{align}
        \end{subfigure}
        \begin{subfigure}{0.27\textwidth}
            \begin{align}\label{t-C14}\tag{C14}\input{shortcut/S.tikz}=\input{shortcut/TT.tikz}\end{align}
        \end{subfigure}\vspace{-0.5em}
        \begin{subfigure}{0.43\textwidth}
            \begin{align}\label{t-C15}\tag{C15}\input{cliffordplustaxioms/TX.tikz} = \input{cliffordplustaxioms/XTd.tikz}\end{align}
        \end{subfigure}
        \begin{subfigure}{0.28\textwidth}
            \begin{align}\label{t-C7}\tag{C7}\input{cliffordczaxioms/C5-00.tikz}=\input{identities/I2.tikz}\end{align}
        \end{subfigure}\vspace{-0.5em}
        \begin{subfigure}{0.32\textwidth}
            \begin{align}\label{t-C8}\tag{C8}\input{cliffordczaxioms/C6-00.tikz}=\input{cliffordczaxioms/C6-03.tikz}\end{align}
        \end{subfigure}
        \begin{subfigure}{0.32\textwidth}
            \begin{align}\label{t-C9}\tag{C9}\input{cliffordczaxioms/C7-00.tikz}=\input{cliffordczaxioms/C7-05.tikz}\end{align}
        \end{subfigure}\vspace{-0.5em}
        \begin{subfigure}{0.62\textwidth}
            \begin{align}\label{t-C10}\tag{C10}\input{cliffordczaxioms/C8-00.tikz}=\input{cliffordczaxioms/C8-08.tikz}\end{align}
        \end{subfigure}
        \begin{subfigure}{0.63\textwidth}
            \begin{align}\label{t-C11}\tag{C11}\input{cliffordczaxioms/C9-00.tikz}=\input{cliffordczaxioms/C9-09.tikz}\end{align}
        \end{subfigure}
        \begin{subfigure}{0.62\textwidth}
            \begin{align}\label{t-C12}\tag{C12}\input{cliffordczaxioms/C10-00.tikz}=\input{cliffordczaxioms/C10-14.tikz}\end{align}
        \end{subfigure}
        \begin{subfigure}{0.62\textwidth}
            \begin{align}\label{t-C13}\tag{C13}\input{cliffordczaxioms/C11-00.tikz}=\input{cliffordczaxioms/C11-01.tikz}\end{align}
        \end{subfigure}
        \begin{subfigure}{0.39\textwidth}
            \begin{align}\label{t-C16}\tag{C16}\input{cliffordplustaxioms/C16l.tikz}=\input{cliffordplustaxioms/C16r.tikz}\end{align}
        \end{subfigure}
        \begin{subfigure}{0.69\textwidth}
            \begin{align}\label{t-C17}\tag{C17}\input{cliffordplustaxioms/C17l.tikz}=\input{cliffordplustaxioms/C17r.tikz}\end{align}
        \end{subfigure}
        \begin{subfigure}{0.84\textwidth}
            \begin{align}\label{t-C18}\tag{C18}\input{cliffordplustaxioms/C18l.tikz}=\input{cliffordplustaxioms/C18r.tikz}\end{align}
        \end{subfigure}
        \begin{subfigure}{0.95\textwidth}
            \begin{align}\label{old-t-C19}\tag{C19}\scalebox{0.8}{\input{cliffordplustaxioms/C19l.tikz}}=\scalebox{0.8}{\input{cliffordplustaxioms/C19r.tikz}}\end{align}
        \end{subfigure}
        \begin{subfigure}{0.69\textwidth}
            \begin{align}\label{old-t-C20}\tag{C20}\input{cliffordplustaxioms/C20l.tikz}=\input{cliffordplustaxioms/C20r.tikz}\end{align}
        \end{subfigure}
        \begin{subfigure}{.44\textwidth}
            \begin{align}\label{t-DSWAPdef}\tag{S}\input{gates/SWAP.tikz}=\input{cliffordczaxioms/DSWAPdef-02.tikz}\end{align}
        \end{subfigure}
    \end{center}\end{minipage}}
    \caption{Set of relations $\QCcliffordplustCZ$ on two qubits. Equations $(\textup{C}_i)$ are the ones of the complete set of relations for the pro considered in \cite{SelingerCliffordPlusT}.}
    \label{appendix:oldgraphicallanguage4}
\end{figure}

\begin{figure}[ht]
    \scalebox{1}{
    \fbox{\begin{minipage}{.975\textwidth}\begin{center}
        \vspace{-1em}
        \begin{subfigure}{0.2\textwidth}
            \begin{align}\label{cs-C1}\tag{C1}\input{qutritcliffordaxioms/ipow4.tikz} = \input{gates/empty.tikz}\end{align}
        \end{subfigure}
        \begin{subfigure}{0.34\textwidth}
            \begin{align}\label{cs-C2}\tag{C2}\input{cliffordpluscsaxioms/K2.tikz}=\input{cliffordpluscsaxioms/ipow3.tikz}\end{align}
        \end{subfigure}
        \begin{subfigure}{0.4\textwidth}
            \begin{align}\label{cs-C3}\tag{C3}\input{cliffordaxioms/SSSS.tikz}=\input{cliffordaxioms/I.tikz}\end{align}
        \end{subfigure}
        \begin{subfigure}{0.52\textwidth}
            \begin{align}\label{cs-C4}\tag{C4}\input{cliffordpluscsaxioms/SKSKSK.tikz}=\input{cliffordpluscsaxioms/ipow3.tikz}\end{align}
        \end{subfigure}
        \begin{subfigure}{0.38\textwidth}
            \begin{align}\label{cs-C5}\tag{C5}\input{cliffordpluscsaxioms/CSpow4.tikz}=\input{identities/I2.tikz}\end{align}
        \end{subfigure}
        \begin{subfigure}{0.35\textwidth}
            \begin{align}\label{cs-C6}\tag{C6}\input{cliffordpluscsaxioms/ScommTopLeft.tikz}=\input{cliffordpluscsaxioms/ScommTopRight.tikz}\end{align}
        \end{subfigure}
        \begin{subfigure}{0.38\textwidth}
            \begin{align}\label{cs-C7}\tag{C7}\input{cliffordpluscsaxioms/ScommBotLeft.tikz}=\input{cliffordpluscsaxioms/ScommBotRight.tikz}\end{align}
        \end{subfigure}
        \begin{subfigure}{0.57\textwidth}
            \begin{align}\label{cs-C8}\tag{C8}\input{cliffordpluscsaxioms/XTopLeft.tikz}=\input{cliffordpluscsaxioms/XTopRight.tikz}\end{align}
        \end{subfigure}
        \begin{subfigure}{0.57\textwidth}
            \begin{align}\label{cs-C9}\tag{C9}\input{cliffordpluscsaxioms/XBotLeft.tikz}=\input{cliffordpluscsaxioms/XBotRight.tikz}\end{align}
        \end{subfigure}
        \begin{subfigure}{0.59\textwidth}
            \begin{align}\label{cs-C10}\tag{C10}\input{cliffordpluscsaxioms/KTopLeft.tikz}=\input{cliffordpluscsaxioms/KTopRight.tikz}\end{align}
        \end{subfigure}
        \begin{subfigure}{0.59\textwidth}
            \begin{align}\label{cs-C11}\tag{C11}\input{cliffordpluscsaxioms/KBotLeft.tikz}=\input{cliffordpluscsaxioms/KBotRight.tikz}\end{align}
        \end{subfigure}
        \begin{subfigure}{0.37\textwidth}
            \begin{align}\label{cs-C12}\tag{C12}\input{cliffordpluscsaxioms/ICSCSI.tikz}=\input{cliffordpluscsaxioms/CSIICS.tikz}\end{align}
        \end{subfigure}
        \begin{subfigure}{0.59\textwidth}
            \begin{align}\label{cs-C13}\tag{C13}\input{cliffordpluscsaxioms/C13lred.tikz}=\input{cliffordpluscsaxioms/C13rred.tikz}\end{align}
        \end{subfigure}
        \begin{subfigure}{0.67\textwidth}
            \begin{align}\label{cs-C14-old}\tag{C14}\input{cliffordpluscsaxioms/C14lred.tikz}=\input{cliffordpluscsaxioms/C14rred.tikz}\end{align}
        \end{subfigure}
        \begin{subfigure}{0.68\textwidth}
            \begin{align}\label{cs-C15}\tag{C15}\input{cliffordpluscsaxioms/C15lred.tikz}=\input{cliffordpluscsaxioms/C15rred.tikz}\end{align}
        \end{subfigure}
        \begin{subfigure}{0.75\textwidth}
            \begin{align}\label{cs-C16-old}\tag{C16}\input{cliffordpluscsaxioms/C16l.tikz}=\input{cliffordpluscsaxioms/C16r.tikz}\end{align}
        \end{subfigure}
        \begin{subfigure}{0.92\textwidth}
            \begin{align}\label{cs-C17-old}\tag{C17}\input{cliffordpluscsaxioms/C17l.tikz}=\input{cliffordpluscsaxioms/C17r.tikz}\end{align}
        \end{subfigure}
        \begin{subfigure}{0.62\textwidth}
            \begin{align}\tag{B}\label{cs-DSWAP}\input{gates/SWAP.tikz}=\input{cliffordpluscsczaxioms/DSWAP-final.tikz}\end{align}
        \end{subfigure}
    \end{center}\end{minipage}}
    }
    \caption{Set of relations $\QCcliffordpluscsCZ$ on three qubits. Equations $(\textup{C}_i)$ are the ones of the complete set of relations for the pro considered in \cite{SelingerCliffordPlusCS}.}
    \label{appendix:oldgraphicallanguage5}
\end{figure}

\begin{figure}[ht]
    \centering
    \fbox{%
        \begin{minipage}{0.98\textwidth}
        \centering
        \begin{minipage}[c]{0.31\textwidth}
            \centering
            \begin{equation}\tag{R1}
            \scalebox{0.9}{\input{cnotdihedral/R1L.tikz}}
            =
            \scalebox{0.9}{\input{gates/Id.tikz}}
            \end{equation}
        \end{minipage}
        \begin{minipage}[c]{0.36\textwidth}
            \centering
            \begin{equation}\tag{R2}\label{cnot-old-R2}
            \scalebox{0.9}{\input{identities/CXC_00.tikz}}
            =
            \scalebox{0.9}{\input{identities/CXC_06.tikz}}
            \end{equation}
        \end{minipage}
        \begin{minipage}[c]{0.31\textwidth}
            \centering
            \begin{equation}\tag{R3}\label{cnot-old-R3}
            \scalebox{0.9}{\input{cliffordczaxioms/CNOTXCNOT.tikz}}
            =
            \scalebox{0.9}{\input{cliffordczaxioms/XX.tikz}}
            \end{equation}
        \end{minipage}
        \begin{minipage}[c]{0.29\textwidth}
            \centering
            \begin{equation}\tag{R4}\label{cnot-old-R4}
            \scalebox{0.9}{\input{identities/CNOTCNOT.tikz}}
            =
            \scalebox{0.9}{\input{identities/I2.tikz}}
            \end{equation}
        \end{minipage}
        \begin{minipage}[c]{0.47\textwidth}
            \centering
            \begin{equation}\tag{R5}\label{cnot-old-R5}
            \scalebox{0.9}{\input{cliffordczaxioms/CNOTswCNOTswCNOT.tikz}}
            =
            \scalebox{0.9}{\input{gates/SWAP.tikz}}
            \end{equation}
        \end{minipage}
        \begin{minipage}[c]{0.36\textwidth}
            \centering
            \begin{equation}\tag{R6}
            \scalebox{0.9}{\input{cnotdihedral/R6L.tikz}}
            =
            \scalebox{0.9}{\input{cnotdihedral/R6R.tikz}}
            \end{equation}
        \end{minipage}
        \begin{minipage}[c]{0.28\textwidth}
            \centering
            \begin{equation}\tag{R7}
            \scalebox{0.9}{\input{cnotdihedral/R7L.tikz}}
            =
            \scalebox{0.9}{\input{gates/Id.tikz}}
            \end{equation}
        \end{minipage}
        \begin{minipage}[c]{0.36\textwidth}
            \centering
            \begin{equation}\tag{R8}\label{cnot-old-R8}
            \scalebox{0.9}{\input{cliffordczaxioms/R8-old-decomp.tikz}}
            =
            \scalebox{0.9}{\input{cliffordczaxioms/ZZ.tikz}}
            \end{equation}
        \end{minipage}
        \begin{minipage}[c]{0.53\textwidth}
            \centering
            \begin{equation}\tag{R9}
            \scalebox{0.9}{\input{cnotdihedral/R9L.tikz}}
            =
            \scalebox{0.9}{\input{cnotdihedral/R9R.tikz}}
            \end{equation}
        \end{minipage}
        \begin{minipage}[c]{0.22\textwidth}
            \centering
            \begin{equation}\tag{R10}
                \scalebox{0.9}{\input{gates/wc.tikz}}^{\otimes 8} = \scalebox{0.9}{\input{gates/empty.tikz}}
            \end{equation}
        \end{minipage}
        \begin{minipage}[c]{0.37\textwidth}
            \centering
            \begin{equation}\tag{R11}\label{cnot-old-R11}
            \scalebox{0.9}{\input{cnotdihedral/R11oldL.tikz}}
            =
            \scalebox{0.9}{\input{cnotdihedral/R11oldR.tikz}}
            \end{equation}
        \end{minipage}
        \begin{minipage}[c]{0.33\textwidth}
            \centering
            \begin{equation}\tag{R12}
            \scalebox{0.9}{\input{cliffordplustaxioms/CNOTTCNOT.tikz}}
            =
            \scalebox{0.9}{\input{cliffordplustaxioms/T.tikz}}
            \end{equation}
        \end{minipage}
        \begin{minipage}[c]{0.9\textwidth}
            \centering
            \begin{equation}\tag{R13}
            \scalebox{0.82}{\input{cnotdihedral/R13L.tikz}}
            =
            \scalebox{0.82}{\input{cnotdihedral/R13R.tikz}}
            \end{equation}
        \end{minipage}
        \end{minipage}%
    }
    \caption{Set of relations $\QColddihedral$ \cite{AmyCNOTDihedral}.}
    \label{appendix:oldgraphicallanguage6}
\end{figure}

\begin{table}[t]
\centering
\small
\renewcommand{\arraystretch}{1.15}

\newcommand{\spacedhline}{\noalign{\vskip 0.7ex}\hline\noalign{\vskip 0.7ex}}

\begin{tabular}{@{}l p{0.35\linewidth} p{0.35\linewidth}@{}}
\hline
\noalign{\vskip 0.7ex}

\textbf{Fragment} &
\textbf{Encoding $E_\bullet$ (new $\to$ src)} &
\textbf{Decoding $D_\bullet$ (src $\to$ new)} \\
\spacedhline

$\mathit{Cliff}, \mathit{RCliff}$ &
\(
\begin{aligned}
E_\bullet\left(\input{gates/CNOT.tikz}\right) &= \input{completeness/H2CZH2.tikz}
\end{aligned}
\) &
\(
\begin{aligned}
D_\bullet\left(\input{gates/CZ.tikz}\right) &= \input{completeness/H2CNOTH2.tikz}
\end{aligned}
\) \\
\spacedhline

$\mathit{CliffT}$ &
\(
\begin{aligned}
E_{\mathit{CliffT}}\left(\input{gates/CNOT.tikz}\right) &= \input{completeness/H2CZH2.tikz}
\end{aligned}
\) &
\(
\begin{aligned}
D_{\mathit{CliffT}}\left(\input{gates/CZ.tikz}\right) &= \input{completeness/H2CNOTH2.tikz}\\
D_{\mathit{CliffT}}\left(\input{gates/S.tikz}\right) &= \input{shortcut/TT.tikz}
\end{aligned}
\) \\
\spacedhline

$\mathit{Cliff3}$ &
\(
\begin{aligned}
E_{\mathit{Cliff3}}\left(\input{gates/S.tikz}\right) &= \input{gates/Sbis.tikz}\\
E_{\mathit{Cliff3}}\left(\input{gates/H.tikz}\right) &= \input{gates/Hbisp.tikz}\\
E_{\mathit{Cliff3}}\left(\input{gates/CNOT.tikz}\right) &= \input{completeness/H2CZHHH2.tikz}\\
E_{\mathit{Cliff3}}\left(\input{gates/wb.tikz}\right) &= \input{gates/wb.tikz}
\end{aligned}
\) &
\(
\begin{aligned}
D_{\mathit{Cliff3}}\left(\input{gates/S.tikz}\right) &= \input{gates/Sbis.tikz}\\
D_{\mathit{Cliff3}}\left(\input{gates/H.tikz}\right) &= \input{gates/Hbis.tikz}\\
D_{\mathit{Cliff3}}\left(\input{gates/CZ.tikz}\right) &= \input{completeness/HHH2CNOTH2.tikz}\\
D_{\mathit{Cliff3}}\left(\input{gates/wb.tikz}\right) &= \input{gates/wb.tikz}\\
D_{\mathit{Cliff3}}\left(\input{gates/minw.tikz}\right) &= \input{gates/w10.tikz}
\end{aligned}
\) \\
\spacedhline

$\mathit{CliffCS}$ &
\(
\begin{aligned}
E_{\mathit{CliffCS}}\left(\input{gates/H.tikz}\right) &= \input{gates/Kw.tikz}\\
E_{\mathit{CliffCS}}\left(\input{gates/w.tikz}\right) &= \input{gates/w.tikz}
\end{aligned}
\) &
\(
\begin{aligned}
D_{\mathit{CliffCS}}\left(\input{gates/K.tikz}\right) &= \input{gates/Hw.tikz}\\
D_{\mathit{CliffCS}}\left(\input{gates/i.tikz}\right) &= \input{gates/w2.tikz}\\
D_{\mathit{CliffCS}}\left(\input{gates/w.tikz}\right) &= \input{gates/w.tikz}
\end{aligned}
\) \\
\noalign{\vskip 0.7ex}
\hline

\end{tabular}

\caption{Encoding/decoding pairs used in the completeness-transfer arguments.
For the qutrit Clifford and Clifford+CS rows, the source presentation is the
scalar-refined one from \cref{lem:scalar-refinement}.
\emph{Convention:} omitted generators are fixed by $E_\bullet$ and $D_\bullet$.}
\label{tab:ED-clifford-family}
\end{table}

\end{document}